\documentclass[a4paper, 11pt]{article}

\usepackage{standalone, lmodern, xcolor, enumitem}
\usepackage{amsmath, amssymb, amsthm, physics, dsfont, mathtools, stmaryrd, leftidx}
\usepackage{thmtools}
\usepackage{hyperref}

\usepackage{lmodern}

\usepackage[normalem]{ulem}
\usepackage{accents}

\usepackage{pgfplots}
\pgfplotsset{compat=1.16}

\usetikzlibrary{arrows, decorations.pathreplacing, cd, patterns, calc, decorations.markings, backgrounds, tikzmark}

\usepackage{quiver}

\usepackage[new]{old-arrows}
\usepackage{changepage}

\usepackage[bottom]{footmisc}

\theoremstyle{plain}
\newtheorem{theorem}{Theorem}[section]
\newtheorem{proposition}[theorem]{Proposition}
\newtheorem{lemma}[theorem]{Lemma}
\newtheorem{corollary}[theorem]{Corollary}

\numberwithin{equation}{section}
\theoremstyle{definition}
\newtheorem{definition}[theorem]{Definition}
\newtheorem{example}[theorem]{Example}
\newtheorem{remark}[theorem]{Remark}

\numberwithin{figure}{section}

\setlist[enumerate]{topsep = 1ex, leftmargin=.7cm, itemsep= -2pt}

\allowdisplaybreaks

\definecolor{color1}{HTML}{1b9e77}
\definecolor{color2}{HTML}{F1A340}
\definecolor{color3}{HTML}{5e3c99}

\hypersetup{
  colorlinks = true, 
  urlcolor = color3, 
  linkcolor = color1, 
  citecolor = color3 
}

\newcommand{\R}{\mathbb{R}}
\newcommand{\Rp}{\mathbb{R}_{+}}
\newcommand{\Z}{\mathbb{Z}}

\newcommand{\C}{\mathbb{C}}
\newcommand{\vect}{\mathrm{Vect}}
\DeclareMathOperator{\id}{{\mathds{1}}}

\newcommand{\disk}{\mathbb{D}} 
\newcommand{\odisk}{\setsuchthat{z \in \C}{|z| < 1}} 
\newcommand{\cdisk}{\setsuchthat{z \in \C}{|z| \leq 1}} 
\newcommand{\cuhp}{\bar{\mathbb{H}}} 

\newcommand{\ann}{\mathbb{A}}

\newcommand{\blank}{\:\cdot\:}
\newcommand{\setsuchthat}[2]{\left\{ {#1} \:\middle|\: {#2} \right\}}
\newcommand{\setsuchthatinline}[2]{\{ {#1} \:|\: {#2} \}}
\newcommand{\restrict}[2]{{#1}\vert_{{#2}}}

\DeclareMathOperator{\Diff}{{Diff}}

\DeclareMathOperator{\Diffpan}{{Diff}_+^{\an}(S^1)}
\DeclareMathOperator{\DiffC}{{Def}_{\C}(S^1)}
\DeclareMathOperator{\DefC}{\DiffC}
\DeclareMathOperator{\VectR}{\vect^{\an}_{\R}(S^1)}
\DeclareMathOperator{\VectC}{\vect^{\an}_{\C}(S^1)}
\DeclareMathOperator{\Witt}{\VectC}
\newcommand{\DefCnbhd}[1]{U({#1})}

\DeclareMathOperator{\mob}{PSL(2, \C)}
\newcommand{\mobvect}{\mathfrak{sl}(2, \C)}
\newcommand{\genus}{\mathsf{g}}
\newcommand{\boundaries}{\mathsf{b}}
\newcommand{\moduli}[2]{\mathcal{M}_{{#1}, {#2}}}
\newcommand{\annuli}{\moduli{0}{2}}
\newcommand{\moduligb}{\moduli{\genus}{\boundaries}}

\newcommand{\disks}{\moduli{0}{1}}

\newcommand{\teichm}{\mathcal{T}}
\newcommand{\projhgb}{\mathsf{hyp}}
\newcommand{\hgbtofgbzi}{\mathsf{Q}}
\newcommand{\mcg}{\mathsf{MCG}}
\newcommand{\fmoduligbzi}{\check{\mathcal{M}}_{\genus, \boundaries, 0, \underline{1}}}
\newcommand{\fmoduligbnm}{\check{\mathcal{M}}_{\genus, \boundaries, n, \underline{m}}}
\newcommand{\hmoduli}[2]{\mathcal{M}^{\textsf{hyp}}_{{#1}, {#2}}}
\newcommand{\hmoduligb}{\hmoduli{\genus}{\boundaries}}

\newcommand{\an}{\mathrm{an}}
\newcommand{\charge}{\mathbf{c}}

\DeclareMathOperator{\Hom}{Hom}

\DeclareMathOperator{\kernel}{Ker}

\DeclareMathOperator{\image}{Im}

\newcommand{\sew}[2]{\:\leftidx{{_{#1}}}{{\infty}}{{_{#2}}}\:}

\newcommand{\sewself}[2]{\:{\infty}{{_{#1, #2}}}\:}

\newcommand{\sewall}{\:\underline{\infty}\:}

\newcommand{\unravelplus}[1]{{{#1}}_+\:}
\newcommand{\unravelminus}[1]{{{#1}}_-\:}
\newcommand{\unravel}[1]{{{#1}}\:}

\newcommand{\gelfandfuks}{\omega_\mathrm{GF}}
\newcommand{\bott}{\Omega_{\mathrm{BT}}}
\newcommand{\evaluation}{\mathsf{ev}}

\newcommand{\param}[1]{
\mathchoice
    {\big(\: {#1} \:\big)}
    {(#1)}
    {(#1)}
    {(#1)}
}
\newcommand{\parameq}[1]{
\mathchoice
    {\big[\: {#1} \:\big]}
    {[#1]}
    {[#1]}
    {[#1]}
}

\newcommand{\ii}{\mathrm{i}\,}

\makeatletter
\newcommand{\uset}[3][-0.35ex]{%
  \mathrel{\mathop{#3}\limits_{
    \vbox to#1{\kern-7\ex@
    \hbox{$\scriptscriptstyle#2$}\vss}}}}
\makeatother

\newcommand{\diffActingInline}[1]{\,
    \smash{\uset{{#1}}{*}} 
\, }
\newcommand{\diffActing}[1]{\diffActingInline{#1}}

\DeclareMathOperator{\definvset}{{Inv}_{\C}(S^1)}
\DeclareMathOperator{\defmultset}{{Comp}_{\C}(S^1)}

\newcommand{\transnumber}{\mathrm{T}}
\newcommand{\rotnumber}{\mathrm{R}}
\newcommand{\rotnumberAvg}{\mathrm{Rot}}
\newcommand{\rotnumberInt}{\mathrm{Trans}}
\newcommand{\CR}{\mathrm{CR}} 
\newcommand{\rotCR}{{\mathrm{RCR}}}
\newcommand{\rotnumberalg}{\mathrm{rot}}
\newcommand{\rotcocycle}{\Omega_{\mathrm{E}}}
\newcommand{\rotcocyclealg}{\omega_{\mathrm{E}}}
\newcommand{\shrink}[1]{\mathsf{sc}_{{#1}}}
\newcommand{\shrinkgroup}{\mathrm{Sc}}
\newcommand{\shrinkalg}{\R \ell_0}
\newcommand{\scaling}[1]{\shrink{{#1}}}
\newcommand{\scalinggroup}{\shrinkgroup}

\newcommand{\rotation}[1]{\mathsf{r}_{#1}}

\newcommand{\curves}[1]{\mathcal{C}({{#1}})}
\newcommand{\curvesgen}[1]{\mathcal{C}_0({{#1}})}
\newcommand{\functions}[1]{\mathcal{F}({{#1}})}

\newcommand{\compn}[2]{M_{{#2}}({#1})}

\newcommand{\compdef}[3]{\mathcal{U}_{#1, #2, #3}}

\newcommand{\sgrps}{\underline{H}}
\newcommand{\salgs}{\underline{\mathfrak{h}}}
\newcommand{\funtogrp}{\delta}

\DeclareMathOperator{\inversion}{J}

\DeclareMathOperator{\DVE}{d_{\mathsf{vE}}}
\DeclareMathOperator{\Dalg}{d_{\mathsf{alg}}}
\DeclareMathOperator{\Dgrp}{d_{\mathsf{grp}}}
\DeclareMathOperator{\DdR}{d}
\DeclareMathOperator{\Dder}{D}

\newcommand{\cdiffeo}[1]{\mathsf{diff}_{{#1}}}
\newcommand{\cdiffeolift}[1]{\hat{\mathsf{diff}}_{{#1}}}
\newcommand{\cunivalent}[1]{\mathsf{univ}_{{#1}}}
\newcommand{\cunivalentlift}[1]{\hat{\mathsf{univ}}_{{#1}}}

\newcommand{\Univ}{\mathsf{Univ}}

\newcommand{\universalcover}[1]{\mathsf{uc}({#1})}

\newcommand{\ellpar}[1]{a^\parallel_{{#1}}}
\newcommand{\ellpari}[1]{b^\parallel_{{#1}}}
\newcommand{\ellperpi}[1]{a^\perp_{{#1}}}
\newcommand{\ellperp}[1]{b^\perp_{{#1}}}

\newcommand{\thetacoords}[1]{\tilde{{#1}}}

\newcommand{\Cocycle}{\Omega}

\newcommand{\gcc}{{g_{\mathrm{cc}}}} 

\newcommand{\invinv}{\mathsf{t}}



\DeclareRobustCommand\longtwoheadrightarrow
 {\relbar\joinrel\twoheadrightarrow}

\title{Complex deformations of the circle:\protect\\ Group cohomology and Virasoro uniformization}  
\date{November 17, 2026}

\author{Sid Maibach\thanks{Department of Mathematics, ETH Z\"urich, R\"amistrasse 101, 8092 Z\"urich, Switzerland. \\ 
\protect\url{sid.maibach@math.ethz.ch}} \;
and \,
Eveliina Peltola\thanks{Department of Mathematics and Systems Analysis, Aalto University, Otakaari 1, 02150 Espoo, Finland, and Division of Mathematics, University of Cologne, Weyertal 86-90, 50931 Cologne, Germany. \\ 
\protect\url{eveliina.peltola@aalto.fi}}}

\begin{document}
\maketitle

\begin{center}
\begin{minipage}{0.98\textwidth}
\abstract{
We approach the question of complexification of the diffeomorphism group of the circle by considering real-analytic maps from the circle into the punctured complex plane with winding number +1.
Such ``complex deformations'' form an infinite-dimensional manifold with partially defined inversion and composition operations, smooth in the sense of Fr{\"o}licher structures, and with Lie algebra relations at the identity given by the Witt algebra.
We compute their second group cohomology with real coefficients.
It is generated by two cocycles: One cocycle extending the Bott--Thurston cocycle related to the Gelf'and--Fuks cocycle of the Virasoro algebra, and another Euler class cocycle combining mean rotation and conformal radius.

Complex deformations act naturally on the (infinite-dimensional) Segal moduli spaces of Riemann surfaces with enumerated and analytically parametrized boundary components, equipping them with smooth Fr{\"o}licher structures.
We prove a Virasoro uniformization theorem:
The tangent spaces of the Segal moduli spaces are spanned by vector fields induced by the Witt algebra.
Finally, we relate the actions of complex deformations to Fenchel--Nielsen coordinates and Schiffer variation on finite-dimensional moduli spaces of bordered hyperbolic surfaces with one marked point on each boundary component.
}

\bigskip{}


\bigskip{}

\end{minipage}
\end{center}

\newpage

\tableofcontents

\bigskip

\section{Introduction}

\emph{Does there exist a complexification of the diffeomorphism group of the circle?}

\noindent
This question has been answered negatively by Neretin~\cite[Section~1.10]{Neretin:Holomorphic_extensions_of_representations_of_the_group_of_diffeomorphisms_of_the_circle}, Pressley and Segal~\cite[Proposition~3.3.2]{Pressley-Segal:Loop_groups}, and most definitely by Lempert~\cite{Lempert:The_problem_of_complexifying_a_Lie_group}.
Yet, Neretin has also constructed objects that complexify certain aspects of the diffeomorphism group: namely, diffeomorphisms may be viewed as ``thin'' annuli at the boundary of the semigroup of annuli~\cite{Neretin:Complex_semigroup_that_contains_the_group_of_diffeomorphisms_of_the_circle, Neretin:Holomorphic_extensions_of_representations_of_the_group_of_diffeomorphisms_of_the_circle, Neretin:Categories_of_symmetries_and_infinite-dimensional_groups, Segal:Definition_of_CFT_collection, Radnell-Schippers:The_semigroup_of_rigged_annuli_and_the_Teichmueller_space_of_the_annulus, Henriques-Tener:The_Segal-Neretin_semigroup_of_annuli};
and they may be complexified as maps from the circle $S^1$ to $\C \setminus \{0\}$ with a partially defined composition operation, as described in~\cite[Section~1.1]{Neretin:Holomorphic_extensions_of_representations_of_the_group_of_diffeomorphisms_of_the_circle}.
In previous work~\cite{Maibach-Peltola:From_the_conformal_anomaly_to_the_Virasoro_algebra}, we re-introduced the latter as bounded deformations of the unit circle, or \emph{complex deformations} for short.
In the present work, we shall use a slightly generalized notion of complex deformations:
\begin{quote}
Injective, real-analytic maps $S^1 \to \C \setminus \{0\}$ with winding number $+1$ around $0$.
\end{quote}
This set of complex deformations, denoted $\DefC$, is an open submanifold of the topological vector space $\mathcal{O}(S^1)$ of germs of analytic functions on $S^1 \subset \C$.
Using Fr\"olicher structures~\cite{Kriegl-Michor:Convenient_setting_of_global_analysis, Laubinger:Lie_algebra_for_Frolicher_groups} to define a notion of a smooth structure on $\DefC$,
we verify that the tangent space $T_{\id} \DefC$ at the identity is the familiar Lie algebra $\VectC$ of real-analytic complex-valued vector fields on $S^1$, containing the Witt algebra whose generators $\{\ell_n \,|\, n \in \Z\}$ satisfy the Lie bracket $[\ell_n, \ell_m] = (n - m) \ell_{n + m}$.

\subsection*{Main results} 
\begin{enumerate}
\item \emph{Cohomology of complex deformations.}
We define and compute the (Fr\"olicher smooth) group cohomology $H^2(\DefC; \R)$, including relative cohomology with respect to scaling transformations and diffeomorphisms.
Theorem~\ref{thm:cohomology}, as well as Figures~\ref{fig:diag_grp}~\&~\ref{fig:diag_alg}, identify an extension of the Bott--Thurston cocycle~$\bott$, and the Euler class cocycle~$\rotcocycle$ represented by the mean rotation and conformal radius of a complex deformation.

\item \emph{Virasoro uniformization.}
On the (infinite-dimensional) Segal moduli space $\moduligb$ of compact connected genus $\genus$ Riemann surfaces with $\boundaries$ parametrized boundary components, the $\boundaries$ partially defined actions of $\DefC$ via boundary deformation/repara\-metrization (see Definition~\ref{def:action}) induce Lie algebra homomorphisms
\begin{equation}
	\rho_j \colon  \VectC \to \vect(\moduligb), \qquad 1 \leq j \leq \boundaries.
\end{equation}
In Theorem~\ref{thm:virasoro_uniformization}, we prove that the combined image of $\rho = \rho_1 + \dots + \rho_\boundaries$ spans the tangent space of $\moduligb$, and that the kernel of $\rho$ is given by taking holomorphic vector fields with real-analytic boundary behavior on a surface $\Sigma$ representing a point in $\moduligb$ and simultaneously pulling it back via the boundary parametrizations $\zeta_j \colon S^1 \to \partial_j \Sigma$: 
\begin{equation}\label{eq:virasoro_uniformization_ses}
	\{0\} \; \longrightarrow \; \vect(\Sigma) \; \overset{\underline{\zeta}^*}{\longhookrightarrow} \; \big(\VectC\big)^{\times \boundaries} \; \overset{\rho}{\longtwoheadrightarrow} \; T_{[\Sigma]}^{\mathcal{C}_0} \moduligb \; \longrightarrow \; \{0\} ,
\end{equation} 
is a short exact sequence involving the pullback by $\underline{\zeta} = (\zeta_1,\ldots,\zeta_\boundaries)$.
\end{enumerate}

The exact sequence~\eqref{eq:virasoro_uniformization_ses} is dubbed \emph{Virasoro uniformization}, and has been proven before in an algebro-geometric setup~\cite{Kontsevich:Virasoro_and_Teichmuller_spaces, ACKP:Moduli_spaces_of_curves_and_representation_theory, Beilinson-Schechtman:Determinant_bundles_and_Virasoro_algebra, Frenkel-Ben-Zvi:Vertex_Algebras_and_Algebraic_Curves, Dubedat:SLE_and_Virasoro_representations_localization, Gui-Zhang:Analytic_conformal_blocks_of_C2_cofinite_vertex_operator_algebras2}, where instead of boundary parametrizations, one considers formal coordinate neighborhoods at marked points on Riemann surfaces without boundary.
In fact, our proof in Section~\ref{section:moduligb_fr} follows the same elementary arguments, using Kodaira--Spencer deformation theory (see~\cite{Kodaira:Complex_manifolds_and_deformation_of_complex_structures}) and sheaf cohomology on Riemann surfaces (see, e.g.,~\cite{Forster:Lectures_on_Riemann_surfaces}).

Second group cohomology groups classify central extensions (cf.~\cite{Gelfand-Fuchs:Cohomologies_of_the_Lie_algebra_of_vector_fields_on_the_circle,Bott:On_the_characteristic_classes_of_groups_of_diffeomorphisms,Haefliger:Differentiable_cohomology,Ghys:Groups_acting_on_the_circle, Guieu-Roger:Virasoro_book}). 
The computations of the group cohomology of complex deformations in Section~\ref{section:cohomology} 
leading to Theorem~\ref{thm:cohomology} are also important for our program of formulating a universality property of the \emph{conformal anomaly} of conformal field theory (CFT),  
outlined in earlier work~\cite{Maibach-Peltola:From_the_conformal_anomaly_to_the_Virasoro_algebra} and completed in the companion article~\cite{Maibach-Peltola:Universality_of_the_conformal_anomaly}.
There, a $2$-cocycle $\Gamma_\charge$ on $\DefC$ is derived from the conformal anomaly through a construction known as the real determinant line bundle.
The cocycle depends on a choice of central charge $\charge \in \R$, and~\cite[Theorem~1.1]{Maibach-Peltola:From_the_conformal_anomaly_to_the_Virasoro_algebra} shows that the corresponding Lie algebra cocycle is $\charge$ times the imaginary part of the Gel'fand--Fuks cocycle.
In Section~\ref{section:cft}, we explain how (for $\charge \neq 0$) any such cocycle spans the relative cohomology group $H^2(\VectC; \VectR; \R)$.
Furthermore, by computing $H^2(\DefC; \Diffpan, \scalinggroup; \R)$ we can identify the cohomology class of $\Gamma_\charge$ also at the group level (see Proposition~\ref{prop:gamma_bott_rotcocycle}). 
This provides a crucial step in the characterization of the real determinant line bundles as universal local $\Rp$-extensions of the sewing operations on the moduli spaces~$\moduligb$~\cite{Maibach-Peltola:Universality_of_the_conformal_anomaly} (called real one-dimensional modular functors in~\cite[Corollary~1.5]{Maibach-Peltola:From_the_conformal_anomaly_to_the_Virasoro_algebra}).
Concretely, the cocycle $\Gamma_\charge$ and the $\Rp$-extension of $\moduligb$ are related through the actions of $\DefC$ on the moduli spaces studied in Section~\ref{section:surfaces} of the present article.

\subsection*{Segal moduli spaces}
We essentially follow Segal's definition of $\moduligb$ by choosing real-analytic regularity for the boundary parametrizations\footnote{
	The setup using real-analytic boundary parametrizations was also investigated by Huang~\cite{Huang:2D_Conformal_geometry_and_VOAs} for genus zero, and by Radnell~\cite{Radnell:PhD} for genus one.
	More generally, one might consider boundary parametrizations to be smooth~\cite{Henriques:The_functorial_approach_to_chiral_2D_CFT, Alekseev-Meinrenken:Symplectic_geometry_of_Teichmuller_spaces_for_surfaces_with_ideal_boundary}, or Weil--Petersson quasiconformal or quasiconformal~\cite{Takhtajan-Teo:Weil-Petersson_metric_on_the_universal_Teichmuller_space, Radnell-Schippers:Quasisymmetric_sewing_in_rigged_Teichmuller_space, Radnell-Schippers:Fiber_structure_and_local_coordinates_for_the_Teichmuller_space_of_a_bordered_Riemann_surface}.
	The latter appear to be relevant for the complex analytic structure of moduli spaces (see the review~\cite{RSS:Quasiconformal_Teichmuller_theory_as_analytical_foundation_for_CFT}).
	In the present work, we mainly focus on the moduli spaces as algebraic objects, for which the real-analytic parametrizations and the Fr\"olicher smooth setup are sufficient.
}(cf.~\cite{Segal:Definition_of_CFT, Segal:Definition_of_CFT_collection}).
In their role as cobordisms with complex-analytic structure, the main operations on $\moduligb$ are sewing (gluing) of boundary components using the parametrizations, and unraveling (cutting) along parametrized simple loops in the interior.
We equip the moduli spaces $\moduligb$ with smooth Fr\"olicher structures induced by the actions of $\DefC$, which emphasize the sewing and unraveling as smooth operations.
Initially, our definition of the Fr\"olicher structure of the Segal moduli space $\moduligb$ includes deformations of surfaces at simple analytic loops in the interior (for example, the Fenchel--Nielsen twist coordinates, or Schiffer variation).
Our Virasoro uniformization result shows that, infinitesimally, these may be realized by deformations of boundary components only.

We also study the interplay of the infinite-dimensional moduli spaces $\moduligb$ and the ($6 \genus - 6 + 4 \boundaries$)-dimensional moduli spaces $\fmoduligbzi$ of hyperbolic surfaces with one marked point on each (unparametrized) boundary component (Section~\ref{section:hyperbolic}), like the ones considered in~\cite{Abikoff:The_real_analytic_theory_of_Teichmuller_space, Ivashkovich-Shevchishin:Holomorphic_structure_on_the_space_of_Riemann_surfaces_with_marked_boundary, Mirzakhani:Simple_geodesics_and_Weil-Petersson_volumes_of_moduli_spaces_of_bordered_Riemann_surfaces, Liu:Moduli_of_J-holomorphic_curves_with_Lagrangian_boundary_conditions_and_open_Gromov-Witten_invariants_for_an_S1-equivariant_pair}.
The action of the mapping class groups on Teichm\"uller spaces corresponding to $\fmoduligbzi$ is free, and thus $\fmoduligbzi$ are manifolds (as opposed to orbifolds).
As a by-product, we see that the Fr\"olicher structure of $\moduligb$ is non-trivial (Corollary~\ref{corollary:fr_nontrivial}).
Moreover, by our Virasoro uniformization result, deformations of boundary components by complex deformations provide a new set of coordinates on $\fmoduligbzi$.

\subsection*{Lie theory of complex deformations}
The set $\DefC$ sits at an interesting place between the infinite-dimensional Lie group of real-analytic orientation-preserving diffeomorphisms of the circle~\cite{Milnor:Remarks_on_infinite_dimensional_Lie_groups, Kriegl-Michor:Convenient_setting_for_real_analytic_mappings, Dahmen-Schmeding:The_Lie_group_of_real_analytic_diffeomorphisms_is_not_real_analytic} and the impossibility of finding a complex Lie group integrating the Witt algebra relations~\cite{Milnor:Remarks_on_infinite_dimensional_Lie_groups, Pressley-Segal:Loop_groups, Neretin:Holomorphic_extensions_of_representations_of_the_group_of_diffeomorphisms_of_the_circle, Lempert:The_problem_of_complexifying_a_Lie_group, Schottenloher:Mathematical_introduction_to_CFT}.
Since the composition and inversion operations on $\DefC$ are only partially defined, it is not an infinite-dimensional Lie group (even in the Fr\"olicher sense~\cite{Ntumba-Batubenge:On_the_way_to_Frolicher_Lie_groups, Laubinger:Lie_algebra_for_Frolicher_groups}).
In fact, $\DefC$ is not even a local Lie group (in the sense of~\cite{Olver:Non-associative_local_Lie_group, Neeb:Infinite-Dimensional_Lie_Groups, Goldbring:Hilberts_fifth_problem_for_local_groups, Gloeckner-Neeb:Infinite-dimensional_Lie_groups}).
Indeed, the composition of a pair of complex deformations is only guaranteed to exist if both elements are near the identity; 
and the sets of composable and invertible elements are neither open nor closed in $\DefC \times \DefC$ and $\DefC$, respectively (see Definition~\ref{def:complex_deformation} and Remark~\ref{remark:not_open_closed}).
Moreover, just like for $\Diffpan$, the exponential map of $\DefC$ is not surjective onto any neighborhood of the identity (cf.~\cite{Hamilton:The_inverse_function_theorem_of_Nash_and_Moser, Milnor:Remarks_on_infinite_dimensional_Lie_groups}).
These facts further motivate our usage of Fr\"olicher structures, where tangent vectors are represented by smooth curves in $\DefC$, which are the flows of smoothly time-dependent vector fields\footnote{Restricting to the subsets of composable or invertible complex deformations, this tangent space is inhomogeneous, comprising holomorphic vector fields on domains on which the complex deformations at hand have strong properties regarding analytic continuation. See Remark~\ref{rem:tangent_spaces}.} (see Section~\ref{section:complex_deformations} and Appendix~\ref{section:frolicher}).

The \emph{semigroup of annuli} is the moduli space $\annuli$ together with a sewing operation. Pioneered by Neretin and Segal, it complexifies $\Diffpan$ in ways that are complementary to complex deformations~\cite{Neretin:Holomorphic_extensions_of_representations_of_the_group_of_diffeomorphisms_of_the_circle, Neretin:Categories_of_symmetries_and_infinite-dimensional_groups, Segal:Definition_of_CFT_collection, Radnell-Schippers:The_semigroup_of_rigged_annuli_and_the_Teichmueller_space_of_the_annulus, Henriques-Tener:The_Segal-Neretin_semigroup_of_annuli}\footnote{Interestingly, in~\cite[Section~1.11]{Neretin:Holomorphic_extensions_of_representations_of_the_group_of_diffeomorphisms_of_the_circle}, Neretin defines a diffeological structure on the semigroup of annuli, which is compatible with our Fr\"olicher structure on~$\annuli$.} --- while annuli may always be composed but lack inverses, complex deformations sometimes admit inverses, but composition is only partially defined.
The semigroup of annuli plays an important role in the Segal approach to (euclidean) conformal field theory~\cite{Segal:Definition_of_CFT, Segal:Definition_of_CFT_collection}, as recently elaborated in~\cite{BGKRV:Virasoro_structure_and_the_scattering_matrix_for_Liouville_CFT, Henriques:The_functorial_approach_to_chiral_2D_CFT, ARS:The_moduli_of_annuli_in_random_conformal_geometry, Henriques-Tener:Integrating_positive_energy_representations_of_the_Virasoro_algebra, Maibach-Peltola:From_the_conformal_anomaly_to_the_Virasoro_algebra,BGKR:Semigroup_of_annuli_in_Liouville_CFT, Maibach-Peltola:Universality_of_the_conformal_anomaly}.
Lately, this general perspective on CFT has gained increasing attention in probabilistic realizations of CFTs, such as Liouville theory (see~\cite{GKR:Review_on_the_probabilistic_construction_and_conformal_bootstrap_in_Liouville_theory, GKR:Compactified_imaginary_Liouville_theory, GKRV:Segals_axioms_for_Liouville_theory} and references therein), and potentially for SLE loop measures (see~\cite[Section~8.1]{Baverez-Jego:The_CFT_of_SLE_loop_measures_and_the_Kontsevich-Suhov_conjecture} and~\cite[Section~1.3]{Maibach-Peltola:Universality_of_the_conformal_anomaly}). 
In the present article, we develop the geometric framework underlying Segal's axioms via the complex deformations and their actions on the infinite-dimensional Segal moduli spaces.
Alongside with the application to the conformal anomaly~\cite{Maibach-Peltola:Universality_of_the_conformal_anomaly}, which yields implications both for CFT and for SLE loop measures,  
we expect our framework to be especially useful in the context of recent results on Virasoro conformal blocks~\cite{GRSS:Probabilistic_conformal_blocks_for_Liouville_CFT_on_the_torus, BGKRX:Virasoro_conformal_blocks_and_modular_functor_from_Liouville_CFT}.

\subsection*{Acknowledgements}
We would like to thank Peter Michor for suggesting the use of Fr\"olicher structures, and David Radnell and Eric Schippers for insightful discussions on the Segal moduli space. 
Part of this work was included in S.M's doctoral thesis~\cite[Appendix~C]{Maibach:PhD_thesis}.

S.M.~has been supported by the Deutsche Forschungsgemeinschaft (DFG, German Research Foundation) under Germany's Excellence Strategy EXC-2047/1-390685813, by the grant CRC 1060 ``The
Mathematics of Emergent Effects'' (Project-ID 211504053), and by the grant from the Simons Foundation International [SFI-MPS-PP-00012621-19].

This material is part of a project that has received funding from the  European Research Council (ERC) under the European Union's Horizon 2020 research and innovation programme (101042460): 
ERC Starting grant ``Interplay of structures in conformal and universal random geometry'' (ISCoURaGe) 
and from the Academy of Finland grant number 340461 ``Conformal invariance in planar random geometry.''

E.P.~is also supported by 
the Academy of Finland Centre of Excellence Programme grant number 346315 ``Finnish centre of excellence in Randomness and STructures (FiRST)'' 
and by the Deutsche Forschungsgemeinschaft (DFG, German Research Foundation) project number 390534769 ``Matter and Light for Quantum Computing (ML4Q).

\section{Complex deformations of the circle}
\label{section:complex_deformations}

The vector space of real-analytic vector fields on the unit circle, which we denote by $\VectR$, is the Lie algebra of the infinite-dimensional Lie group $\Diffpan$ of real-analytic, orientation-preserving diffeomorphisms of the unit circle $S^1$.
Due to results of Neretin~\cite[Section~1.10]{Neretin:Holomorphic_extensions_of_representations_of_the_group_of_diffeomorphisms_of_the_circle}, Pressley~\&~Segal~\cite[Proposition~3.3.2]{Pressley-Segal:Loop_groups}, and Lempert~\cite{Lempert:The_problem_of_complexifying_a_Lie_group}, it is known, however, that the complexification
\begin{equation}
	\VectC = \VectR \otimes \C
\end{equation}
of $\VectR$ is not the Lie algebra of any infinite-dimensional complex Lie group.
Yet, interpreting $\VectC$ as the vector space of real-analytic complex-valued vector fields on the unit circle embedded into the complex plane,
\begin{equation}
S^1 = \setsuchthat{z \in \C}{|z| = 1},
\end{equation}
the Lie bracket of $\VectC$ can still be related to the composition $\phi \circ \psi$ of certain real-analytic maps $\phi, \psi \colon S^1 \to \C$ defined using analytic continuation (see Figure~\ref{fig:composition} for an illustration).
In this section, we study a set of such maps, which integrates the Lie algebra $\VectC$ as a group-like object --- meaning that the composition law is smooth and associative, yet only partially defined and invertible.

\begin{figure}
\centering
\begin{tikzpicture}
	\def\A{(59.0006, -21.8665).. controls (58.2304, -21.7667) and (58.1779, -22.5051) .. (58.5406, -22.7391).. controls (58.9032, -22.973) and (59.334, -22.513) .. (59.5271, -22.8401).. controls (59.7201, -23.1671) and (59.0667, -23.2547) .. (58.6796, -23.0431).. controls (58.2924, -22.8315) and (57.8129, -22.1583) .. (57.2743, -22.2799).. controls (56.7357, -22.4016) and (56.6966, -22.9379) .. (56.8691, -23.5075).. controls (57.0416, -24.0772) and (57.6316, -25.1429) .. (58.228, -25.2765).. controls (58.8245, -25.4101) and (59.2873, -24.8977) .. (58.9978, -24.5108).. controls (58.7082, -24.1238) and (58.2232, -24.7307) .. (57.9039, -24.523).. controls (57.5846, -24.3153) and (57.8594, -23.8581) .. (58.2524, -23.8399).. controls (58.6453, -23.8218) and (58.808, -23.922) .. (59.0279, -24.1081).. controls (59.2478, -24.2941) and (59.3566, -24.6009) .. (59.7508, -24.6543).. controls (60.145, -24.7076) and (60.433, -24.4789) .. (60.4259, -23.8164).. controls (60.4187, -23.1538) and (59.7709, -21.9664) .. (59.0006, -21.8665) -- cycle}
	\def\B{(58.8855, -21.8576).. controls (58.3393, -21.8451) and (58.2255, -22.3059) .. (58.3978, -22.5905).. controls (58.3252, -22.6058) and (58.2546, -22.6292) .. (58.1872, -22.6601).. controls (57.91, -22.4277) and (57.6032, -22.2056) .. (57.2743, -22.2799).. controls (56.7357, -22.4015) and (56.6966, -22.9379) .. (56.8691, -23.5075).. controls (57.0415, -24.0772) and (57.6316, -25.1428) .. (58.228, -25.2765).. controls (58.8245, -25.4101) and (59.2873, -24.8977) .. (58.9978, -24.5108).. controls (58.9934, -24.505) and (58.989, -24.4993) .. (58.9845, -24.4939).. controls (59.0775, -24.4557) and (59.1642, -24.4036) .. (59.2416, -24.3396).. controls (59.3638, -24.4829) and (59.4992, -24.6202) .. (59.7508, -24.6542).. controls (60.145, -24.7076) and (60.433, -24.4789) .. (60.4259, -23.8163).. controls (60.4187, -23.1538) and (59.7709, -21.9664) .. (59.0006, -21.8665).. controls (58.9603, -21.8613) and (58.922, -21.8583) .. (58.8855, -21.8575) -- cycle}
	\def\C{(58.1872, -22.6601).. controls (57.832, -22.8231) and (57.6043, -23.1781) .. (57.6043, -23.569).. controls (57.6043, -23.7779) and (57.6697, -23.9815) .. (57.7914, -24.1514).. controls (57.8577, -23.9913) and (58.0373, -23.8499) .. (58.2524, -23.8399).. controls (58.6453, -23.8218) and (58.808, -23.922) .. (59.0279, -24.1081).. controls (59.1074, -24.1753) and (59.1724, -24.2584) .. (59.2416, -24.3396).. controls (59.4713, -24.1496) and (59.6043, -23.8671) .. (59.6043, -23.569).. controls (59.6043, -23.3997) and (59.5613, -23.2331) .. (59.4794, -23.085).. controls (59.3144, -23.1923) and (58.9378, -23.1843) .. (58.6796, -23.0431).. controls (58.5289, -22.9607) and (58.3641, -22.8083) .. (58.1872, -22.6601) -- cycle}

	\useasboundingbox (-3,-3) rectangle (3,3);
    \node (pic_left) at (0, 0) {};
    \begin{scope}[shift=(pic_left), scale=1.5]	
		\begin{scope}[shift={(-58.6043, 23.569)}]
			\path[fill=color2, even odd rule, pattern=north west lines, even odd rule, pattern color = color2]
				\B
				\C;
			\path[draw, very thick]
				(58.6043, -23.569) circle (1.0);
			\path[draw=color1, very thick, postaction={decorate, decoration={markings, 
  mark=between positions 0.1 and 1 step 0.1 with {\arrow{to}}}}]
				\A;
		\end{scope}
        \node[circle, fill, inner sep=1pt, outer sep=0.0cm]
            (zero_right) at (0, 0) {};
        \node[circle, fill, inner sep=1pt, outer sep=0.0cm, label=left:{$1$}]
            (one_right) at (1, 0) {};
        \node[text=color1]
			(img_right) at (1.6, 1.3) {$\phi(S^1)$};
        \node[draw, fill=white, rounded corners=2pt, inner sep=3pt]
			(U_right) at (-1.3, 0.7) {$\DefCnbhd{\phi}$};
    \end{scope}
    
\end{tikzpicture}
\caption{
	The image of $S^1$ under a complex deformation $\phi \in \DefC$ and the set $\DefCnbhd{\phi}$, which is defined in Equation~\eqref{eq:def_defcnbhd} as the region bounded by $S^1$ and $\phi(S^1)$.
}
\label{fig:complex_deformation}
\end{figure}

\begin{definition}
\label{def:defc}
A \emph{complex deformation} of the unit circle is an injective real-analytic map $\phi \colon S^1 \to \C \setminus \{0\}$ whose winding number around $0$ equals $+1$.
We denote the set of complex deformations by $\DefC$.
\end{definition}

\noindent
Reasons to use precisely this definition are detailed in Section~\ref{section:invertible_composable}, where also the composition and inversion operations are defined.

\subsection{Complex-valued vector fields on the circle}
\label{section:vector_fields}

In this section, we cover some preliminaries on $\VectC$, $\VectR$, and the Witt algebra.
In the standard coordinate $z$ on $\C$, we denote a vector field $v \in \VectC$ by ${v = v(z) \, \partial_z}$.
Here, $v(z) \in \C$ is a real-analytic function of $z \in S^1$.
In the following, we use this identification to put a topology on the vector space $\VectC$.

Let $\mathcal{O}(S^1)$ denote the complex vector space of complex-valued real-analytic functions on~$S^1$.
Such functions are germs of holomorphic functions defined on open neighborhoods of~$S^1$ inside $\C$.
Concretely, for each $f \in \mathcal{O}(S^1)$, there exists $n \geq 1$ such that $f$ is the restriction of a holomorphic function on the annular open neighborhoods of $S^1$ defined by
\begin{equation}
	U_n = \setsuchthat{z \in \C}{1 - \frac{1}{n} < |z| < 1 + \frac{1}{n}}.
\label{eq:circlenbhd}
\end{equation}
The complex vector space of holomorphic functions on $U_n$, denoted by $\mathcal{O}(U_n)$, comes with the topology of uniform convergence on compact sets in $U_n$.
We equip $\mathcal{O}(S^1)$ with the inductive limit topology induced by the restrictions $\mathcal{O}(U_n) \to \mathcal{O}(S^1)$.
This topology makes $\mathcal{O}(S^1)$ a convenient vector space in the sense of~\cite[Theorem~8.4]{Kriegl-Michor:Convenient_setting_of_global_analysis}.
Following~\cite[Section~1.2 and Theorem~2.14]{Kriegl-Michor:Convenient_setting_of_global_analysis}, the notion of smoothness on such a topological vector space induces a \emph{Fr{\"o}licher structure} (see Appendix~\ref{section:frolicher}) where the smooth curves are given by
\begin{equation}
	\curves{\mathcal{O}(S^1)} = C^\infty(\R, \mathcal{O}(S^1)).
\end{equation}
This is simply the set of smooth functions $\gamma \colon \R \times S^1 \to \C$ such that for each fixed $t \in \R$, the function $\gamma(t, \blank) \colon S^1 \to \C$ is real-analytic.
By identifying vector fields and functions as
\begin{equation}
\begin{aligned}
	\VectC &\longrightarrow \mathcal{O}(S^1), \\
	v(z) \, \partial_z &\longmapsto v(z),
\end{aligned}
\label{eq:vectc_os1}
\end{equation}
we equip $\VectC$ with the topology of $\mathcal{O}(S^1)$.

The \emph{Witt algebra} is the infinite-dimensional complex Lie algebra generated by $\ell_n$, $n \in \Z$, and the relations $[\ell_n, \ell_m] = (n - m) \ell_{n + m}$, where $n, m \in \Z$.
The generators may be realized in $\VectC$ by identifying them with the monomials $- z^{n + 1}$, i.e.,
\begin{equation}
	\ell_n = - z^{n + 1} \, \partial_z, \qquad n \in \Z.
\label{eq:witt_basis}
\end{equation}
Finite complex linear combinations of the Witt algebra generators~\eqref{eq:witt_basis} are dense in $\VectC$, and moreover, the standard Lie bracket on $\VectC$ given by
\begin{equation}
	[v, w] = (w(z) v'(z) - v(z) w'(z)) \, \partial_z, \qquad v, w \in \VectC,
\label{eq:lie_bracket_vector_fields}
\end{equation}
agrees with the relations of the Witt algebra when specialized to the generators.

Given a biholomorphism $F \colon A \to B$ between annular open neighborhoods $A, B$ of $S^1$ inside $\C$, and a vector field $v \in \VectC$ which has analytic continuation to $B$, the pullback of $v$ by $F$ is a holomorphic vector field $F^* v$ on $A$ given by
\begin{equation}
	F^* v = \frac{v(F(z))}{F'(z)} \partial_z.
\label{eq:pullback_vector_field}
\end{equation}
In particular, the \emph{inversion} of the Riemann sphere $\hat \C = \C \cup \{\infty\}$ given by
\begin{equation}
	\inversion \colon \hat \C \to \hat \C, \qquad z \mapsto \frac{1}{z}
\label{eq:inversion}
\end{equation}
acts on the generators $\ell_n$ by pullback as
\begin{equation}
	\inversion^* \ell_n 
	= - z^{-(n+1)} (-z^{2}) \partial_z
	= - \ell_{-n}, \qquad n \in \Z.
\end{equation}
When viewed as a vector field on the Riemann sphere, for $n \geq 2$, we observe that $\ell_n$ extends holomorphically to $\C$ with a singularity at $\infty$,  while $\ell_{-n}$ extend holomorphically to the punctured sphere $\hat \C \setminus \{0\}$ with a singularity at $0$.
The vector fields $\ell_{-1}$, $\ell_0$, and $\ell_1$ extend holomorphically to the whole Riemann sphere $\hat \C$, and they span a Lie subalgebra of $\VectC$ isomorphic to $\mobvect$.

As a real Lie algebra, $\VectC$ is generated by $\setsuchthat{\ell_n, \, \ii \ell_n}{n \in \Z}$. 
More geometrically, it is convenient to consider the following tangential ($\parallel$) and normal ($\perp$) vector fields to $S^1$: 
\begin{equation}
\begin{aligned}
	\ellpar{n} = & \; \frac{\ell_n - \ell_{-n}}{2}, \qquad
	&\ellpari{n} = & \; \frac{\ell_n + \ell_{-n}}{2 \ii}, \\
	\ellperpi{n} = & \; \frac{\ell_n - \ell_{-n}}{2 \ii}, \qquad
	&\ellperp{n} = & \; \frac{\ell_n + \ell_{-n}}{2}.
\end{aligned}
\label{eq:def_ellpar_ellperp}
\end{equation}
The tangential vector fields generate the space $\VectR$ of real-analytic sections of the tangent bundle of $S^1$.
In the coordinate $\theta$ defined by $z = e^{\ii \theta}$, these become
\begin{equation}
	\ellpar{n}(e^{\ii \theta}) = \ii e^{\ii \theta} \sin(n \theta)
	\qquad \textnormal{and} \qquad 
	\ellpari{n}(e^{\ii \theta}) = \ii e^{\ii \theta} \cos(n \theta) .
\label{eq:ellpar}
\end{equation}
Note that the factor $\ii e^{\ii \theta}$ rotates the real-valued functions of $\theta$ into tangent vector fields.
More generally, we use the coordinate notation
\begin{equation}
	v = \thetacoords{v}(\theta) \, \partial_\theta, \qquad
	\partial_\theta = \ii e^{\ii \theta} \, \partial_z, \qquad
	\thetacoords{v}(\theta)
	= - \ii e^{-\ii \theta} v(e^{\ii \theta}) = \frac{v(z)}{\ii z},
\label{eq:theta_coords}
\end{equation}
along with the coordinate function $v(z)$ in~\eqref{eq:vectc_os1}.

\subsection{Invertible and composable complex deformations}
\label{section:invertible_composable}

\begin{figure}[t]
\centering
\begin{tikzpicture}
	\def\A{(58.6043, -23.569) circle (1.0)}
	\def\B{(58.6854, -25.1077).. controls (58.6462, -25.0937) and (58.6144, -25.0624) .. (58.5928, -25.0167).. controls (58.5689, -24.9662) and (58.5681, -24.8725) .. (58.5907, -24.7806).. controls (58.6095, -24.7041) and (58.6067, -24.6687) .. (58.5789, -24.6323).. controls (58.5292, -24.5671) and (58.3673, -24.4825) .. (58.1861, -24.427).. controls (58.0545, -24.3867) and (58.0059, -24.369) .. (57.9773, -24.3511).. controls (57.9454, -24.3311) and (57.9051, -24.2865) .. (57.878, -24.2411).. controls (57.8437, -24.1837) and (57.7348, -24.0283) .. (57.6984, -23.9847).. controls (57.6391, -23.9138) and (57.5582, -23.8371) .. (57.527, -23.8222).. controls (57.4813, -23.8003) and (57.4391, -23.8049) .. (57.3606, -23.84).. controls (57.2777, -23.8771) and (57.2243, -23.8895) .. (57.1709, -23.884).. controls (57.0935, -23.8759) and (57.0401, -23.8293) .. (57.0338, -23.7643).. controls (57.0291, -23.7162) and (57.0355, -23.6951) .. (57.0637, -23.6629).. controls (57.1055, -23.6153) and (57.1387, -23.603) .. (57.225, -23.6031).. controls (57.2666, -23.6031) and (57.333, -23.606) .. (57.3725, -23.6094) -- (57.4443, -23.6156) -- (57.4745, -23.5932).. controls (57.5555, -23.5335) and (57.6469, -23.2872) .. (57.648, -23.1262).. controls (57.6484, -23.0577) and (57.6463, -23.0464) .. (57.6268, -23.0064).. controls (57.6082, -22.9686) and (57.5928, -22.9525) .. (57.52, -22.895).. controls (57.4114, -22.8092) and (57.3803, -22.7694) .. (57.3691, -22.7017).. controls (57.352, -22.5991) and (57.4421, -22.5178) .. (57.5483, -22.5398).. controls (57.6015, -22.5509) and (57.6502, -22.5802) .. (57.7307, -22.6495).. controls (57.8272, -22.7325) and (57.8452, -22.7407) .. (57.9282, -22.74).. controls (57.9894, -22.7395) and (58.0059, -22.736) .. (58.0908, -22.7048).. controls (58.1985, -22.6654) and (58.2774, -22.6295) .. (58.3217, -22.5997).. controls (58.3656, -22.5701) and (58.4837, -22.5196) .. (58.5226, -22.5139).. controls (58.5569, -22.5088) and (58.5799, -22.5129) .. (58.7622, -22.5552).. controls (58.9065, -22.5889) and (58.9987, -22.5996) .. (59.059, -22.5897).. controls (59.1349, -22.5772) and (59.1659, -22.5528) .. (59.2322, -22.4533).. controls (59.3081, -22.3392) and (59.4089, -22.2804) .. (59.4853, -22.3056).. controls (59.5675, -22.3328) and (59.6077, -22.4213) .. (59.5753, -22.5042).. controls (59.5539, -22.5588) and (59.5221, -22.5981) .. (59.447, -22.6629).. controls (59.4073, -22.6971) and (59.3667, -22.7391) .. (59.3567, -22.7562).. controls (59.2996, -22.8535) and (59.347, -23.0629) .. (59.4896, -23.3445).. controls (59.5299, -23.4241) and (59.5671, -23.5028) .. (59.5722, -23.5194).. controls (59.5773, -23.5361) and (59.585, -23.557) .. (59.5893, -23.5659).. controls (59.5999, -23.5877) and (59.5959, -23.6157) .. (59.5582, -23.7878).. controls (59.5542, -23.8065) and (59.5486, -23.844) .. (59.5459, -23.871).. controls (59.5432, -23.898) and (59.5398, -23.9223) .. (59.538, -23.925).. controls (59.5363, -23.9276) and (59.5297, -23.9667) .. (59.5232, -24.0119).. controls (59.501, -24.1659) and (59.4867, -24.1875) .. (59.37, -24.2425).. controls (59.1016, -24.3691) and (58.9645, -24.4501) .. (58.8645, -24.541).. controls (58.7724, -24.6248) and (58.7664, -24.6527) .. (58.8163, -24.7636).. controls (58.8508, -24.8402) and (58.8695, -24.9039) .. (58.8695, -24.9449).. controls (58.8695, -25.0198) and (58.8218, -25.095) .. (58.7643, -25.111).. controls (58.7297, -25.1206) and (58.7209, -25.1202) .. (58.6853, -25.1076) -- cycle}
	\def\C{(58.3512, -22.5503).. controls (58.4726, -22.5203) and (58.5737, -22.6275) .. (58.6872, -22.616).. controls (58.7982, -22.5787) and (58.9159, -22.5668) .. (59.0275, -22.6092).. controls (59.1764, -22.6494) and (59.3131, -22.7556) .. (59.3567, -22.908).. controls (59.3997, -23.0239) and (59.4474, -23.1414) .. (59.5332, -23.2329).. controls (59.6273, -23.3417) and (59.5948, -23.4898) .. (59.635, -23.6159).. controls (59.6837, -23.7204) and (59.6955, -23.8366) .. (59.6746, -23.9493).. controls (59.6405, -24.1554) and (59.4246, -24.2988) .. (59.2217, -24.2753).. controls (59.0864, -24.2952) and (59.0527, -24.4734) .. (58.911, -24.4935).. controls (58.7757, -24.5504) and (58.6281, -24.5141) .. (58.487, -24.5348).. controls (58.333, -24.5508) and (58.1642, -24.5231) .. (58.0539, -24.405).. controls (57.9405, -24.3207) and (57.9868, -24.1404) .. (57.8554, -24.0779).. controls (57.6778, -23.9887) and (57.612, -23.763) .. (57.6619, -23.5793).. controls (57.698, -23.4838) and (57.7227, -23.3848) .. (57.7106, -23.2815).. controls (57.6995, -23.1522) and (57.7157, -23.0108) .. (57.8108, -22.9139).. controls (57.903, -22.8213) and (58.0321, -22.7737) .. (58.1135, -22.6687).. controls (58.1796, -22.6082) and (58.2627, -22.5656) .. (58.3512, -22.5504)(58.6041, -23.5727) -- (58.6041, -23.5769)}
	
	\useasboundingbox (-5.5,-2.5) rectangle (5.5,2);
    \node (pic_left) at (-3, 0) {};
    \begin{scope}[shift=(pic_left), scale=1.5]	
		\begin{scope}[shift={(-58.6043, 23.569)}]
			\path[fill=color2, even odd rule, pattern=north west lines, even odd rule, pattern color = color2]
				\A
				\B;
			\path[draw, very thick]
				\B;
			\path[draw=color1, very thick]
				\A;
		\end{scope}
        \node[circle, fill, inner sep=1pt, outer sep=0.0cm]
            (zero_left) at (0, 0) {};
		\node[anchor=west]
			(img_left) at (0.9, 0.9) {$\phi^{-1}(S^1)$};
        \node[draw, fill=white, rounded corners=2pt, inner sep=3pt, anchor=west]
			(U_left) at (0.5, -1.35) {$\DefCnbhd{\phi^{-1}} = \phi^{-1}\big(\DefCnbhd{\phi}\big)$};
        \node[circle, fill, inner sep=0.8pt, outer sep=0.0cm]
            (U_pt_left) at (0.12, -1.35) {};
    \end{scope}

    \node (pic_right) at (3, 0) {};
    \begin{scope}[shift=(pic_right), scale=1.5]
		\begin{scope}[shift={(-58.6043, 23.569)}]
			\path[fill=color2, even odd rule, pattern=north west lines, even odd rule, pattern color = color2]
				\A
				\C;
			\path[draw, very thick]
				\A;
			\path[draw=color1, very thick]
				\C;
		\end{scope}
        \node[circle, fill, inner sep=1pt, outer sep=0.0cm]
            (zero_right) at (0, 0) {};
		\node[text=color1, anchor=west]
			(img_right) at (0.9, 0.9) {$\phi(S^1)$};
        \node[draw, fill=white, rounded corners=2pt, inner sep=3pt]
			(U_right) at (0.5, -0.3) {$\DefCnbhd{\phi}$};
        \node[circle, fill, inner sep=0.8pt, outer sep=0.0cm]
            (U_pt_right) at (0.98, -0.43) {};
    \end{scope}

    \path[->] (-0.5, 0) edge[above] node {$\phi$} (0.5, 0);
    \path[draw] (U_pt_right) -- (U_right);
    \path[draw] (U_pt_left) -- (U_left);
    
\end{tikzpicture}
\caption{
	An invertible complex deformation $\phi \in \definvset$ maps $\DefCnbhd{\phi^{-1}}$ to $\DefCnbhd{\phi}$.
}
\label{fig:inverse}
\end{figure}

Since the set $\DefC$ of complex deformations is an open subset of $\mathcal{O}(S^1)$ (cf.~\cite{Kriegl-Michor:Convenient_setting_of_global_analysis}) we can equip $\DefC$ with the subspace topology and Fr{\"o}licher structure.
The condition on the winding number around $0$ implies that $\DefC$ with this topology is connected.
Moreover, requiring the image of $S^1$ under a complex deformation $\phi \in \DefC$ to avoid~$0$ unambiguously defines the inside boundary of the following closed subset of $\C$: 
\begin{equation}
	\DefCnbhd{\phi} = \overline{\bigcup\setsuchthat{V \in \pi_0\Big(\hat \C \setminus \big(S^1 \cup \phi(S^1)\big)\Big)}{0, \infty \notin V}} .
\label{eq:def_defcnbhd}
\end{equation}
This is the set of points bounded by the inner and outer boundaries of the union
of the curve $\phi(S^1)$ and the unit circle $S^1$ as viewed from $0$ and $\infty$; see Figure~\ref{fig:complex_deformation}.

\begin{figure}
\centering
\includestandalone[]{figures/fig_composition}
\caption{
	The composition $\phi \circ \psi$ of a composable pair $(\phi, \psi) \in \defmultset$ is defined by analytically continuing $\phi$ to $\psi(S^1)$ so that we have a biholomorphism $\phi \colon \DefCnbhd{\psi} \to \phi\big(\DefCnbhd{\psi}\big)$.
}
\label{fig:composition}
\end{figure}

\noindent
Recall that a function is \emph{univalent} if it is holomorphic and injective.

\begin{definition}\leavevmode
\begin{enumerate}
	\item
	A complex deformation $\phi \in \DefC$ is \emph{invertible} if $\phi^{-1} \colon \phi(S^1) \to \C$ extends univalently and without zeros to an open neighborhood of $\DefCnbhd{\phi}$.
	We denote the set of invertible complex deformations by $\definvset \subset \DefC$.
	\item
	A pair $\phi, \psi \in \DefC$ of complex deformations is \emph{composable} if $\phi$ extends univalently and without zeros an open neighborhood of $\DefCnbhd{\psi}$.
	We denote the set of composable pairs by $\defmultset \subset \DefC \times \DefC$.
\end{enumerate}
\label{def:complex_deformation}
\end{definition}
The requirements of univalent extensions to neighborhoods of $U(\phi)$ are sufficient to make inverses and compositions uniquely defined.
Namely, the inverse $\phi^{-1}$ of a general complex deformation $\phi \in \DefC$ is initially defined on $\phi(S^1)$.
If $\phi^{-1}$ does not have singularities on $S^1$, but does have singularities in the set $\DefCnbhd{\phi}$, then $\phi$ may have several analytic continuations to $S^1$.
Similarly, for a composition, we need that $\phi$ analytically continues to $\psi(S^1)$ in a unique way.
Initially, $\phi$ is defined just on $S^1$, and if there is a singularity in $\DefCnbhd{\psi}$, there may be multiple ways to continue.
Avoiding zeros ensures that inverses and compositions are again complex deformations, that is, they take values in $\C \setminus \{0\}$ with the correct winding number.
The inversion and composition laws
\begin{equation}
\begin{aligned}
\definvset &\longrightarrow \definvset, \qquad & \defmultset &\longrightarrow \DefC, \\
\phi &\longmapsto \phi^{-1}, \qquad & (\phi, \psi) &\longmapsto \phi \circ \psi,
\end{aligned}
\label{eq:def_inv_mult}
\end{equation}
are respectively defined by the analytic continuation of $\phi^{-1}$ to $\DefCnbhd{\phi}$, and the analytic continuation of $\phi$ to $\DefCnbhd{\psi}$.
See also Figures~\ref{fig:inverse} and~\ref{fig:composition}.

\begin{proposition}
Inversion and composition have the following properties.
\label{prop:defc_inv_comp}
\begin{enumerate}
	\item
	\label{item:defc_inv_comp1}
	The identity map $\id$ on $\C$ serves as the identity element of $\DefC$ in the sense that $(\id, \phi), (\phi, \id) \in \defmultset$ for all $\phi \in \DefC$ and $\id \circ \phi = \phi \circ \id = \phi$.
	\item
	\label{item:defc_inv_comp2}
	For $\phi \in \definvset$ we have $\phi^{-1} \in \definvset$, and $(\phi^{-1})^{-1} = \phi$. Moreover,  we have $(\phi, \phi^{-1}), (\phi^{-1}, \phi) \in \defmultset$, and $\phi \circ \phi^{-1} = \phi^{-1} \circ \phi = \id$.
	\item
	\label{item:defc_inv_comp3}
	For $\phi_1, \phi_2, \phi_3 \in \DefC$, we have $(\phi_1 \circ \phi_2) \circ \phi_3 = \phi_1 \circ (\phi_2 \circ \phi_3)$ if each of the four pairs is composable.
\end{enumerate}
\end{proposition}

\begin{proof}\leavevmode
\begin{enumerate}
	\item
	$\id$ is defined on $\C$, and $\DefCnbhd{\id} = S^1$.
	\item
	By invertability, $\phi^{-1}$ extends to $\DefCnbhd{\phi}$, yielding $\phi^{-1} \colon \DefCnbhd{\phi} \to \DefCnbhd{\phi^{-1}}$.
	Hence, $\phi$ extends to $\DefCnbhd{\phi^{-1}}$.
	The conditions for composing $\phi$ and $\phi^{-1}$ both ways are precisely the conditions for their individual invertability.
	\item
	Either $\phi_2$ is extended to $\DefCnbhd{\phi_3}$ first, and then $\phi_1$ is extended to $\DefCnbhd{\phi_2 \circ \phi_3}$, resulting in $\phi_1 \circ (\phi_2 \circ \phi_3)$; 
	or $\phi_1$ is extended to $\DefCnbhd{\phi_2}$, and then $\phi_1 \circ \phi_2$ is extended to $\DefCnbhd{\phi_3}$, resulting in $(\phi_1 \circ \phi_2) \circ \phi_3$.
	By the identity theorem, we only need to show that both compositions agree in a neighborhood of a point $z_0 \in S^1$.
	In the case where $\phi_3(S^1) \cap S^1 \neq \emptyset$, we choose $z_0$ such that $\phi_3(z_0) \in S^1$.
	Then, $w_0 = \phi_2(\phi_3(z_0))$ is defined without analytic continuation and is an element of $\DefCnbhd{\phi_2}$.
	Thus, the definition of the value of $\phi_1$ at $w_0$ 
	and the composition $\phi_1 \circ \phi_2$ use the same analytic extension of $\phi_1$ to a neighborhood of $\DefCnbhd{\phi_2}$, which shows that both ways of composing agree in a neighborhood of $z_0$.
	In the case where $\phi_3(S^1) \cap S^1 = \emptyset$, the curve $\phi_3(S^1)$ lies either inside or outside $S^1$.  
	Using the analytic extension of $\phi_2$ to $\DefCnbhd{\phi_3}$, the curve $\phi_2(\DefCnbhd{\phi_3})$ might intersect $S^1$ or not.
	If it does not intersect $S^1$, then we have an inclusion $\DefCnbhd{\phi_2 \circ \phi_3} \subset\DefCnbhd{\phi_2}$ or $\DefCnbhd{\phi_2 \circ \phi_3} \subset \DefCnbhd{\phi_2}$,
	and the analytic extensions of $\phi_1$ agree in either case. 
	Finally, if $\phi_2(\DefCnbhd{\phi_3}) \cap S^1 \neq \emptyset$, then we have $(\phi_2 \circ \phi_3)(z_0) = \phi_2(\phi_3(z_0)) \in S^1$ for some $z_0 \in S^1$.
	\qedhere
\end{enumerate}
\end{proof}

The third property, associativity, only applies to the composition of three complex deformations.
This is called ``local associativity'' in the context of local Lie groups~\cite{Olver:Non-associative_local_Lie_group}.
It does not necessarily follow that compositions of more than three elements would be associative, e.g., if some orders of composition are not defined.
Moreover, for a composable pair $(\phi, \psi) \in \defmultset \cap (\definvset \times \definvset)$ of invertible complex deformations, the inverses may not be composable, and the composition $\phi \circ \psi$ may not be invertible.

\begin{remark}
$\DefC$ is not a local Lie group in the sense of~\cite{Olver:Non-associative_local_Lie_group, Neeb:Infinite-Dimensional_Lie_Groups, Goldbring:Hilberts_fifth_problem_for_local_groups, Gloeckner-Neeb:Infinite-dimensional_Lie_groups}, since this would require that $\definvset$ and $\defmultset$ are open in $\DefC$.
As an open subset of $\mathcal{O}(S^1)$, the set $\DefC$ of complex deformations is an infinite-dimensional manifold modeled on the convenient vector space $\mathcal{O}(S^1)$.
The diffeomorphism group $\Diffpan$ is the submanifold of $\phi \in \DefC$ such that $\phi(S^1) = S^1$.
The subsets $\definvset \subset \DefC$ and $\defmultset \subset \DefC \times \DefC$, however, are neither open nor closed in $\DefC$, and therefore do not come with a manifold structure.
Indeed, for any $\phi \in \DefC \setminus \Diffpan$, there exists a smooth curve through $\phi$ that immediately exits $\DefC$ by introducing a singularity inside $\DefCnbhd{\phi}$.
Likewise, a smooth curve through the composable pair $(\phi, \psi) \in \defmultset$ with $\psi \notin \Diffpan$ may immediately introduce a singularity of $\phi$ inside $\DefCnbhd{\psi}$.
On the other hand, if $\phi \in \DefC \setminus \definvset$ is not invertible due to singularities of $\phi^{-1}$ only on $S^1$, there exists a smooth curve in $\DefC$ through $\phi$ which immediately moves these singularities out of $\DefCnbhd{\phi}$, making the complex deformation invertible.
This shows that $\definvset$ is not closed in $\DefC$, and the same argument also shows that $\defmultset$ is not closed in $\DefC \times \DefC$.
\label{remark:not_open_closed}
\end{remark}

Still, inversion and composition are Fr{\"o}licher smooth (Fr-smooth), in the sense of Definition~\ref{def:fr} in Appendix~\ref{section:frolicher}, and moreover, we find the following properties near the identity.

\begin{proposition}
	The following properties hold with respect to the subset Fr{\"o}licher structures on $\definvset \subset \DefC$ and $\defmultset \subset \DefC \times \DefC$.
	\label{prop:defc_local}
	\begin{enumerate}
		\item
		The inversion law of $\DefC$ defined by Equation~\eqref{eq:def_inv_mult} is Fr-smooth. 
		Moreover, for any curve $\gamma \in \curves{\DefC}$ such that $\gamma(0) = \id$, there exists $\varepsilon > 0$ such that $\gamma(t) \in \definvset$ for all $t \in (-\varepsilon, \varepsilon)$.
		\item
		The composition law of $\DefC$ defined by Equation~\eqref{eq:def_inv_mult} is Fr-smooth. Moreover, for any curve $(\gamma, \eta) \in \curves{\DefC} \times \curves{\DefC}$ such that $\gamma(0) = \eta(0) = \id$, there exists $\varepsilon > 0$ such that $(\gamma(t), \eta(s)) \in \defmultset$ for all $t, s \in (-\varepsilon, \varepsilon)$.
	\end{enumerate}
\end{proposition}

\begin{proof}
The Fr{\"o}licher structures on $\DefC$, $\definvset$, and $\defmultset$ are generated by those curves $\gamma \in \curves{\mathcal{O}(S^1)}$ which take values in the respective subset of $\mathcal{O}(S^1)$.
Thus, given an initial curve $\gamma$ in $\definvset$, for fixed $t \in \R$, inversion yields a function $\gamma^{-1}(t, z)$, which is smooth in $(t, z)$ and analytic for $z \in \DefCnbhd{\phi}$ with real-analytic boundary behavior in particular for $z \in S^1$.
Since the inverse depends smoothly on $t \in \R$, we find that inversion is Fr-smooth.
Similarly, given initial curves $(\gamma, \eta)$ in $\defmultset$, the function $\gamma(t, \eta(t, \blank)) \in \DefC$ depends smoothly on $t \in \R$, showing that composition is Fr-smooth.

By~\cite[Theorem~8.4]{Kriegl-Michor:Convenient_setting_of_global_analysis}, the inductive limit topology on $\mathcal{O}(S^1)$ is regular.
Thus, for each $\varepsilon > 0$, there exists $n > 0$ such that the restriction of $\gamma$ to the bounded set $(-\varepsilon, \varepsilon)$ takes values in $\mathcal{O}(U_n)$ with $U_n$ as in Equation~\eqref{eq:circlenbhd}.
By further increasing $n$, we can assure that $\gamma(t, \blank)$ is univalent and without zeros on $U_n$ for each $t \in (-\varepsilon, \varepsilon)$.
Fixing such $n$, we can decrease $\varepsilon$ further such that $S^1 \subset \gamma(t, U_n)$ for each $t \in (-\varepsilon, \varepsilon)$, making $\gamma$ invertible in a neighborhood of $t = 0$.
Alternatively, we can decrease $\varepsilon$ such that $\eta(s, S^1) \subset U_n$ for each $s \in (-\varepsilon, \varepsilon)$, making $\gamma$ and $\eta$ composable in a neighborhood of $t = s = 0$.
\end{proof}

\subsection{Lie algebra and flow equations}

The tangent space $T_{\id} \DefC$ at the identity map $\id$ in the Fr{\"o}licher space setup~\eqref{eq:curvaceous_tangent} discussed in Appendix~\ref{section:frolicher} consists of equivalence classes $[\gamma]_\sim$ of curves $\gamma \in \curves{\DefC}$ rooted at $\gamma(0) = \id$.
Because by Proposition~\ref{prop:defc_local}, the curves rooted at $\id$ become invertible and composable for small enough time, for the Fr{\"o}licher spaces $\definvset$ and $\defmultset$ we have
\begin{equation}
	T_{\id} \definvset = T_{\id} \DefC, \qquad T_{\id} \defmultset = T_{\id} (\DefC \times \DefC).
\end{equation}
Using these properties, we equip $T_{\id} \DefC$ with a Lie algebra structure using a construction that was carried out by Laubinger for Fr{\"o}licher--Lie groups (see~\cite[Theorem~3.12]{Laubinger:Lie_algebra_for_Frolicher_groups}).
The condition for such a construction is that the injective map
\begin{equation}
\begin{aligned}
	\Xi \colon T_{\id} \DefC &\longhookrightarrow T_0 T_{\id} \DefC, \\
	[\gamma]_\sim &\longmapsto \big[t \mapsto t[\gamma]_\sim\big]_\sim ,
\end{aligned}
\label{eq:laubinger_condition}
\end{equation}
is bijective.
Then, the Lie bracket of tangent vectors $[\gamma]_\sim, [\eta]_\sim \in T_{\id} \DefC$ is
\begin{equation}
	\big[[\gamma]_\sim, [\eta]_\sim \big] = \Xi^{-1}\Big(\big[s \mapsto [t \mapsto \gamma(s) \circ \eta(t) \circ \gamma^{-1}(s) \circ \eta^{-1}(t)]_\sim\big]_\sim\Big).
\label{eq:lie_bracket_defc}
\end{equation}
The fact that $\DefC$ is an open submanifold of the convenient vector space $\mathcal{O}(S^1)$ implies that the Fr{\"o}licher structure on $T_{\id} \DefC$ is isomorphic to $\mathcal{O}(S^1)$ (note that $\mathcal{O}(S^1)$ is a vector space, whence its tangent space is isomorphic to itself), which yields bijectivity of $\Xi$.
In turn, we can identify $\mathcal{O}(S^1)$ with $\VectC$ using Equation~\eqref{eq:vectc_os1}.

Given the Fr{\"o}licher structure on $\VectC$ defined by the identification~\eqref{eq:vectc_os1} with $\mathcal{O}(S^1)$, a smooth curve $v \in \curves{\VectC}$ is a smoothly time-dependent real-analytic complex-valued vector field on $S^1$.
Similar to~\eqref{eq:vectc_os1}, we denote it by ${v = v(t, z) \partial_z}$, such that the coordinate function is smooth on $\R \times S^1$ and real-analytic in $z$ for each fixed~$t$.
Using analytic continuations of $v$ to neighborhoods of $S^1$ in $\C$, the right-trivializing flow equation is given by
\begin{equation}
	\dot{\Phi}_{v}(t, z) = v(t, \Phi_{v}(t, z)), \qquad \Phi_{v}(0, z) = z.
\label{eq:flow}
\end{equation}
The solution $\Phi_v$ is the flow of $v$.
It exists on $(-\varepsilon, \varepsilon) \times S^1$ for some $\varepsilon > 0$, and is smooth in the first and real-analytic in the second coordinate.
On the other hand, the left-trivializing flow equation of a time-dependent vector field $w \in \curves{\VectC}$ is given by
\begin{equation}
	\dot{\Phi}^w(t, z) = (\Phi^w)'(t, z) \; w(t, z), \qquad \Phi^{w}(0, z) = z.
\label{eq:flow_right}
\end{equation}
In this case, the solution $\Phi^w$ exists for all time as a curve $\Phi^w \in \curves{\mathcal{O}(S^1)}$.
Given vector fields $v, w \in \curves{\VectC}$ such that $\Phi_v = \Phi^w$, they are related by the pullback~\eqref{eq:pullback_vector_field} as $w = \Phi_v^* v$.
Conversely, given a curve $\gamma \in \curves{\mathcal{O}(S^1)}$ rooted at $\gamma(0) = \id$, the vector fields defined by
\begin{equation}
	v(t, z) = \dot{\gamma}(t, \gamma^{-1}(t, z)), \qquad w(t, z) = \frac{\dot{\gamma}(t, z)}{\gamma'(t, z)},
\label{eq:vect_of_curve}
\end{equation}
satisfy $\gamma = \Phi_v = \Phi^w$.
Note that we have $v(0, z) = w(0, z)$, and thus, the associated tangent vectors agree, that is, $[\gamma]_\sim = [t \mapsto \Phi_v(t,\blank)]_\sim = [t \mapsto \Phi^w(t,\blank)]_\sim \in T_{\id} \DefC$.
Since $\Phi^w$ is defined for all time, we may use it for the following identification of Lie algebras: 

\begin{proposition}
The map
\begin{equation}
\begin{aligned}
	\VectC &\longrightarrow T_{\id} \DefC , \\
	w &\longmapsto [t \mapsto \Phi^{w}(t,\blank)]_\sim \, = \, [t \mapsto \Phi_{w}(t,\blank)]_\sim ,
\end{aligned}
\end{equation}
is a Fr-smooth Lie algebra isomorphism.
\label{prop:defc_algebra}
\end{proposition}

\begin{proof}
This is an immediate consequence of the expression of the Lie bracket on $\VectC$ in terms of the left-trivializing flow given by
\begin{equation}
	[v, w] = \eval{\pdv{}{t}{s}}_{t=s=0}
	\Phi^v(t, \blank)
	\circ \Phi^w(s, \blank)
	\circ (\Phi^v(t, \blank))^{-1}
	\circ (\Phi^w(s, \blank))^{-1}.
\end{equation}
\end{proof}

\begin{remark}
	Note that the exponential map $v \mapsto \Phi^v(1, \blank)$ is not surjective, just like for $\Diffpan$, by a well-known argument (see~\cite[Warning~1.6]{Milnor:Remarks_on_infinite_dimensional_Lie_groups}).
	However, via Equation~\eqref{eq:vect_of_curve}, any complex deformation is the $t = 1$ flow of some time-dependent vector field $v \in \curves{\VectC}$.
	Actually, real-analytic diffeomorphisms that can be embedded in a flow are extremely rare, forming a meager (first category) set~\cite{Arnold:Small_denominators_I_Mappings_of_the_circumference_onto_itself, Herman:Sur_la_conjugaison_differentiable_des_diffeomorphismes_du_cercle_a_des_rotations, Yoccoz:Analytic_linearization_of_circle_diffeomorphisms}.
\end{remark}

\begin{remark}
The tangent spaces of $\definvset$ and $\defmultset$, respectively at $\phi \in \definvset$ and $(\phi, \psi) \in \defmultset$, may also be identified with spaces of vector fields.
In the case of $T_\phi \definvset$, there are two options.
On the one hand, we may identify it with vector fields $v \in \vect(\DefCnbhd{\phi})$.
In that case, the left-trivializing flow is holomorphic in $\DefCnbhd{\phi}$, and we find the tangent vector $[t \mapsto \Phi^v(t, \blank) \circ \phi]_\sim \in T_\phi \definvset$.
On the other hand, a vector field $w \in \vect(\DefCnbhd{\phi^{-1}})$ yields a tangent vector $[t \mapsto \phi \circ \Phi^w(t, \blank)]_\sim \in T_\phi \definvset$, and it is related to $v$ by pullback as $\phi^* v = w$.
For $T_{(\phi, \psi)} \defmultset$, a pair of vector fields $v, w \in \vect(\DefCnbhd{\psi})$ defines left-trivializing flows such that $\phi$ is composable with $\Phi^v$, and $\Phi^w$ is composable with $\psi$. Moreover, $\phi \circ \Phi^v(t, \blank)$ is still composable with $\Phi^w(t, \blank) \circ \psi$, yielding a tangent vector $[t \mapsto (\phi \circ \Phi^v(t, \blank), \Phi^w(t, \blank) \circ \psi)]_\sim \in T_{(\phi, \psi)} \defmultset$.
\label{rem:tangent_spaces}
\end{remark}

\subsection{Special cases and decompositions}
\label{section:subgroups}

In the following, we list various interesting types of complex deformations.
\begin{enumerate}
	\item \emph{Rotations.}
	The complex deformations $\rotation{\alpha}(z) = e^{\ii \alpha} z$ for $\alpha \in \R$ form a one-dimensional subgroup of $\DefC$. It is isomorphic to $S^1$ and generated by the vector field $\ii \ell_0$.
	\item \emph{Scaling transformations.}
	The complex deformations,
	\begin{equation}
		\scalinggroup = \setsuchthat{\scaling{\tau}(z) = e^{-2\pi \tau} z}{\tau \in \R},
	\label{eq:def_shrinkgroup}
	\end{equation}
	form a one-dimensional subgroup isomorphic to $\R$ generated by the vector field $\ell_0$.
	\item \emph{M{\"o}bius transformations.}
	A M{\"o}bius transformation $F \in \mob$ is a complex deformation if and only if it preserves the winding number of $S^1$ around $0$.
	In particular, the inversion $\inversion$ defined in Equation~\eqref{eq:inversion} is not a complex deformation.
	However, the conjugation by $\inversion$ defines an involution $\phi \mapsto \inversion \circ \phi \circ \inversion$ on $\DefC$.
	\item \emph{Diffeomorphisms.}
	The	group $\Diffpan$ of real-analytic diffeomorphisms of the unit circle is a subgroup of $\DefC$.
	At the identity, it corresponds to the Lie subalgebra $\VectR$ of $\VectC$ comprising tangent vector fields (Equation~\eqref{eq:ellpar}).

	\item \emph{Univalent functions.}
	Any univalent function $F \colon \cdisk \to \C$ with real-analytic boundary behavior and with $0 \in F(\{z < |1|\})$, is a complex deformation.
	These are generated by vector fields $v \in \curves{\VectC}$, which are holomorphic on the unit disk for each $t$.
	By the Koebe $1/4$-theorem, we see that if $|F'(0)| \geq 4$, then such a univalent function is invertible as a complex deformation.
	We denote the subset of univalent functions such that $F(0) = 0$ by $\Univ \subset \DefC$.

	\item \emph{(Strict) contractions~\&~expansions.}
	Complex deformations $\phi \in \DefC$ such that $\phi(S^1) \subset \setsuchthatinline{z \in \C}{|z| < 1}$ (respectively $\phi(S^1) \subset \setsuchthatinline{z \in \C}{|z| > 1}$) contract (expand) the unit circle in the complex plane.
	They have the property that the set $\DefCnbhd{\phi}$ in~\eqref{eq:def_defcnbhd} is a proper annulus.
	They are generated by inward (outward) pointing vector fields, also coined ``Markovian'' vector fields in~\cite{BGKRV:Virasoro_structure_and_the_scattering_matrix_for_Liouville_CFT} due to their role in the Markov property of Liouville CFT.
\end{enumerate}
Some of these types of complex deformations are related by the following decomposition, which will be used in Section~\ref{section:list_cocycles} to define certain 2-cocycles on $\DefC$.

\begin{proposition}
	Any complex deformation $\phi \in \DiffC$ may be uniquely decomposed into a diffeomorphism $\cdiffeo{\phi} \in \Diffpan$ and a univalent function $\cunivalent{\phi} \in \Univ$, normalized by $\cunivalent{\phi}(0) = 0$ and $\cunivalent{\phi}'(0) > 0$, such that
	\begin{equation}
		\phi = \cunivalent{\phi} \circ \cdiffeo{\phi}.
	\label{eq:decomp}
	\end{equation}
	Moreover, the dependence of $\cunivalent{\phi}$ and $\cdiffeo{\phi}$ on $\phi$ is Fr-smooth.
	\label{prop:decomp_def}
\end{proposition}
\begin{proof}
By the definition of a complex deformation, the real-analytic loop $\phi(S^1)$ surrounds $0$.
Let $U$ be the closed domain with positively oriented boundary $\phi(S^1)$ and $\cunivalent{\phi} \colon \cdisk \to U$ the Riemann mapping uniquely determined by $\cunivalent{\phi}(0) = 0$ and $\cunivalent{\phi}'(0) > 0$.
As the boundary of $U$ is real-analytic, $\cunivalent{\phi}$ is complex-analytic in a neighborhood of the disk.
The composition $\cdiffeo{\phi} = \cunivalent{\phi}^{-1} \circ \phi$ is a real-analytic diffeomorphism of $S^1$.
We precompose it with $\cunivalent{\phi}$ to obtain the desired decomposition~\eqref{eq:decomp}.
Since $\cunivalent{\phi}$ is a Riemann mapping of a domain with smoothly time-dependent smooth boundary, $\cunivalent{\phi}$ itself is Fr-smooth with respect to a smooth deformation of the boundary (see e.g.,~\cite[Theorem~28.1]{Bell:Cauchy_transform_potential_theory_and_conformal_mapping}) and Fr-smoothness of $\cdiffeo{\phi}$ follows since composition is Fr-smooth by Proposition~\ref{prop:defc_local}.
\end{proof}

\section{Cohomology of complex deformations}
\label{section:cohomology}

In this section, we compute the second group cohomology group of the complex deformations of the unit circle defined in Section~\ref{section:complex_deformations}.
More precisely, we consider the analogue of (Fr-smooth) group cohomology with respect to the partially defined composition operation on $\DefC$ with coefficients in $\R$.
The main result of this section is Theorem~\ref{thm:cohomology}.

The cohomology may also be relative to finitely many subgroups $H_1, \ldots, H_N$ such as those listed in Section~\ref{section:subgroups}.
We are mainly interested in two subgroups:
The diffeomorphism group $H_1 = \Diffpan$, and the group $H_2 = \shrinkgroup$ of scaling transformations defined in~\eqref{eq:def_shrinkgroup}.
The relative and non-relative cohomology groups form exact sequences, which fit into the braided diagram in Figure~\ref{fig:diag_grp}.
In Section~\ref{section:computation_cohomology}, we shall compute the terms of this diagram.
In addition to the braided diagram, our method involves an exact sequence relating the (relative) cohomology groups of $\DefC$ to the respective Lie algebra cohomology, and to characters of the fundamental groups.
A basis of Lie group and Lie algebra cocycles is presented beforehand, in Section~\ref{section:list_cocycles}.
In Section~\ref{section:cft}, we apply the results to our previous work~\cite{Maibach-Peltola:From_the_conformal_anomaly_to_the_Virasoro_algebra}.
Let us begin with the setup.

\subsection{Relative group and algebra cohomology}

For concreteness, we define the relative cohomology groups of $G = \DefC$ in analogy to the bar resolution of group cohomology~\cite{Weibel:An_introduction_to_homological_algebra}.
Let $\sgrps = (H_1, \ldots, H_N)$ be subgroups of $G$ such that $H_j \cap H_k = \{\id\}$ for $1 \leq j < k \leq N$.
The $n$-cochains on $G$ relative to $\sgrps$ with coefficients in $\R$ are the Fr-smooth functions (in the sense of Definition~\ref{def:fr} in Appendix~\ref{section:frolicher})
\begin{equation}
	C^n(G; \sgrps; \R) = \setsuchthat{\Cocycle \in \functions{\compn{G}{n}}}{\restrict{\Cocycle}{H^n_j} = 0 \textnormal{ for all } 1 \leq j \leq N} ,
\end{equation}
where $\compn{G}{n}$ is the subset of tuples in $(g_1, \ldots, g_n) \in G^n$ such that products $g_{j} g_{j+1} \dots g_{j+k}$ of any number of consecutive elements exist, for $1 \leq j \leq n$ and $1 \leq k \leq n - j$.
In particular, the set $\compn{G}{2} = \defmultset$ comprises pairs of composable deformations. 

The non-relative cochains are obtained as $C^n(G; \R) = C^n(G; \{\id\}; \R)$.
By forgetting a single subgroup $H_j$, leaving the rest $\underline{\hat{H}}$, we have a short exact sequence
\begin{equation}\label{eq:grp_ses}
\begin{alignedat}{2}
	\{0\} \; \longrightarrow \; C^n(G; \sgrps; \R) \; & \longhookrightarrow \; C^n(G; \underline{\hat{H}}; \R) && \; \longtwoheadrightarrow \; C^n(H_j; \R) \; \longrightarrow \; \{0\} ,
\end{alignedat}
\end{equation} 
where the maps are defined by restriction of the cochains.
As usual, the group cohomology differential $\Dgrp \Cocycle \in C^{n+1}(G; \sgrps; \R)$ of a cochain $\Cocycle \in C^n(G; \sgrps; \R)$ is defined by
\begin{equation}\label{eq:dgrp}
\begin{aligned} 
	(\Dgrp \Cocycle)(g_1, \ldots, g_{n+1})
	= \; & \Cocycle(g_2, \ldots, g_{n+1})
	+ (-1)^{n+1} \Cocycle(g_1, \ldots, g_{n}) \\
	&+ \sum_{j = 1}^n (-1)^j \,  \Cocycle(g_1, \ldots, g_{j - 1}, g_{j}g_{j+1}, g_{j + 1}, \ldots, g_{n+1}).
\end{aligned} 
\end{equation} 
The relative group cohomology groups are then defined by
\begin{align}
	Z^n(G; \sgrps; \R) = \; &
	\kernel \Dgrp, \\
	B^n(G; \sgrps; \R) = \; &
	\image \Dgrp, \\
	H^n(G; \sgrps; \R) = \; & Z^n(G; \sgrps; \R) / B^n(G; \sgrps; \R).
\end{align}
The short exact sequence~\eqref{eq:grp_ses} leads to a long exact sequence of cohomology groups:
\begin{equation}
\begin{aligned}
	\cdots
	\; &
	\longrightarrow H^{n - 1}(H_j; \R)
	\longrightarrow H^n(G; \sgrps; \R)
	\longrightarrow H^n(G;\underline{\hat{H}} ; \R)
	\\
	\; &
	\longrightarrow H^n(H_j; \R)
	\longrightarrow H^{n + 1}(G;\sgrps ; \R)
	\longrightarrow \cdots
\end{aligned}
\end{equation}
with the transgression maps defined by the usual zig-zag lemma \cite[Section~24]{Munkres:Elements_of_algebraic_topology}.

Let $\mathfrak{g} = \VectC$, and denote by $\mathfrak{h}_1, \ldots, \mathfrak{h}_N \subset \mathfrak{g}$ the Lie algebras of the subgroups $H_1, \ldots, H_N$ of $G$.
Let $H^n(\mathfrak{g}; \salgs; \R)$ denote the relative Lie algebra cohomology~\cite{Chevalley-Eilenberg:Cohomology_theory_of_Lie_groups_and_Lie_algebras, Fuchs:Cohomology_of_infinite-dimensional_Lie_algebras, Borel-Wallach:Continuous_cohomology_discrete_subgroups_and_representations_of_reductive_groups}, and
$\Dalg \colon C^{n}(\mathfrak{g}; \salgs; \R) \to C^{n+1}(\mathfrak{g}; \salgs; \R)$ the differential. 
The derivative ${\DVE \colon H^2(G; \sgrps; \R)} \to H^2(\mathfrak{g}; \salgs; \R)$, also called \emph{van Est map}~\cite{Van_Est:Group_cohomology_and_Lie_algebra_cohomology_in_Lie_groups1, Van_Est:Group_cohomology_and_Lie_algebra_cohomology_in_Lie_groups2}, 
\begin{equation}
	(\DVE \Cocycle)(v, w)
	= \frac{1}{2}
	\eval{\pdv{s}}_{s = 0}
	\eval{\pdv{t}}_{t = 0}
	\bigg(
	\Cocycle\big(\Phi_v(t, \blank), \Phi_w(s, \blank)\big)
	- \Cocycle\big(\Phi_w(s, \blank), \Phi_v(t, \blank)\big)
	\bigg),
\label{eq:def_dlie}
\end{equation}
is obtained\footnote{
	We put a prefactor $\smash{\frac{1}{2}}$ in Equation~\eqref{eq:def_dlie} as in~\cite{Maibach-Peltola:From_the_conformal_anomaly_to_the_Virasoro_algebra}, 
and as opposed to~\cite[Proposition~3.14]{Khesin-Wendt:The_geometry_of_infinite-dimensional_groups}.
} using the flow equation~\eqref{eq:flow} and the fact from Proposition~\ref{prop:defc_local} that the flows are always composable for small enough time.
Note that for coboundaries, we have
\begin{equation}
	(\DVE \Dgrp f)(v, w) = -\frac{1}{2}
	\eval{\pdv{t}}_{t = 0}
	f\big(\Phi_{[v,w]}(t, \blank)\big),
	\qquad f \in C^1(G; \sgrps; \R).
\label{eq:vE_grp_alg_dif}
\end{equation}

In a finite-dimensional setting with sufficiently connected Lie groups, van Est maps $\DVE$ such as~\eqref{eq:def_dlie} are often isomorphisms; for examples, see~\cite[Theorem~22.1]{Chevalley-Eilenberg:Cohomology_theory_of_Lie_groups_and_Lie_algebras},~\cite{Van_Est:Group_cohomology_and_Lie_algebra_cohomology_in_Lie_groups1, Van_Est:Group_cohomology_and_Lie_algebra_cohomology_in_Lie_groups2}, and~\cite[Theorem~3.6.1]{Mostow:Cohomology_of_topological_groups_and_solvmanifolds}. 
For the complex deformations $G = \DefC$, we encounter two obstructions.
On the one hand, the fundamental group of $G$ is isomorphic to $\Z$ due to the winding property around $0$ in Definition~\ref{def:defc}.
On the other hand, $\DefC$ is infinite-dimensional and not quite a group, whence none of the standard results apply.
Despite these issues, we find that the relative Lie group and Lie algebra cohomologies can only differ for topological reasons.

In Proposition~\ref{prop:ses_hom_grp_alg} below, we prove a result specifically for $G = \DefC$, which is conceptually close to an exact sequence used by 
Neeb~\cite[Theorem~7.2]{Neeb:Abelian_extensions_of_infinite-dimensional_Lie_groups} 
(which holds in a quite general setting of infinite-dimensional Lie groups).
However, Neeb's setup does not fit our setting of Fr{\"o}licher structures, partially defined composition laws, and relative cohomology --- so we adapt his proof. 
While partially extending Neeb's result, our result is at the same time simplified, since our coefficients are in $\R$ regarded as a trivial $G$-module, as opposed to a possibly infinite-dimensional non-trivial module in Neeb's theorem.

The obstruction coming from the fundamental group $\pi_1(G, \id) \cong \Z$ is overcome by initially integrating Lie algebra coboundaries on the universal cover.
Since $G = \DefC$ is an open submanifold of $\mathcal{O}(S^1)$, the universal cover of $\DefC$ exists~\cite[Paragraph~27.14]{Kriegl-Michor:Convenient_setting_of_global_analysis}, and we denote it by $\universalcover{G}$.
In particular, we can use the path-construction of the universal cover~\cite{Fulton:Algebraic_topology}, 
where $\universalcover{G}$ is defined as the set of homotopy classes of (Fr-smooth) paths in $G$ rooted at the identity.
Then, $\universalcover{G}$ comes with a natural Fr{\"o}licher structure given by lifts of curves $\curves{G}$, and the projection $\universalcover{G} \twoheadrightarrow G$ is Fr-smooth. 
This construction yields an embedding $\pi_1(G, \id) \subseteq \universalcover{G}$ of the (Fr-smooth) fundamental group as the group of equivalence classes of paths returning to the identity.

Note that in our infinite-dimensional setting, Lie's third theorem, which allows integration of Lie algebra cocycles to the group level, might not hold.
Instead, we initially assume in Proposition~\ref{prop:ses_hom_grp_alg} that we are able to do so, and then in the next Section~\ref{section:list_cocycles} we indeed find a cocycle on $G$ for each Lie algebra cohomology class by explicit construction.

\begin{proposition}
Let $G = \DefC$, and $\sgrps, \mathfrak{g}, \salgs$ as above.
If the van Est map~\eqref{eq:def_dlie} is surjective, then the following sequence is exact:
\label{prop:ses_hom_grp_alg}
\begin{equation}\label{eq:vE_exact_seq}
\begin{alignedat}{2}
	\{0\} \; \longrightarrow \; \Hom(\pi_1(G; \sgrps), \R) \; & \overset{\funtogrp}{\longhookrightarrow} \; H^2(G; \sgrps; \R) && \; \overset{\DVE}{\longtwoheadrightarrow} \; H^2(\mathfrak{g}; \salgs; \R) \; \longrightarrow \; \{0\} ,
\end{alignedat}
\end{equation} 
where $\funtogrp(\gamma)$ is the cohomology class of the central extension defined by
\begin{equation}
	(\universalcover{G} \times \R) / \setsuchthat{(\alpha, \gamma(\alpha))}{\alpha \in \pi_1(G; \sgrps)} , \qquad \gamma \in \Hom(\pi_1(G; \sgrps), \R) .
\label{eq:ext_uc}
\end{equation}
\end{proposition}

\begin{proof}
The central extension of any subroup $H_j$, which is the quotient of $\universalcover{H_j} \times \R$ by $\setsuchthat{(\alpha, \gamma(\alpha))}{\alpha \in \pi_1(G; \sgrps)}$ is the trivial central extension $H_j \times \R$, so $\funtogrp(\gamma)$ is indeed a relative cohomology class.
It is Fr-smooth since the quotient~\eqref{eq:ext_uc} is formed by a discrete subgroup.
As the Lie algebra cocycle of a central extension of $G$ only depends on derivatives along Fr-smooth curves at the identity, it only depends on $\funtogrp(\gamma)$ close to the identity, where it is coboundary. 
Thus, $\DVE(\funtogrp(\gamma)) = 0$. Observe also that $\funtogrp(\gamma) = 0$ if and only if $\gamma = 0$ (i.e., the factor $\R$ is unaffected by the relation in~\eqref{eq:ext_uc}), so by linearity, $\funtogrp$~is an injection.
This proves the first half of the exactness of the sequence in~\eqref{eq:vE_exact_seq}.  

Next, we integrate a Lie algebra coboundary to the universal cover of $G$ with the goal of constructing the corresponding function $\gamma \in \Hom(\pi_1(G; \sgrps), \R)$.
Let $\Cocycle \in H^2(G; \sgrps; \R)$ represent any group-level cocycle such that the Lie algebra cocycle $\DVE \Cocycle$ is a coboundary, that is, $\DVE \Cocycle = \Dalg F$ 
for some dual vector\footnote{
	In the case where a Fr{\"o}licher space $X$ also comes with an $\R$-vector space structure, we denote its Fr-smooth dual by $X^\vee = \setsuchthat{F \in \functions{X}}{\textnormal{$F$ is linear}}$.
} $F \in \mathfrak{g}^\vee$ vanishing on $\mathfrak{h}_1^\vee, \ldots, \mathfrak{h}_N^\vee$:
\begin{equation}
	(\DVE \Cocycle)(v, w) 
	= v \, F(w) - w \, F(v) - F([v,w]) = - F([v,w]) , \qquad v, w \in \mathfrak{g} .
\end{equation}
(Here, we have $v \, F(w) = w \, F(v) = 0$ because we consider cohomology with coefficients in $\R$ regarded as a trivial $\mathfrak{g}$-module.)
Consider the Lie algebra central extension $\mathfrak{g} \times_{\DVE\Cocycle} \R$ as an exact sequence split by $F$,
\begin{equation}
\begin{tikzcd}[ampersand replacement=\&]
{\{0\}} \& \R \& {\mathfrak{g} \times_{\DVE\Cocycle} \R} \& {\mathfrak{g}} \& {\{0\}} ,
\arrow[from=1-1, to=1-2]
\arrow[from=1-2, to=1-3]
\arrow["p", bend left, from=1-3, to=1-2, shift left = 2, bend left, start anchor = west, end anchor = east]
\arrow[from=1-3, to=1-4]
\arrow["{\! (v, F(v)) \; \mapsfrom \; v}", from=1-4, to=1-3, shift left = 2, bend left, start anchor = west, end anchor = east]
\arrow[from=1-4, to=1-5]
\end{tikzcd}
\end{equation}
where $p \colon (v, a) \mapsto a - F(v)$ is an $\R$-linear map, $p \in (\mathfrak{g} \times_{\DVE\Cocycle} \R)^\vee$, such that its kernel $\kernel p = \setsuchthat{(v, a)}{F(v) = a}$ is the graph of $F$.
This $p \in (\mathfrak{g} \times_{\DVE\Cocycle} \R)^\vee$ defines a Lie algebra $1$-cochain on the central extension.
Taking into account that $F$ and thus also $p$ vanish on the subalgebras $\mathfrak{h}_j \times \{0\}$, it actually is a relative $1$-cochain $p \in C^1(\mathfrak{g} \times_{\DVE\Cocycle} \R; \salgs \times \{0\}; \R)$ on the Lie algebra central extension.
We apply the Lie algebra differential $\Dalg$ to $p$:
\begin{equation}
\begin{aligned}
	(\Dalg p)((v, a), (w, b)) 
	= \; & - p([(v, a), (w, b)]) \\
	= \; & - p \big( [v, w], (\DVE \Cocycle)(v, w) \big)
	= 0 , \qquad v, w \in \mathfrak{g}, \, a,b \in \R .
\end{aligned}
\end{equation}
Thus, the map $p$ is actually a relative $1$-cocycle: $p \in Z^1(\mathfrak{g} \times_{\DVE \Cocycle} \: \R; \salgs \times \{0\}; \R)$.

Next, the universal cover of the group central extension $G \times_{\Cocycle} \R$ may be identified as
\begin{equation}
	\universalcover{G \times_{\Cocycle} \R} = \universalcover{G} \times_{\hat \Cocycle} \R,
\end{equation}
where the lifted cocycle $\hat \Cocycle$ is defined by pullback, that is, $\hat \Cocycle(\hat \phi, \hat \psi) = \Cocycle(\phi, \psi)$ if $\hat \phi$ and $\hat \psi$ are lifts of $\phi, \psi \in G$ respectively.
We now proceed to define an invariant differential $1$-form $\alpha_p$ of $p$ on this universal cover by evaluating the pullback on the cotangent vector $p \in T^{\id} (\universalcover{G} \times_{\hat{\Cocycle}} \R)$ as in Equation~\eqref{eq:eval_dual} in Appendix~\ref{section:frolicher}:
\begin{equation}
	\alpha_p(v) = \evaluation \big(p, (\hat{\phi}, \lambda)^*([\gamma]_\sim) \big), \qquad [\gamma]_\sim \in T_{(\hat{\phi}, \lambda)} (\universalcover{G} \times_{\hat \Cocycle} \R).
\end{equation}
However, the definition of the pullback requires some more explanation, because $(\hat{\phi}, \lambda)$ may not be invertible, and even if it is, the tangent vector may not have a representative curve $\gamma(t)$ rooted at $(\hat{\phi}, \lambda)$ which stays composable with the inverse for any small time $t \neq 0$.
Yet, at the level of $G$, any tangent curve $\gamma$ rooted at $\phi$ defines a vector field $v = \dot{\gamma}(t, \gamma^{-1}(t, z))$ on $\phi(S^1)$ as in Equation~\eqref{eq:vect_of_curve}.
The pullback $\phi^* v$ is well-defined and integrates to a curve rooted at $\id \in G$.
By lifting it to the universal cover of the central extension, this construction shows that the pullback $(\hat{\phi}, \lambda)^*([\gamma]_\sim)$ is well-defined.

Since the Lie algebra cohomology differential $\Dalg p$ and 
the exterior derivative $\DdR \alpha_p$ are defined by the invariant formula~\eqref{eq:invariant_formula}, we have $\DdR \alpha_p = \alpha_{\Dalg p}$. Moreover, since $\Dalg p = 0$, we conclude that the form is closed, that is, $\DdR \alpha_p = \alpha_{\Dalg p} = 0$. 
Since $\universalcover{G} \times_{\hat \Cocycle} \R$ is simply connected, we can apply the Poincar\'e lemma~\cite[Lemma~33.20]{Kriegl-Michor:Convenient_setting_of_global_analysis} to find a unique function $\varphi \in \functions{\universalcover{G} \times_{\hat \Cocycle} \R}$ such that
\begin{equation}
	\DdR \varphi = \alpha_p 
	\qquad \textnormal{and} \qquad
	\varphi(\id) = 0.
\end{equation}
Since the invariant differential $1$-form $\alpha_p$ restricts to the trivial $1$-form on the subgroups $\universalcover{H_j} \times \{0\}$, the function $\varphi$ is trivial on these subgroups as well.
We claim that the function $\varphi$ is a $1$-cocycle on $\universalcover{G} \times_{\hat \Cocycle} \R$ relative to the subgroups, that is, a homomorphism
\begin{equation}
	\varphi\big((\hat \phi, \lambda) \star (\hat \psi, \mu)\big)
	= \varphi (\hat \phi \circ \hat \psi, \lambda + \mu + \Cocycle(\phi, \psi)
	= \varphi (\hat \phi, \lambda) + \varphi (\hat \psi, \mu),
\end{equation}
where $(\hat \phi, \lambda), (\hat \psi, \mu) \in \universalcover{G} \times_{\hat \Cocycle} \R$ and ``$\blank \star \blank$'' denotes the multiplication in the central extension of the universal cover. 
The method to prove this claim is to consider a function
\begin{equation}
\begin{aligned}
	Q \colon \universalcover{G} \times_{\hat \Cocycle} \R &\longrightarrow \R \\
	(\hat \psi, \mu)
	&\longmapsto \varphi\big((\hat \phi, \lambda) \star (\hat \psi, \mu)\big) - \varphi (\hat \psi, \mu) - \varphi (\hat \phi, \lambda).
\end{aligned}
\end{equation}
for fixed $(\hat \phi, \lambda) \in \universalcover{G} \times_{\hat \Cocycle} \R$.
By invariance of $\alpha_p$, the differential of $Q$ vanishes: 
\begin{equation}
	\DdR Q = \DdR\Big(\varphi ((\hat \phi, \lambda) \star \blank) - \varphi(\blank) - \varphi(\hat \phi, \lambda)\Big)
	= \big((\hat \phi, \lambda) \star \blank\big)^* \DdR \varphi - \DdR \varphi = 0.
\end{equation}
Hence, $Q$ is constant, and moreover, the normalization $\varphi(\id) = 0$ implies that $Q \equiv 0$.

Note that $\lambda \mapsto (\id, \lambda) \in \universalcover{G} \times_{\hat{\Cocycle}} \R$ for $\lambda \in \R$ is a linear map, that is, a homomorphism. Consider the derivative
\begin{equation}
	\partial_\lambda\big|_{\lambda = 0} \varphi (\id, \lambda) = \DdR_{(\id, 0)} \varphi (0, 1) = (\alpha_p)_{(\id, 0)}(0, 1) = p(0, 1) = 1 - F(0) = 1.
\end{equation}
Thus, we have $\varphi (\id, \lambda ) = \lambda$, which splits the central extension of the universal cover,
\begin{equation}
\begin{tikzcd}[ampersand replacement=\&]
{\{0\}} \& \R \& {\universalcover{G} \times_{\hat \Cocycle} \R} \& {\universalcover{G}} \& {\{0\}} ,
\arrow[from=1-1, to=1-2]
\arrow[from=1-2, to=1-3]
\arrow["\varphi", bend left, from=1-3, to=1-2, shift left = 2, bend left, start anchor = west, end anchor = east]
\arrow[from=1-3, to=1-4]
\arrow["{\sigma}", from=1-4, to=1-3, shift left = 2, bend left, start anchor = west, end anchor = east]
\arrow[from=1-4, to=1-5]
\end{tikzcd}
\label{eq:uc_split}
\end{equation}
where the section defined by $\sigma(\hat{\phi}) = (\id, -\varphi((\hat \phi, 0))) \star (\hat \phi, 0)$ is a homomorphism:
\begin{equation}
	\sigma(\hat{\phi} \circ \hat{\psi}) =
	\underbrace{(\id, - \varphi(\hat{\phi}, 0))
	\star (\hat{\phi}, 0)}_{\smash{=\;\sigma(\hat{\phi})}}
	\star \underbrace{(\id, - \varphi(\hat{\psi}, 0))
	\star (\hat{\psi}, 0)}_{\smash{=\;\sigma(\hat{\psi})}}
	\star \underbrace{(\id, -\hat{\Omega}(\hat{\phi}, \hat{\psi}))
	\star (\id, \varphi(\id, \Omega(\phi, \psi)))}_{\smash{=\;0}}
\end{equation}

The split central extension~\eqref{eq:uc_split} is isomorphic to the trivial one, $\universalcover{G} \times \R$, with the isomorphism $\Psi(\hat{\phi}, \lambda) = \big(\hat{\phi}, \lambda - \varphi(\hat{\phi}, 0)\big)$, yielding the following morphisms of exact sequences:
\begin{equation}
\begin{tikzcd}[ampersand replacement=\&]
{\{0\}} \& \R \& {G \times_{\Cocycle} \R} \& G \& {\{0\}} \\
{\{0\}} \& \R \& {\universalcover{G} \times_{\hat{\Cocycle}} \R} \& {\universalcover{G}} \& {\{0\}} \\
{\{0\}} \& \R \& {\universalcover{G} \times \R} \& {\universalcover{G}} \& {\{0\}} .
\arrow[from=1-1, to=1-2]
\arrow[from=1-2, to=1-3]
\arrow[from=1-3, to=1-4]
\arrow[from=1-4, to=1-5]
\arrow[from=2-1, to=2-2]
\arrow[equals, from=2-2, to=1-2]
\arrow[from=2-2, to=2-3]
\arrow["P", from=2-3, to=1-3]
\arrow[from=2-3, to=2-4]
\arrow["\Pi", from=2-4, to=1-4]
\arrow[from=2-4, to=2-5]
\arrow[from=3-1, to=3-2]
\arrow[equals, from=3-2, to=2-2]
\arrow[from=3-2, to=3-3]
\arrow["\Psi", from=3-3, to=2-3]
\arrow[from=3-3, to=3-4]
\arrow[equals, from=3-4, to=2-4]
\arrow[from=3-4, to=3-5]
\end{tikzcd}
\end{equation}
The kernel of the composition of $\Psi$ with the projection $P \colon \universalcover{G} \times_{\hat{\Cocycle}} \R \to G \times_{\Cocycle} \R$ is 
\begin{equation}
	\kernel(P \circ \Psi) = \setsuchthat{(\hat \phi, \lambda)}{\Pi(\hat{\phi}) = \id \textnormal{ and } \lambda = \varphi(\hat{\phi}, 0)} .
\end{equation}
Note that this is the graph of the homomorphism $\varphi$ restricted to $\kernel \Pi \times \{0\}$, which may be identified with the fundamental group $\pi_1(G) \subset \universalcover{G}$ as a discrete subgroup of $\universalcover{G} \times \R$, showing that the central extension $G \times_\Cocycle \R$ is indeed a quotient of the form~\eqref{eq:ext_uc}.
\end{proof}

\subsection{Cocycles on complex deformations}
\label{section:list_cocycles}

\begin{figure}[t]
\begin{equation*}
\hspace*{-7mm}
\begin{tikzcd}[
column sep = {8em, between origins},
row sep = {8em, between origins},
ampersand replacement=\&
]
\begin{array}{c} \Hom(\pi_0(\shrinkgroup), \R) \\ = \{0\} \end{array} \&\& \begin{array}{c} \Hom(\pi_1(\DefC; \shrinkgroup), \R) \\ \cong \R \end{array} \&\& \begin{array}{c} \Hom(\pi_1(\Diffpan), \R) \\ \cong \R \end{array} \\
\& \begin{array}{c} \Hom(\pi_1(\DefC; \Diffpan, \shrinkgroup); \R) \\ = \{0\} \end{array} \&\& \begin{array}{c} \Hom(\pi_1(\DefC), \R) \\ \cong \R \end{array} \\
\begin{array}{c} \Hom(\pi_0(\Diffpan), \R) \\ = \{0\} \end{array} \&\& \begin{array}{c} \Hom(\pi_1(\DefC; \Diffpan), \R) \\ = \{0\} \end{array} \&\& \begin{array}{c} \Hom(\pi_1(\shrinkgroup), \R) \\ = \{0\} \end{array}
\arrow[densely dashed, bend left, from=1-1, to=1-3]
\arrow[densely dashed, from=1-3, to=2-4]
\arrow[densely dashed, from=2-4, to=3-5]
\arrow[from=1-1, to=2-2]
\arrow[from=2-2, to=3-3]
\arrow[bend right, from=3-3, to=3-5]
\arrow[densely dashdotdotted, from=2-4, to=1-5]
\arrow[densely dashdotdotted, bend right, from=3-1, to=3-3]
\arrow[densely dashdotdotted, from=3-3, to=2-4]
\arrow[densely dotted, bend left, from=1-3, to=1-5]
\arrow[densely dotted, from=2-2, to=1-3]
\arrow[densely dotted, from=3-1, to=2-2]
\end{tikzcd}
\end{equation*}
\caption{Characters of the relative homotopy groups of complex deformations and the subgroups of diffeomorphisms and scaling transformations.}
\label{fig:diag_hom}
\end{figure}
\begin{figure}[t]
\begin{equation*}
\hspace*{-7mm}
\begin{tikzcd}[
column sep = {8em, between origins},
row sep = {8em, between origins},
ampersand replacement=\&
]
\begin{array}{c} H^1(\shrinkgroup; \R) \\ = \R \Im \frac{\ii}{12} \log \rotCR \end{array} \&\& \begin{array}{c} H^2(\DefC; \shrinkgroup; \R) \\ = \R \Re \bott \oplus \R \Im \bott \\ \oplus \R \Re \rotcocycle \oplus \R \Im \rotcocycle \end{array} \&\& \begin{array}{c} H^2(\Diffpan; \R) \\ = \R \Re \bott \oplus \R \Re \rotcocycle \end{array} \\
\& \begin{array}{c} H^2(\DefC; \Diffpan, \shrinkgroup; \R) \\ = \R \Im \bott \oplus \R \Im \rotcocycle \end{array} \&\& \begin{array}{c} H^2(\DefC; \R) \\ = \R \Re \bott \oplus \R \Im \bott \\ \oplus \R \Re \rotcocycle \end{array} \\
\begin{array}{c} H^1(\Diffpan; \R) \\ = \{0\} \end{array} \&\& \begin{array}{c} H^2(\DefC; \Diffpan; \R) \\ = \R \Im \bott \end{array} \&\& \begin{array}{c} H^2(\shrinkgroup; \R) \\ = \{0\} \end{array}
\arrow[densely dashed, bend left, from=1-1, to=1-3]
\arrow[densely dashed, from=1-3, to=2-4]
\arrow[densely dashed, from=2-4, to=3-5]
\arrow[from=1-1, to=2-2]
\arrow[from=2-2, to=3-3]
\arrow[bend right, from=3-3, to=3-5]
\arrow[densely dashdotdotted, from=2-4, to=1-5]
\arrow[densely dashdotdotted, bend right, from=3-1, to=3-3]
\arrow[densely dashdotdotted, from=3-3, to=2-4]
\arrow[densely dotted, bend left, from=1-3, to=1-5]
\arrow[densely dotted, from=2-2, to=1-3]
\arrow[densely dotted, from=3-1, to=2-2]
\end{tikzcd}
\end{equation*}
\caption{
	Relative Lie group cohomology of complex deformations and the subgroups of diffeomorphisms and scaling transformations.
}
\label{fig:diag_grp}
\end{figure}
\begin{figure}[t]
\begin{equation*}
\hspace*{-7mm}
\begin{tikzcd}[
column sep = {8em, between origins},
row sep = {8em, between origins},
ampersand replacement=\&
]
\begin{array}{c} H^1(\shrinkalg; \R) \\ = \R \Im \rotnumberalg \end{array} \&\& \begin{array}{c} H^2(\VectC; \shrinkalg; \R) \\ = \R \Re \gelfandfuks \oplus \R \Im \gelfandfuks \\ \oplus \R \Im \rotcocyclealg \end{array} \&\& \begin{array}{c} H^2(\VectR; \R) \\ = \R \Re \gelfandfuks \end{array} \\
\& \begin{array}{c} H^2(\VectC; \VectR, \shrinkalg; \R) \\ = \R \Im \gelfandfuks \oplus \R \Im \rotcocyclealg \end{array} \&\& \begin{array}{c} H^2(\VectC; \R) \\ = \R \Re \gelfandfuks \oplus \R \Im \gelfandfuks \end{array} \\
\begin{array}{c} H^1(\VectR; \R) \\ = \{0\} \end{array} \&\& \begin{array}{c} H^2(\VectC; \VectR; \R) \\ = \R \Im \gelfandfuks \end{array} \&\& \begin{array}{c} H^2(\shrinkalg; \R) \\ = \{0\} \end{array}
\arrow[densely dashed, bend left, from=1-1, to=1-3]
\arrow[densely dashed, from=1-3, to=2-4]
\arrow[densely dashed, from=2-4, to=3-5]
\arrow[from=1-1, to=2-2]
\arrow[from=2-2, to=3-3]
\arrow[bend right, from=3-3, to=3-5]
\arrow[densely dashdotdotted, from=2-4, to=1-5]
\arrow[densely dashdotdotted, bend right, from=3-1, to=3-3]
\arrow[densely dashdotdotted, from=3-3, to=2-4]
\arrow[densely dotted, bend left, from=1-3, to=1-5]
\arrow[densely dotted, from=2-2, to=1-3]
\arrow[densely dotted, from=3-1, to=2-2]
\end{tikzcd}
\end{equation*}
\caption{
	Relative Lie algebra cohomology of complex deformations and the subgroups of diffeomorphisms and scaling transformations.
}
\label{fig:diag_alg}
\end{figure}

In this section, we introduce the various cocycles on $\DefC$ and subgroups that appear in the cohomologies in Figures~\ref{fig:diag_grp}~\&~\ref{fig:diag_alg}.
These will be computed in the next section.
All the cocycles are complex-valued at first, and then we take their real and imaginary parts.

The Bott--Thurston cocycle is originally defined as a real-valued $2$-cocycle on the diffeomorphism group of $S^1$, see~\cite{Bott:On_the_characteristic_classes_of_groups_of_diffeomorphisms} and~\cite[Section~4.5]{Guieu-Roger:Virasoro_book}, where it defines the central extension is known as the \emph{Bott--Virasoro group}~\cite{Khesin-Wendt:The_geometry_of_infinite-dimensional_groups, AST:Berezin_quantization_conformal_welding_and_the_Bott-Virasoro_group}.
The Fr-smooth extension to complex deformations is defined by the usual formula
\begin{equation}
	\bott(\phi_1, \phi_2) =
	\frac{1}{24 \pi} \int_{S^1} \log((\phi_1 \circ \phi_2)')(z) \: \dd \log(\phi_2'(z)), \qquad (\phi_1, \phi_2) \in \defmultset.
\label{eq:bott_thurston}
\end{equation}
The associated Lie algebra cocycle, as defined by Equation~\eqref{eq:def_dlie}, reads
\begin{align}
\nonumber
(\DVE \bott)(v, w)
= \; &  \frac{1}{2}
\eval{\pdv{s}}_{s = 0}
\eval{\pdv{t}}_{t = 0}
\Bigg(
\bott(\Phi_v(t, \blank), \Phi_w(s, \blank))
- \bott(\Phi_w(s, \blank), \Phi_v(t, \blank))
\Bigg)
\\ 
= \; & 
\frac{1}{24 \pi} \int_{S^1} v'(z) \; w''(z) \, \dd z , \qquad v,w \in \VectC ,
\label{eq:DVE_bott}
\end{align}
and it is expressed in terms of the generators~\eqref{eq:witt_basis} as
\begin{equation}
	(\DVE \bott)(\ell_n, \ell_m)
	= \frac{\ii}{12} (n^3 - n) \delta_{n + m}, \qquad n, m \in \Z.
\label{eq:lie_bott}
\end{equation}
Up to a coboundary given by the term in Equation~\eqref{eq:lie_bott} which is linear in $n$, this agrees with the Gel'fand--Fuks cocycle~\cite{Gelfand-Fuchs:Cohomologies_of_the_Lie_algebra_of_vector_fields_on_the_circle}.
We express the latter in the $\theta$-coordinate of Equation~\eqref{eq:theta_coords}, where $v = \thetacoords{v}(\theta) \, \partial_\theta$ and $w = \thetacoords{w}(\theta) \, \partial_\theta$.
For $\ell_n = \thetacoords{\ell}_n(\theta) \, \partial_\theta$, we find
\begin{equation}
	\thetacoords{\ell}_n(\theta) = \ii e^{\ii n \theta}, \qquad
	\thetacoords{\ell}_n'(\theta) = - n e^{\ii n \theta}, \qquad
	\thetacoords{\ell}_n''(\theta) = - \ii n^2 e^{\ii n \theta}, \qquad
	n \in \Z.
\end{equation}
In these coordinates, the Gel'fand--Fuks cocycle is defined by either of the formulas
\begin{equation}
	\gelfandfuks(v, w)
	=
	\frac{1}{24\pi}
	\int_0^{2\pi}
	\thetacoords{v}'(\theta)
	\thetacoords{w}''(\theta)
	\; \dd \theta,
	\qquad
	\gelfandfuks(\ell_n, \ell_m)
	=
	\frac{\ii}{12} n^3 \delta_{n + m} .
\label{eq:gelfandfuks}
\end{equation}
The cohomologous cocycles $\DVE \bott$ and $\gelfandfuks$ correspond to the unique central extension of $\Witt$ by $\C$ known as the \emph{Virasoro algebra}~\cite{Gelfand-Fuchs:Cohomologies_of_the_Lie_algebra_of_vector_fields_on_the_circle, Virasoro:Subsidiary_conditions_and_ghosts_in_dual-resonance_models}.

The second complex-valued cocycle that we will consider is precisely the Lie algebra coboundary mentioned above.
On the group level, it locally resembles the differential of the logarithm of a complex-valued function $\rotCR$ on $\DefC$.
Below, in Equation~\eqref{eq:RCR_def}, we utilize the decomposition in Proposition~\ref{prop:decomp_def} of a complex deformation $\phi$ into a diffeomorphism $\cdiffeo{\phi}$ and a univalent function $\cunivalent{\phi}$ to define the argument of $\rotCR$ as the mean rotation $\rotnumberAvg(\cdiffeo{\phi})$ of the diffeomorphism, whereas the absolute value of $\rotCR$ is the reciprocal of the conformal radius $\CR(\cunivalent{\phi})$ of the univalent function.

The \emph{mean rotation} $\rotnumberAvg(\phi)$ of $\phi \in \Diffpan$ is the remainder modulo~$2\pi$,
\begin{equation}
	\rotnumberAvg(\phi) = \arg e^{\ii \rotnumberInt(\hat{\phi})} \in [0,2\pi), \qquad \rotnumberInt(\hat{\phi}) = \frac{1}{2\pi} \int_{0}^{2\pi} (\hat{\phi}(x) - x) \, \dd x.
	\label{eq:rotnumberAvg}
\end{equation}
of the \emph{mean translation} $\rotnumberInt(\hat{\phi})$ of a lift of $\phi$ to a $2\pi$-periodic diffeomorphism $\hat{\phi} \colon \R \to \R$.
The mean rotation does not depend on the choice of lift.
Note that the mean translation is Fr-smooth on the universal cover $\universalcover{\Diffpan}$, and the mean rotation is Fr-smooth except for jumps by $2\pi$ at the branch cut.
In particular, a rotation $\rotation{\alpha}$, as listed in Section~\ref{section:subgroups}, has mean rotation $\rotnumberAvg(\rotation{\alpha}) = \alpha$.
We extend the mean rotation to complex deformations $\phi \in \DefC$ by considering $\rotnumberAvg(\cdiffeo{\phi})$.
Similarly, we define the \emph{conformal radius} of a complex deformation as that of the univalent function $\cunivalent{\phi}$ (cf.~Proposition~\ref{prop:decomp_def}), denoting it by
\begin{equation}
	\CR(\cunivalent{\phi}) = \frac{1}{(\cunivalent{\phi})'(0)} \, \in \, \Rp, \qquad \phi \in \DefC.
\end{equation}
Observe that the normalization of $\cunivalent{\phi}$ makes the conformal radius a positive real number.
From the decomposition, we see that if $\psi \in \Diffpan$, then $\CR(\phi \circ \psi) = \CR(\phi)$.
A scaling transformation $\shrink{\tau} \in \shrinkgroup$, as listed in Section~\ref{section:subgroups}, has conformal radius $\CR(\shrink{\tau}) = e^{2\pi \tau}$.

Finally, we combine the mean rotation and conformal radius to define 
\begin{equation}
	\rotCR(\phi) = \frac{e^{\ii \rotnumberAvg(\cdiffeo{\phi})}}{\CR(\cunivalent{\phi})} \, \in \, \C \setminus \{0\}, \qquad \phi \in \DefC,
	\label{eq:RCR_def}
\end{equation}
which is a Fr-smooth complex-valued function on $\DefC$.
We have chosen the definition of $\rotCR$ so that for a scaling transformation and rotation, we have 
\begin{equation}
\rotCR(\shrink{\tau} \circ \rotation{\alpha}) = e^{-2 \pi \tau + \ii \alpha}, 
\end{equation}
which is the image of $1$ under this complex deformation.
We would like to define a complex-valued $2$-cocycle $\rotcocycle(\phi_1, \phi_2)$ on composable pairs $\phi_1, \phi_2 \in \DefC$ by taking the group-level differential~\eqref{eq:dgrp} of the logarithm of $\rotCR$.
To obtain a continuous cocycle, we need to apply the logarithm to~\eqref{eq:rotnumberAvg} with carefully chosen lifts (see Proposition~\ref{prop:rotnumber_lifts}).
Since the choice of lift of $\phi_1 \circ \phi_2 $ depends on both $\phi_1$ and $\phi_2$, this results in a non-trivial cocycle.
Here, the logarithm is essential because the branch cut makes $\log \rotCR$ discontinuous on $\DefC$, whence $\rotcocycle$ is not a coboundary.

The univalent function of a complex deformation $\phi \in \DefC$ satisfies $\cunivalent{\phi}(0) = 0$, so it restricts to a map from the punctured closed disk to the punctured plane. 
This restriction may be lifted to a univalent function on the closed upper half-plane $\cuhp = \setsuchthatinline{z \in \C}{{\Re z} \geq 0}$,
\begin{equation}
	\cunivalentlift{\phi} \colon \cuhp \to \C , \qquad 
	\cunivalent{\phi}(e^{\ii z}) = e^{\ii \cunivalentlift{\phi}(z)} , \qquad  z \in \cuhp .
\end{equation}
The set of univalent functions $F \colon \odisk \to \C$ such that $F(0) = 0$ and $F'(0) > 0$ with real-analytic boundary behavior is simply connected, since each $F$ is determined by the real-analytic loop $F(S^1)$ with winding number $+1$ around $0$, and the set of such real-analytic loops contracts to $S^1$.
Thus, we may choose lifts $\hat{F} \colon \cuhp \to \C$ such that $F(e^{\ii z}) = e^{\ii \hat{F}(z)}$ for $z \in \cuhp$ in a consistent way for all $F$, that is, $\hat{F}$ depends continuously (even Fr-smoothly) on $F$, and $\id$ lifts to $\id$.
This normalization is such that for $\phi \in \Diffpan$, we have $\cunivalentlift{\phi} = \id$.

\begin{proposition}
Let $\cdiffeolift{\phi_1}, \cdiffeolift{\phi_2} \colon \R \to \R$ be lifts of the diffeomorphisms $\cdiffeo{\phi_1}$, $\cdiffeo{\phi_2}$ defined by a pair $(\phi_1, \phi_2) \in \defmultset$ of composable complex deformations.
Let 
\begin{equation}
	\cdiffeolift{\phi_1 \circ \phi_2} = \cunivalentlift{\phi_1 \circ \phi_2}^{-1} \circ \cunivalentlift{\phi_1} \circ \cdiffeolift{\phi_1} \circ \cunivalentlift{\phi_2} \circ \cdiffeolift{\phi_2} 
\label{eq:cdiffeolift_composition}
\end{equation}
define a lift of $\cdiffeo{\phi_1 \circ \phi_2}$. 
Then, the difference in mean translation
\begin{equation}
	\rotnumberInt(\cdiffeolift{\phi_1}) 
	+ \rotnumberInt(\cdiffeolift{\phi_2}) 
	- \rotnumberInt(\cdiffeolift{\phi_1 \circ \phi_2}) 
	\label{eq:rotnumberAvg_cocycle}
\end{equation}
is Fr-smooth and independent of the choice of the lifts $\cdiffeolift{\phi_1}$ and $\cdiffeolift{\phi_2}$.
\label{prop:rotnumber_lifts}
\end{proposition}
\begin{proof}
Since the lifts of both the diffeomorphisms and univalent functions are $2\pi$-periodic,
replacing $\cdiffeolift{\phi_1}$ and $\cdiffeolift{\phi_2}$ by $\cdiffeolift{\phi_1} + 2\pi k_1$ and $\cdiffeolift{\phi_2} + 2\pi k_2$ adjusts the lift~\eqref{eq:cdiffeolift_composition} to
\begin{equation}
\begin{aligned}
& \; \cunivalentlift{\phi_1 \circ \phi_2}^{-1} \circ \cunivalentlift{\phi_1} \circ (\cdiffeolift{\phi_1} + 2\pi k_1) \circ \cunivalentlift{\phi_2} \circ (\cdiffeolift{\phi_2} + 2\pi k_2) \\
= \; &\cunivalentlift{\phi_1 \circ \phi_2}^{-1} \circ \cunivalentlift{\phi_1} \circ \cdiffeolift{\phi_1} \circ \cunivalentlift{\phi_2} \circ \cdiffeolift{\phi_2} + 2\pi (k_1 + k_2) , \qquad k_1, k_2 \in \Z .
\end{aligned}
\end{equation}
Hence, the changes of lifts in Equation~\eqref{eq:rotnumberAvg_cocycle} cancel out.
The combination is Fr-smooth since composition preserves Fr-smoothness.
\end{proof}

With the lifts of Proposition~\ref{prop:rotnumber_lifts}, we define the $2$-cocycle on $(\phi_1, \phi_2) \in \defmultset$ as
\begin{equation}
\begin{aligned}
	\rotcocycle(\phi_1, \phi_2) = 
	&\frac{\ii}{12}
	\log \bigg(
		\frac{
		\CR(\cunivalent{\phi_1 \circ \phi_2})
		}{
		\CR(\cunivalent{\phi_1}) \, 
		\CR(\cunivalent{\phi_2})
		}
	\bigg) \\
	&-\frac{1}{12} \Big(
		\rotnumberInt(\cdiffeolift{\phi_1}) 
		+ \rotnumberInt(\cdiffeolift{\phi_2}) 
		- \rotnumberInt(\cdiffeolift{\phi_1 \circ \phi_2}) 
	\Big).
\end{aligned}
\label{eq:def_rotcocycle}
\end{equation}
Ignoring the need for a choice of lifts, we have $\rotcocycle = \frac{\ii}{12} \Dgrp \log \rotCR$, as intended.
In fact, if both $\phi_1$ and $\phi_2$ are close enough to the identity, this formula may be used to compute $\rotcocycle(\phi_1, \phi_2)$.
Since the Lie algebra cocycle $\DVE \rotcocycle$, defined by Equation~\eqref{eq:def_dlie}, only depends on derivatives at the identity, we can express it in terms of the following derivative of $\rotCR$.

\begin{lemma}
The derivative of $\rotCR$ at the identity is $\C$-linear and reads
\begin{equation}
	\DdR_{\id} \rotCR(\ell_n) =
	\eval{\pdv{t}}_{t = 0}
	\rotCR(\Phi_{\ell_n}(t, \blank))
	= - \delta_{n, 0}, \qquad n \in \Z,
	\label{eq:derivative_rot}
\end{equation}
where $\Phi_{\ell_n}$ is the right-trivializing flow~\eqref{eq:flow} of $\ell_n$.
\label{lem:derivative_rot}
\end{lemma}
\begin{proof}
First, consider the solution of the flow equation for positive modes given by
\begin{equation}
	\Phi_{\ell_n}(t, z) = (z^{-n} + nt)^{-\frac{1}{n}}, \qquad n \geq 1 ,
\end{equation}
which has a well-defined branch in a neighborhood of the closed unit disk (indeed, its singularities are at $\infty$ and $(-nt)^{-1/n}$ for $t \in [0, 1)$).
Note that the origin is a fixed point of $\Phi_{\ell_n}$, and $\eval{\partial_z}_{z = 0} \Phi_{\ell_n}(t, z) = 1$.
Thus, in the decomposition of Proposition~\ref{prop:decomp_def} we have $\cunivalent{\Phi_{\ell_n}(t, z)} = \Phi_{\ell_n}(t, z)$, and hence $\cdiffeo{\Phi_{\ell_n}(t, z)} = \id$, making the mean rotation vanish and the conformal radius equal to $1$ for all $t \in [0, 1)$.
We conclude that $\DdR_{\id} \rotCR(\ell_{n}) = 0$ for $n \geq 1$.
Because the derivative $\DdR_{\id} \rotCR$ in our setup is a priori only $\R$-linear, we separately consider $\ii \ell_n$ for $n \geq 1$.
Similarly, we have  $\Phi_{\ii \ell_n}(t, z) = (z^{-n} + \ii nt)^{-\frac{1}{n}}$, and $\DdR_{\id} \rotCR(\ii \ell_{n}) = 0$ for $n \geq 1$.
Concerning the negative modes, observe that the tangential vector fields $\ellpar{n}$ and $\ellpari{n}$ for $n \geq 1$, defined in~(\ref{eq:def_ellpar_ellperp},~\ref{eq:ellpar}), mix the positive and negative modes for $\ell_n$ and $\ii \ell_n$ respectively.
Since they, respectively, lift to the $(2\pi)$-periodic vector fields $\sin(nx) \, \partial_x$ and $\cos(nx) \, \partial_x$, which alternate signs at $(4\pi/n)$-periodicity, the mean rotation of their flows (diffeomorphisms) vanish.
Hence, we find that $\DdR_{\id} \rotCR(\ellpar{n}) = 0 =\DdR_{\id} \rotCR(\ellpari{n})$.
By $\R$-linearity, we conclude that $\DdR_{\id} \rotCR(\ell_n) = 0 = \DdR_{\id} \rotCR(\ii \ell_n)$ for $n \leq -1$.

Lastly, the vector field $\ell_0 = -z \partial_z$ generates a scaling transformation, so
\begin{equation}
	\DdR_{\id} \rotCR(\ell_0)
	= \eval{\pdv{t}}_{t = 0} \rotCR(\shrink{\frac{t}{2\pi}})
	= \eval{\pdv{t}}_{t = 0} e^{-t} = -1,
\end{equation}
and the rotated vector field $\ii \ell_0 = - \ii z \partial_z$ generates a rotation, so
\begin{equation}
	\DdR_{\id} \rotCR(\ii \ell_0)
	= \eval{\pdv{t}}_{t = 0} \rotCR(\rotation{-t}) = \eval{\pdv{t}}_{t = 0} e^{-\ii t} = -\ii.
\end{equation}
This shows the asserted formula~\eqref{eq:derivative_rot}.
\end{proof}

Using Lemma~\ref{lem:derivative_rot}, we can compute the image of $\rotcocycle$ under the van Est map:

\begin{proposition}
In the $\C$-basis $\setsuchthat{\ell_n}{n \in \Z}$ of $\Witt$, we have 
\begin{equation}
	(\DVE \rotcocycle)(\ell_n, \ell_m) = \frac{\ii}{12} n \, \delta_{n + m}, \qquad n, m \in \Z.
\end{equation}
\label{prop:DVE_rotcocycle}
\end{proposition}
\begin{proof}
	For $t$ and $s$ close to $0$ in the van Est map~\eqref{eq:def_dlie}, we have $(\DVE \rotcocycle)(\ell_n, \ell_m) = \frac{\ii}{12} (\DVE \Dgrp \log \rotCR)(\ell_n, \ell_m)$.
	By Equation~\eqref{eq:vE_grp_alg_dif} and Lemma~\ref{lem:derivative_rot}, this equals
	\begin{equation}
	- \frac{\ii}{24}
	\eval{\pdv{t}}_{t = 0}
	\log \rotCR(\Phi_{[\ell_n, \ell_m]}(t, \blank))
	= 
	\frac{\ii}{24} (n - m)
	\delta_{n+m, 0} ,
	\end{equation}
	where we also used the Witt algebra relations and the fact that $\rotCR(\id) = 1$.
\end{proof}

Note that the Lie algebra cocycle in Proposition~\ref{prop:DVE_rotcocycle} is indeed the part of~\eqref{eq:lie_bott} with the linear dependence in $n$.
It agrees with the Lie algebra cohomology differential of 
\begin{align}
	\rotnumberalg(v) = \; & \frac{1}{48\pi} \int_0^{2\pi} v(\theta) \; \dd\theta, \qquad v \in \VectC, \\ 
	\rotnumberalg(\ell_n) = \; & \frac{\ii }{24} \delta_{n, 0} , \qquad n \in \Z ,
\end{align}
defining the Lie algebra cocycle $\rotcocyclealg = \Dalg \rotnumberalg$,
\begin{equation}
\begin{aligned}
	\rotcocyclealg(v, w) = \; & \frac{1}{24\pi} \int_{0}^{2\pi} v(\theta) w'(\theta) \; \dd \theta, \qquad v,w \in \VectC, \\ 
	\rotcocyclealg(\ell_n, \ell_m) = \; & \frac{\ii}{12} n \, \delta_{n + m, 0} , \qquad n,m \in \Z .
\end{aligned}
\label{eq:def_rotcocyclealg}
\end{equation}
Overall, we obtain the following relations:
\begin{equation}
	\DVE \bott = \gelfandfuks - \rotcocyclealg, \qquad \DVE \rotcocycle = \rotcocyclealg.
\label{eq:cocycle_relations}
\end{equation}

\begin{remark}
	The cohomology class of the real-valued $2$-cocycle $\Re \rotcocycle$ may also be obtained from the transgression $\funtogrp(\gamma)$ in Proposition~\ref{prop:ses_hom_grp_alg} for $\gamma : \pi_1(\DefC) \hookrightarrow \R$ given by identifying a full positive rotation with $1 \in \Z \cong \pi_1(\DefC)$.
	Indeed, restricting to the diffeomorphism group $\Diffpan$ for simplicity, we find that the central extension by $\Re \rotcocycle$ is characterized via a morphism of central extensions from the universal cover of $\Diffpan$, which is a central extension by $\pi_1(\Diffpan) \cong \Z$:
	\begin{equation}
	\begin{tikzcd}[ampersand replacement=\&]
		{\{0\}} \& {\Z} \& {\universalcover{\Diffpan}} \& {\Diffpan} \& {\{0\}} \\
		{\{0\}} \& \R \& {\Diffpan \times_{\rotcocycle} \R} \& \Diffpan \& {\{0\}}.
		\arrow[from=1-1, to=1-2]
		\arrow[from=1-2, to=1-3]
		\arrow[from=1-3, to=1-4]
		\arrow[from=1-4, to=1-5]
		\arrow[from=2-1, to=2-2]
		\arrow[hook, "\gamma", from=1-2, to=2-2]
		\arrow[from=2-2, to=2-3]
		\arrow["{(\Pi, \rotnumberInt)}", from=1-3, to=2-3]
		\arrow[from=2-3, to=2-4]
		\arrow[equals, from=1-4, to=2-4]
		\arrow[from=2-4, to=2-5]
	\end{tikzcd}
	\end{equation}
	Moreover, $\gamma$ is identified with the Euler class of $\Diffpan$, hence the notation $\rotcocycle$; see also~\cite{Ghys:Groups_acting_on_the_circle} and~\cite[Section~4.5]{Guieu-Roger:Virasoro_book}.
\label{remark:euler_class}
\end{remark}

\begin{remark}
Instead of the mean rotation and translation~\eqref{eq:rotnumberAvg}, one may also consider the Poincar\'e rotation number of a diffeomorphism, which describes the average position of a point in $S^1$ under repeated application of the diffeomorphism~\cite[Chapter~11]{Katok-Hasselblatt:Introduction_to_the_modern_theory_of_dynamical_systems}.
It is also defined using a lift of $\phi \in \Diffpan$ to a $2\pi$-periodic diffeomorphism $\hat{\phi} \colon \R \to \R$, such that $\phi(e^{\ii \theta}) = e^{\ii \hat{\phi}(\theta)}$ for $\theta \in \R$.
The rotation number does not depend on the choice of lift.
More precisely, we first define the \emph{translation number}~\cite[Proposition~11.1.1]{Katok-Hasselblatt:Introduction_to_the_modern_theory_of_dynamical_systems} as 
\begin{equation}
	\transnumber(\hat{\phi}) = \lim_{n \to \infty} \frac{1}{n} \hat{\phi}^{\circ n}(0) \, \in \, \R, 
\label{eq:transnumber}
\end{equation} 
where the power ``$\circ n$'' stands for $n$-fold composition. 
The \emph{rotation number} is defined as 
\begin{equation}
	\rotnumber(\phi) = \arg e^{\ii \transnumber(\hat{\phi})} \, \in \, [0,2\pi), \qquad
	\phi \in \Diffpan
	\label{eq:rotnumber}
\end{equation}
(cf.~\cite[Definition~11.1.2]{Katok-Hasselblatt:Introduction_to_the_modern_theory_of_dynamical_systems}).  
In particular, a rotation $\rotation{\alpha}$, as listed in Section~\ref{section:subgroups}, has rotation number $\rotnumber(\rotation{\alpha}) = \alpha$, and if $\phi \in \Diffpan$ has a fixed point, then $\rotnumber(\phi) = 0$.

The rotation number~\eqref{eq:rotnumber} is discontinuous at the branch cut of the argument.
Yet, if we define the lift of the composition $\phi_1 \circ \phi_2$ by $\hat{\phi}_{1, 2} = \hat{\phi}_1 \circ \hat{\phi}_2$, then the difference of translation numbers
\begin{equation}
	-\frac{1}{12} \Big(
		\transnumber(\hat{\phi}_1) 
		+ \transnumber(\hat{\phi}_2) 
		- \transnumber(\hat{\phi}_{1,2}) 
	\Big)
	\label{eq:rotnumber_cocycle}
\end{equation}
is independent of the choice of the lifts $\hat{\phi}_1$ and $\hat{\phi}_2$.
Equation~\eqref{eq:rotnumber_cocycle} defines a continuous $2$-cocycle on $\Diffpan$ that is cohomologous to $\Re \rotcocycle$ by a continuous coboundary~\cite[Section~4.5]{Guieu-Roger:Virasoro_book}.
However, this cocycle is not Fr-smooth:
Consider the classical Arnol'd family
\begin{equation}
	\phi_{a, \varepsilon}(e^{\ii \theta})
	= e^{\ii (\theta + a + \varepsilon \cos(\theta))},
	\qquad 
	a \in \R, \; \varepsilon \in (-1,1) .
	\label{eq:Arnold_family}
\end{equation}
Arnol'd showed in~\cite[Section~12]{Arnold:Small_denominators_I_Mappings_of_the_circumference_onto_itself} (see also \cite[Section~16.2]{Devaney:Introduction_to_chaotic_dynamical_systems})
that for any fixed $\varepsilon \neq 0$, the function $a \mapsto \rotnumber(\phi_{a,\varepsilon})$ is a Cantor function (or Devil's staircase) in the sense that each rational level set $\setsuchthat{a \in [0, 2\pi)}{\rotnumber(\phi_{a,\varepsilon}) = q}$, for $q \in \mathbb{Q} \cap [0, 2\pi)$, is a positive-length closed interval --- this is termed phase-locking at the Arnol'd tongues.
\label{remark:rotation_number}
\end{remark}

\subsection{Computation of the cohomology}
\label{section:computation_cohomology}

\begin{theorem}
	The group-level and algebra-level cohomologies of $\DefC$ relative to $\Diffpan$ and $\shrinkgroup$ are respectively as in Figure~\ref{fig:diag_grp} and Figure~\ref{fig:diag_alg}.
\label{thm:cohomology}
\end{theorem}

\begin{proof}
The asserted Lie algebra cohomologies in Figure~\ref{fig:diag_alg} follow from a variation of a classical computation of the second cohomology of the Witt algebra.
This is well-known for complex-valued vector fields $\VectC$ with complex coefficients, 
or for real-valued vector fields $\VectR$ with real coefficients; see~\cite{Guieu-Roger:Virasoro_book, Schottenloher:Mathematical_introduction_to_CFT} and references therein.
The latter gives $H^2(\VectR; \R)$ in the top right corner of Figure~\ref{fig:diag_alg}. 
However, we also need the second cohomology of $\VectC$ with coefficients in $\R$, which we compute next.

Consider a cocycle $\alpha \in Z^2(\Witt; \R)$, and the coboundary $\Dalg f$ for 
\begin{equation}
\begin{aligned}
	f(\ell_n) = \; & \frac{\alpha(\ell_n, \ell_0)}{n}, \qquad
	f(\ii \ell_n) = \frac{\alpha(\ii \ell_n, \ell_0)}{n}, \qquad n \neq 0 ,  \\
	f(\ell_0) = \; & f(\ii \ell_0) = 0,
\end{aligned}
\label{eq:coboundary_f}
\end{equation}
which is a function $f \in (\Witt)^\vee$ defined by its values in the $\R$-basis $\setsuchthat{\ell_n, \ii \ell_n}{n \in \Z}$ of $\Witt$. 
We expand the cocycle property of $\alpha$ in this basis for $n, m, k \in \Z$,
\begin{alignat}{4}
	(n - m) &\; \alpha(\ell_{n+m}, \ell_k) &&+ (m - k)\; \alpha(\ell_{m+k}, \ell_n) &&+ (k - n) \; \alpha(\ell_{k + n}, \ell_m) \; &= 0,
	\label{eq:cocycle_property_basis_a} \\*
	(n - m) &\; \alpha(\ii \ell_{n+m}, \ell_k) &&+ (m - k) \; \alpha(\ell_{m+k}, \ii \ell_n) &&+ (k - n) \; \alpha(\ii \ell_{k + n}, \ell_m) \;  &= 0, 
	\label{eq:cocycle_property_basis_b} \\*
	- (n - m) &\; \alpha(\ell_{n+m}, \ell_k) &&+ (m - k) \; \alpha(\ii \ell_{m+k}, \ii \ell_n) &&+ (k - n) \; \alpha(\ii \ell_{k + n}, \ii \ell_m) \;  &= 0, 
	\label{eq:cocycle_property_basis_c} \\*
	- (n - m) &\; \alpha(\ell_{n+m}, \ii \ell_k) &&- (m - k) \; \alpha(\ell_{m+k}, \ii \ell_n) &&- (k - n) \; \alpha(\ell_{k + n}, \ii \ell_m) \;  &= 0,
\label{eq:cocycle_property_basis_d}
\end{alignat}
By considering $k = 0$, we find that the cohomologous cocycle $\beta = \alpha - \Dalg f$ satisfies 
\begin{alignat}{3}
	\beta(\ell_n, \ell_m)
	= & \; \alpha(\ell_n, \ell_m)
	&&- \frac{n - m}{n + m} \; \alpha(\ell_{n + m}, \ell_0)
	\,\;\; = \; 0 ,
	\\
	\beta(\ii \ell_n, \ii \ell_m)
	= & \; \alpha(\ii \ell_n, \ii \ell_m)
	&&+ \frac{n - m}{n + m} \; \alpha(\ell_{n + m}, \ell_0)
	\,\;\; = \; 0 ,
	\\
	\beta(\ii \ell_n, l_m)
	= & \; \alpha(\ii \ell_n, \ell_m)
	&&- \frac{n - m}{n + m} \; \alpha(\ii \ell_{n + m}, \ell_0)
\; = \; 0 , \qquad n + m \neq 0 .
\end{alignat}
Thus, $\beta$ is supported on the anti-diagonal $n + m = 0$ for $n, m \in \Z$, where the values are given by some constants $c_n, d_n, e_n \in \R$, $n \in \Z$,
\begin{align}
	\beta(\ell_n, \ell_m) = & \; c_n \, \delta_{n+m,0}, \qquad
	\; c_{-n} = - c_{n}, \\
	\beta(\ii \ell_n, \ii \ell_m) = & \; d_n \, \delta_{n+m,0}, \qquad
	d_{-n}  \,= - d_{n}, \\
	\beta(\ii \ell_n, \ell_m) = & \; e_n \, \delta_{n+m,0}, \qquad
\; e_{-n} = - e_{n},
\end{align}
already determining $c_0 = d_0 = e_0 = 0$.
By considering the cocycle properties~(\ref{eq:cocycle_property_basis_a}--\ref{eq:cocycle_property_basis_d}) for $\beta$ instead of $\alpha$, we find for $n + m + k = 0$ the identities
\begin{align}
(n - m) c_{n + m} + \, (m - k) c_{m + k} + (k - n) c_{k + n} = & \; 0,
\label{eq:cocycle_property_beta_a} \\
(n - m) e_{n + m} + \, (m - k) e_{m + k} + (k - n) e_{k + n} = & \; 0,
\label{eq:cocycle_property_beta_b} \\
-(n - m) c_{n + m} + (m - k) d_{m + k} + (k - n) d_{k + n} = & \; 0.
\label{eq:cocycle_property_beta_c}
\end{align}

Note how the third relation~\eqref{eq:cocycle_property_beta_c} mixes the constants $c_n$ and $d_n$.
By setting $k = n$ and $m = -2n$, we find that $c_n = -d_n$ for all $n \in \Z$. 
By setting $k = 1$ and $m = -(n+1)$ for $n \geq 2$, we infer that $(c_n,e_n)$ are solutions of equal and independent linear recursion relations, with initial conditions $(c_1, c_2) \in \R^2$ and $(e_1, e_2) \in \R^2$: 
\begin{align}
	c_{n + 1} = & \; \frac{1}{n - 1} \Big(
		(n + 2) \; c_n - (2n + 1) \; c_1)
	\Big),
	\\
	e_{n + 1} = & \; \frac{1}{n - 1} \Big(
		(n + 2) \; e_n - (2n + 1) \; e_1)
	\Big).
\end{align}
The initial conditions $(1, 2)$ and $(1, 8)$ correspond to the linearly independent solutions $n$ and $n^3$ respectively.
Since the recursion equations are linear, all solutions are linear combinations of these two.
The solutions $c_n = n$, $e_n = 0$ and $c_n = 0$, $e_n = n$ are respectively proportional to $\Im \rotcocyclealg$ and $\Re \rotcocyclealg$ in Equation $\eqref{eq:def_rotcocyclealg}$.
The other two solutions $c_n = n^3$, $e_n = 0$ and $c_n = 0$, $e_n = n^3$ are proportional to $\Im \gelfandfuks$ and $\Re \gelfandfuks$, the real and imaginary parts of the Gel'fand--Fuks cocycle as in Equation~\eqref{eq:gelfandfuks}.
Note in addition that the couboundary $\Dalg f$ for the function $f$ as in Equation~\eqref{eq:coboundary_f} above vanishes on the subalgebra $\R \ell_0$ --- and thus, it is also a coboundary in the cohomology relative to $\R \ell_0$.

Observe next that $\Im \rotcocyclealg = \Im \Dalg \rotnumberalg$ and $\Re \rotcocyclealg = \Re \Dalg \rotnumberalg$ are coboundaries in $Z^2(\Witt; \R)$, so  
the Lie algebra cohomology of $\DefC$ is spanned by $\Im \gelfandfuks$ and $\Re \gelfandfuks$.
Relative to $\R \ell_0$ however, $\Im \rotcocyclealg$ is not a coboundary anymore since $\Im \rotnumberalg(\ell_0) \neq 0$ by~\eqref{eq:def_rotcocyclealg}.
Thus, the cohomology $H^2(\VectC; \R \ell_0; \R)$ contains it as an additional generator.
The cocycle $\R \gelfandfuks$ does not vanish on $\VectR$, thus, 
$H^2(\VectC; \VectR; \R)$ only contains the cocycle $\Im \gelfandfuks$.
For the cohomology relative to both $\VectR$ and $\R \ell_0$, both of the arguments above hold, and thus it is spanned by $\Im \gelfandfuks$ and $\Im \rotcocyclealg$.

The first Lie algebra cohomology is given by derivations modulo inner derivations. 
On the one hand, for the simple Lie algebra $\VectR$, every derivation is trivial.
On the other hand, for the abelian Lie algebra $\R \ell_0$, any linear function is a derivation, and no derivation is inner.
We may view this one-dimensional space as the span of the function $\Im \rotnumberalg$.
This concludes the characterization of the Lie algebra cohomologies in Figure~\ref{fig:diag_alg}.

We will use Proposition~\ref{prop:ses_hom_grp_alg} in combination with the above established algebra-level cohomology and the homotopy groups, Figure~\ref{fig:diag_hom},
to derive the group-level cohomology, Figure~\ref{fig:diag_grp}. 
To this end, we observe that each of the fundamental groups in Figure~\ref{fig:diag_hom} is either trivial or isomorphic to $\Z$.
In case of the latter, $\Hom(\Z, \R) \cong \R$. This is the case for $\DefC$ and $\Diffpan$ (cf.~\cite[Remark~4.4.3]{Guieu-Roger:Virasoro_book}).
However, the fundamental group of $\DefC$ relative to $\Diffpan$ is trivial again, because $\Diffpan$ contains the subgroup of rotations, which contains representatives of all homotopy classes in $\DefC$.

Lastly, we observe that for every Lie algebra cocycle above, we have found a corresponding group-level cocycle in Section~\ref{section:list_cocycles} that integrates the respective Lie algebra cocycle.
(This proves that the van Est map in Proposition~\ref{prop:ses_hom_grp_alg} is surjective.)
Moreover, Proposition~\ref{prop:ses_hom_grp_alg} yields a short exact sequence between the diagrams
\begin{equation}
	\{0\} \;\longrightarrow\; \textnormal{Figure~\ref{fig:diag_hom}} \;\longrightarrow\; \textnormal{Figure~\ref{fig:diag_grp}} \;\longrightarrow\; \textnormal{Figure~\ref{fig:diag_alg}} \;\longrightarrow\; \{0\},
\end{equation}
and each of the four braided chain complexes in each of the figures is exact, by the respective long exact sequences of relative homotopy groups, group-level cohomology, and Lie algebra cohomology.
The additional group-level cocycle coming from the homotopy group via Proposition~\ref{prop:ses_hom_grp_alg} is $\Re \rotcocycle$, see~\eqref{eq:def_rotcocycle}, which indeed differentiates to a trivial Lie algebra cocycle.
Exactness of all sequences shows that we have found all group-level cohomology classes. 
\end{proof}

\subsection{Applications to conformal field theory}
\label{section:cft}

From the conformal anomaly of CFT, and a fixed central charge $\charge \in \R$, one may obtain a non-trivial central extension of $\DefC$ through the construction of a real determinant line bundle. 
We refer to~\cite{Friedan-Shenker:The_analytic_geometry_of_two-dimensional_conformal_field_theory, 
Friedrich:On_connections_of_CFT_and_SLE, Kontsevich-Suhov:On_Malliavin_measures_SLE_and_CFT, 
Dubedat:SLE_and_Virasoro_representations_localization, 
Benoist-Dubedat:SLE2_loop_measure, Maibach-Peltola:From_the_conformal_anomaly_to_the_Virasoro_algebra} for the construction of the real determinant line bundle, and to~\cite{Segal:Definition_of_CFT_collection, Huang:2D_Conformal_geometry_and_VOAs, Maibach-Peltola:From_the_conformal_anomaly_to_the_Virasoro_algebra} 
for the construction of the central extension.
By~\cite[Theorem~1.1]{Maibach-Peltola:From_the_conformal_anomaly_to_the_Virasoro_algebra}, 
this central extension corresponds to a cocycle~$\Gamma_\charge \in Z^2(\DefC, \R)$ whose Lie algebra cocycle is
\begin{equation}
	\gamma_\charge = \DVE \Gamma_\charge = \charge \Im \gelfandfuks \in Z^2(\VectC; \R),
\label{eq:previous_main}
\end{equation}
which is $\charge$ times the imaginary part of the Gel'fand--Fuks cocycle~\eqref{eq:gelfandfuks}.
As explained in~\cite[Remark~1.2]{Maibach-Peltola:From_the_conformal_anomaly_to_the_Virasoro_algebra}, the group-level cocycle $\Gamma_\charge$ vanishes on $\Diffpan$, making it an element of the relative cohomology which we compute in Theorem~\ref{thm:cohomology}:
\begin{equation}
	[\Gamma_\charge] \in H^2(\DefC; \Diffpan; \R) \cong \R \Im \bott.
\end{equation}
Thus, $\Gamma_\charge$ is cohomologous to $\charge \Im \bott$, the imaginary part of the Bott--Thurston cocycle~\eqref{eq:bott_thurston}, where the proportionality constant $\charge$ is determined by Equation~\eqref{eq:previous_main}.
The difference $\Gamma_\charge - \charge \Im \bott$ is a coboundary which vanishes on $\Diffpan$.
To resolve this difference more precisely, we will next show that the cocycle $\Gamma_\charge$ is also relative to the group $\scalinggroup$ of scaling transformations, as defined in Section~\ref{section:subgroups}.

\begin{proposition}
The cocycle $\Gamma_\charge$ vanishes on the group $\scalinggroup$ of scaling transformations.
Moreover, in the cohomology relative to both $\Diffpan$ and $\scalinggroup$, we have
\begin{equation}
	[\Gamma^\charge] = [\charge \Im (\bott + \rotcocycle)] \in H^2(\DefC; \Diffpan, \scalinggroup; \R).
\label{eq:gamma_bott_rotcocycle}
\end{equation}
\label{prop:gamma_bott_rotcocycle}
\end{proposition}
\begin{proof}
Let $g_\tau$ denote the euclidean metric on the cylinder $C_\tau = S^1 \times [0, \tau]$.
Using the notation in~\cite{Maibach-Peltola:From_the_conformal_anomaly_to_the_Virasoro_algebra}, let $[g_\tau]$ denote the element in the real determinant line of $C_\tau$.
The trivialization $\hat{\mu}^\charge$ of the central extension of $\DefC$, defined with respect to the unit length cylinder $\ann = C_1$, is given by~\cite[Equation~(3.21)]{Maibach-Peltola:From_the_conformal_anomaly_to_the_Virasoro_algebra}.
For a scaling transformation $\scaling{\tau} \in \DefC$ it reduces to $\hat{\mu}^\charge(\scaling{\tau}) = [g_{1 + \tau}] \otimes [g_1]^\vee$.
We now make the following geometric observation:
Gluing (more precisely, sewing, as defined in Section~\ref{section:surfaces}) the cylinders $C_{\tau_1}$ and $C_{\tau_2}$ yields a cylinder $C_{\tau_1 + \tau_2}$, and the euclidean metric $g_{\tau_1 + \tau_2}$ is obtained just by combining the metrics $g_{\tau_1}$ and $g_{\tau_2}$ into $g_{\tau_1} \cup g_{\tau_2}$.
This implies that the product $\hat{\mu}^\charge(\scaling{\tau_1}) \diffActing{} \hat{\mu}^\charge(\scaling{\tau_2})$ in the central extension is just given by $\hat{\mu}^\charge(\scaling{\tau_1 + \tau_2})$.
Since the cocycle of interest is defined by $\hat{\mu}^\charge(\scaling{\tau_1}) \diffActing{} \hat{\mu}^\charge(\scaling{\tau_2}) = e^{\Gamma_c(\scaling{\tau_1}, \scaling{\tau_2})} \hat{\mu}^\charge(\scaling{\tau_1 + \tau_2})$, we find that indeed, $\Gamma_c(\scaling{\tau_1}, \scaling{\tau_2}) = 0$.

By Figure~\ref{fig:diag_grp}, the second group cohomology of $\DefC$ relative to both $\Diffpan$ and $\scalinggroup$ is generated by two cocycles: $\Im \bott$ and $\Im \rotcocycle$.
In this case, the group-level cohomology is isomorphic to the Lie algebra cohomology in Figure~\ref{fig:diag_alg}.
Note that Equation~\eqref{eq:previous_main} holds not just up to coboundaries.
Using the relations~\eqref{eq:cocycle_relations}, we have
\begin{equation}
	\DVE(\Gamma_\charge - \charge \Im \bott) = \gamma_\charge - (\charge \Im \gelfandfuks - \charge \Im \rotcocyclealg)  = \charge \Im \rotcocyclealg,
\end{equation}
and we thus obtain the desired Equation~\eqref{eq:gamma_bott_rotcocycle} on the group level.
\end{proof}

\section{The Segal moduli space}
\label{section:surfaces}

In this section, we study Riemann surfaces with enumerated and analytically parametrized boundary components, and the geometric structure of their moduli spaces $\moduligb$.
We will term these ``Segal moduli spaces'' to highlight the importance of the boundary parametrizations and their central role in Segal's axioms of CFT.
They also admit natural algebraic structure given by the sewing (gluing) and unraveling (cutting) operations, as well as actions of $\DefC$, where a complex deformation either deforms a boundary component, or deforms the surface along a simple analytically parametrized curve in the interior.

We discuss the algebra of $\moduligb$ in Sections~\ref{subsec:Moduli_spaces}--\ref{subsec:sewing}.
Then, in Section~\ref{subsec:deformations_MS} we endow $\moduligb$ with a geometry given by the Fr{\"o}licher structure arising from the actions of $\DefC$, in the sense of Appendix~\ref{section:frolicher}.
The main result of this section is the Virasoro uniformization, Theorem~\ref{thm:virasoro_uniformization} in Section~\ref{section:moduligb_fr}, which asserts that $\DefC$-actions just on boundary components generate in fact all of the tangent vectors of $\moduligb$ (in the sense of~\eqref{eq:curvesgen_tangent}).
  
In Section~\ref{section:hyperbolic}, we identify a subspace of (bordered) hyperbolic surfaces in $\moduligb$ with a finite-dimensional moduli space $\fmoduligbzi$ of hyperbolic surfaces with one marked point on each (geodesic) boundary component.
By showing that the Fenchel--Nielsen coordinate functions on $\fmoduligbzi$ are also smooth on $\moduligb$ (Theorem~\ref{thm:fn_smooth}), we confirm that the newly defined Fr{\"o}licher structure is non-trivial (Corollary~\ref{corollary:fr_nontrivial}).
We also briefly discuss the relationship of the complex deformations with earthquakes, grafting, and Schiffer variation.

\subsection{Analytically parametrized boundary components}
\label{subsec:Moduli_spaces}

Let us first set up terminology and notation for the moduli spaces at hand.

\begin{definition}
A \emph{Riemann surface with enumerated and analytically parametrized boundary components} is a compact connected Riemann surface $\Sigma$ with $\boundaries$~boundary components enumerated $\partial_1 \Sigma, \ldots, \partial_{\boundaries} \Sigma$, and parametrized by real-analytic maps $\zeta_j \colon S^1 \to \partial_j \Sigma$ with negative orientation.
We denote such surfaces by tuples $\param{\Sigma, \zeta_1, \ldots, \zeta_{\boundaries}} = \param{\Sigma,\underline{\zeta}}$.
\label{def:surface}
\end{definition}

If the boundary parametrizations are clear from context, we abbreviate $\param{\Sigma,\underline{\zeta}}$ to just $\Sigma$, and if it is clear what surface the parametrizations $\underline{\zeta} = (\zeta_1, \ldots, \zeta_{\boundaries})$ bound, we occasionally denote it by $\param{\blank, \zeta_1, \ldots, \zeta_{\boundaries}}$.
For example, we define the basic surfaces
\begin{align}
	\disk &= \param{\setsuchthatinline{z \in \C}{|z| \leq 1},\quad \inversion} &&\in \disks,
	\label{eq:unit_disk} \\
	\ann_\tau
	&= \param{\setsuchthatinline{z \in \C}{e^{-2\pi \tau} \leq |z| \leq 1},\quad \inversion,\quad \scaling{\tau}}
	&&\in \annuli, \qquad \tau > 0,
\label{eq:standard_annulus}
\end{align}
using the inversion $\inversion$ from Equation~\eqref{eq:inversion} and scaling $\scaling{\tau}$ from Equation~\eqref{eq:def_shrinkgroup}.

\begin{definition}
An \emph{isomorphism} of Riemann surfaces $\param{\Sigma_1, \zeta_1, \ldots, \zeta_{\boundaries_1}}$ and $\param{\Sigma_2, \xi_1, \ldots, \xi_{\boundaries_2}}$ with enumerated and analytically parametrized boundary components is a biholomorphism $F \colon \Sigma_1 \to \Sigma_2$ such that $F \circ \zeta_j = \xi_j$ for all $1 \leq  j \leq  \boundaries_1 = \boundaries_2$.
\label{def:surface_iso}
\end{definition}

\begin{remark}
Note that by the identity theorem on Riemann surfaces, there are no automorphisms of such surfaces with $\boundaries \geq 1$, except for the identity~\cite[Theorem~1.11]{Forster:Lectures_on_Riemann_surfaces}.
\end{remark}

\begin{definition}
For each $\genus,\boundaries \geq 0$, we define the \emph{Segal moduli spaces} by
\begin{equation}
	\moduligb
	=
	\left\{ \textnormal{\; \parbox{0.65 \textwidth}{\raggedright
	genus $\genus$ Riemann surfaces with $\boundaries$ enumerated and analyt- ically parametrized boundary components (Definition~\ref{def:surface})
	}\; } \right\}_{\substack{/ \; \textnormal{isom.} \\ \; \textnormal{(Def.~\ref{def:surface_iso})}}} 
\label{eq:def_moduligb}
\end{equation}
and we denote the isomorphism class of $\param{\Sigma, \underline{\zeta}}$ by $\parameq{\Sigma, \underline{\zeta}}$, or just by $[\Sigma] \in \moduligb$.
\end{definition}

\subsection{Sewing and unraveling}
\label{subsec:sewing}

The main interest in the moduli spaces $\moduligb$ is their algebraic structure given by the sewing operations.
Given two surfaces $\param{\Sigma_1, \underline{\zeta}}$ and $\param{\Sigma_2, \underline{\xi}}$, we obtain a new surface by identifying one boundary component of each surface, labeled by $1 \leq j \leq \boundaries_1$ and $1 \leq k \leq \boundaries_2$, as follows:
\begin{equation}
\begin{aligned}
	&\param{\Sigma_1, \underline{\zeta}} \sew{j}{k} \param{\Sigma_2, \underline{\xi}} 
	\\ =\; &
	\param{(\Sigma_1 \sqcup \Sigma_2) /_\sim, \; \zeta_1, \ldots, \zeta_{j - 1}, \zeta_{j + 1}, \ldots \zeta_{\boundaries_1}, \xi_1, \ldots, \xi_{k - 1}, \xi_{k + 1}, \ldots \xi_{\boundaries_2}},
\end{aligned}
\label{eq:sewing}
\end{equation}
where ``$\sim$'' identifies the boundary components $\partial_j \Sigma_1$ and $\partial_k \Sigma_2$ via $\xi_k \circ \inversion \circ \smash{\zeta^{-1}_j}$.

Because all boundary parametrizations are negatively oriented, the inversion $\inversion$, defined in~\eqref{eq:inversion}, ensures that we are identifying the inside of $S^1$ in the parametrization of $\partial_j \Sigma_1$ to the outside of $S^1$ in the parametrization of $\partial_k \Sigma_2$.
Since isomorphisms preserve the boundary parametrizations, Equation~\eqref{eq:sewing} yields well-defined \emph{sewing operations}
\begin{equation}
\begin{aligned}
	\moduli{\genus_1}{\boundaries_1} \times \moduli{\genus_2}{\boundaries_2} &\longrightarrow \moduli{\genus_1 + \genus_2}{\boundaries_1 + \boundaries_2 - 2} , \\
	([\Sigma_1], [\Sigma_2])
	&\longmapsto
	[\Sigma_1] \sew{j}{k} [\Sigma_2] = [\Sigma_1 \sew{j}{k} \Sigma_2].
	\end{aligned}
\label{eq:sewing_moduligb}
\end{equation}
Analogously, we define a \emph{self-sewing} operation
\begin{equation}
\begin{aligned}
	\moduli{\genus}{\boundaries} &\longrightarrow \moduli{\genus + 1}{\boundaries - 2} , \\
	[\Sigma] 
	&\longmapsto
	\sewself{j}{k} [\Sigma] = [\sewself{j}{k} \Sigma] ,
\end{aligned}
\label{eq:self_sewing_moduligb}
\end{equation}
where $\sewself{j}{k} \Sigma = \param{\Sigma /_\sim, \; \zeta_1, \ldots, \zeta_{j - 1}, \zeta_{j + 1}, \ldots , \zeta_{k - 1}, \zeta_{k + 1}, \ldots \zeta_{\boundaries}}$, and ``$\sim$'' now identifies the boundary components $\partial_j \Sigma$ and $\partial_k \Sigma$ of the same surface via $\zeta_k \circ \inversion \circ \smash{\zeta^{-1}_j}$.

\begin{remark}
By sewing a unit disk $\disk$, defined as the surface~\eqref{eq:unit_disk}, to each of the boundary components of a Riemann surface $\Sigma$, we obtain the capped surface\footnote{
The underline notation in~\eqref{eq:capped} stands for a multiple application of sewing operations.
Sometimes (especially when applying many sewing operations in a sequence) it is more convenient to postpone this relabeling until after the last sewing operation is applied.
(This is valid, as the sewing operation is associative.)
For example, in the $\boundaries = 2$ case, we should have written $(\Sigma \sew{1}{1} \disk) \sew{1}{1} \disk$ instead of $\Sigma \sew{1}{1} \disk \sew{2}{1} \disk$.}
\begin{equation}
	\Sigma \sewall \underline{\disk} = \Sigma \sew{1}{1} \disk \; \cdots \sew{\boundaries}{1} \disk.
\label{eq:capped}
\end{equation}
The choice of negative orientation of the boundary parametrizations $\zeta_j$ (Definition~\ref{def:surface}) implies that these parametrizations 
extend to conformal maps $\zeta_j \colon \disk \to \Sigma \sewall \underline{\disk}$ mapping into the respective cap by the identity map, for $1 \leq j \leq \boundaries$.
Remembering these conformal maps allows one to recover the surface $\Sigma$ with boundary by ``cutting out'' their images.
See~\cite{Radnell-Schippers:Fiber_structure_and_local_coordinates_for_the_Teichmuller_space_of_a_bordered_Riemann_surface, 
RSS:Quasiconformal_Teichmuller_theory_as_analytical_foundation_for_CFT} for more details on this correspondence.
\end{remark}

We let $\ring{\Sigma}$ denote the interior of the surface $\Sigma$.
Given a separating simple real-analytic loop $\vartheta \colon S^1 \to \ring{\Sigma}$, we define the \emph{unraveling} (or cutting) operation along $\vartheta$ in the following manner. 
Denote by $\unravelplus{\vartheta} \Sigma$ and $\unravelminus{\vartheta} \Sigma$ the surfaces which are the closures of the connected components of $\Sigma \setminus \vartheta(S^1)$ in $\Sigma$ such that the seams parametrized by $\vartheta$ are respectively positively and negatively oriented.
Both come with analytically parametrized boundary components, by carrying over the corresponding boundary parametrizations from $\Sigma$ and relabeling them in ascending order.
The new boundary components receive the next available index, and are parametrized respectively by $\vartheta \circ \inversion$ and $\vartheta$ so that both are negatively oriented.
Let $a_1 < \dots < a_{j - 1}$ be the labels the boundary components which end up in $\unravelplus{\vartheta} \Sigma$, and $b_1 < \dots < b_{k-1}$ those in $\unravelminus{\vartheta} \Sigma$.
Then, the unraveled surfaces are given by
\begin{equation}
\begin{cases}
	\unravelplus{\vartheta} \Sigma = \param{\unravelplus{\vartheta} \Sigma, \; \zeta_{a_1}, \ldots, \zeta_{a_{j - 1}}, \vartheta \circ \inversion}, \\
	\unravelminus{\vartheta} \Sigma = \param{\unravelminus{\vartheta} \Sigma, \; \zeta_{b_1}, \ldots, \zeta_{b_{k - 1}}, \vartheta}.
\end{cases}
\label{eq:unravel_separating}
\end{equation}
Note that the unraveling operation is constructed precisely so that
\begin{equation}
(\unravelplus{\vartheta} \Sigma) \sew{j}{k} (\unravelminus{\vartheta} \Sigma) = \Sigma ,
\end{equation}
where $j$ and $k$ are respectively the labels of the seam.

For a non-separating simple real-analytic loop $\vartheta \colon S^1 \to \ring{\Sigma}$, let $\unravel{\vartheta} \Sigma$ denote the closure of $\Sigma \setminus \vartheta(S^1)$ 
where the boundary curve is doubled (in the sense of prime ends). This is a surface of genus $\genus - 1$ and with $\boundaries + 2$ boundary components, which we parametrize as
\begin{equation}
	\unravel{\vartheta} \Sigma = \param{\unravel{\vartheta} \Sigma, \; \zeta_1, \ldots, \zeta_{\boundaries}, \vartheta \circ \inversion, \vartheta},
\label{eq:unravel_nonseparating}
\end{equation}
so that all parametrizations are negatively oriented.
With this definition, we have
\begin{equation}
	\sewself{\boundaries+1}{\boundaries+2} (\unravel{\vartheta} \Sigma) = \Sigma.
\end{equation}

\subsection{Boundary and interior deformations}
\label{subsec:deformations_MS}

\begin{figure}
\centering
\begin{tikzpicture}
	\def\A{(58.6043, -23.569) circle (1.0cm)}
	\def\B{(57.0171, -24.4545).. controls (56.9777, -24.3013) and (56.952, -24.0512) .. (56.952, -23.7688).. controls (56.952, -23.4911) and (56.9768, -23.2447) .. (57.0151, -23.0908).. controls (56.8665, -23.2207) and (56.7941, -23.3924) .. (56.7997, -23.6985).. controls (56.8057, -24.0289) and (56.8115, -24.274) .. (57.0171, -24.4545) -- cycle(57.1956, -23.0689).. controls (57.1445, -23.0894) and (57.0547, -23.1207) .. (57.0719, -23.1714).. controls (57.0876, -23.218) and (57.1703, -23.1794) .. (57.2195, -23.1763).. controls (57.2122, -23.1364) and (57.2042, -23.1004) .. (57.1956, -23.0689) -- cycle(57.2586, -23.5413).. controls (57.2586, -23.5413) and (57.2586, -23.5413) .. (57.2586, -23.5413).. controls (57.1451, -23.6458) and (57.0621, -23.7734) .. (57.0542, -23.9565).. controls (57.0475, -24.1101) and (57.1091, -24.3236) .. (57.2141, -24.389).. controls (57.245, -24.2348) and (57.2644, -24.014) .. (57.2644, -23.7688).. controls (57.2644, -23.6899) and (57.2624, -23.6137) .. (57.2587, -23.5413) -- cycle}
	\def\C{(57.3072, -24.6322).. controls (57.3624, -24.6325) and (57.4257, -24.6114) .. (57.4062, -24.5653).. controls (57.3752, -24.4923) and (57.2879, -24.435) .. (57.2142, -24.389).. controls (57.2142, -24.389) and (57.2141, -24.389) .. (57.2141, -24.389).. controls (57.1863, -24.5281) and (57.1491, -24.6132) .. (57.1082, -24.6132).. controls (57.1811, -24.6251) and (57.2459, -24.6303) .. (57.3072, -24.6322) -- cycle(57.1082, -24.6132).. controls (57.1283, -24.6132) and (57.1476, -24.5925) .. (57.1653, -24.555).. controls (57.1069, -24.5239) and (57.058, -24.4905) .. (57.0171, -24.4545).. controls (57.0427, -24.5544) and (57.0742, -24.6132) .. (57.1082, -24.6132) -- cycle(57.0151, -23.0908).. controls (57.0587, -23.0526) and (57.1089, -23.0181) .. (57.1657, -22.9848).. controls (57.1516, -22.9546) and (57.1366, -22.934) .. (57.1209, -22.9271).. controls (57.1167, -22.9262) and (57.1125, -22.9252) .. (57.1083, -22.9243).. controls (57.1083, -22.9243) and (57.1082, -22.9243) .. (57.1082, -22.9243).. controls (57.0733, -22.9243) and (57.0411, -22.9862) .. (57.0151, -23.0908) -- cycle(57.1209, -22.9272).. controls (57.1366, -22.9341) and (57.1516, -22.9546) .. (57.1657, -22.9848).. controls (57.1657, -22.9848) and (57.1657, -22.9848) .. (57.1657, -22.9848).. controls (57.1762, -23.0074) and (57.1863, -23.0349) .. (57.1956, -23.0689).. controls (57.1965, -23.0685) and (57.1975, -23.0681) .. (57.1984, -23.0678).. controls (57.2497, -23.047) and (57.3178, -23.0359) .. (57.3365, -22.9734).. controls (57.2655, -22.9591) and (57.1917, -22.9427) .. (57.1209, -22.9272) -- cycle(57.2194, -23.1763).. controls (57.2377, -23.2767) and (57.2514, -23.4017) .. (57.2586, -23.5413).. controls (57.3366, -23.4694) and (57.4015, -23.3973) .. (57.3985, -23.3062).. controls (57.3947, -23.1911) and (57.3016, -23.1712) .. (57.2194, -23.1763) -- cycle}
	\def\D{(59.729, -23.3623) -- (59.6388, -23.3545).. controls (59.5656, -23.3481) and (59.434, -23.3712) .. (59.336, -23.4138).. controls (59.1598, -23.4905) and (58.9959, -23.6077) .. (58.8711, -23.756) -- (58.7914, -23.8506)(59.8713, -23.1291).. controls (59.8655, -23.168) and (59.8446, -23.2057) .. (59.8147, -23.2489).. controls (59.7557, -23.3344) and (59.6815, -23.425) .. (59.6036, -23.4934).. controls (59.5036, -23.5812) and (59.3965, -23.65) .. (59.2796, -23.7085).. controls (59.1709, -23.7628) and (59.047, -23.8133) .. (58.9274, -23.8363).. controls (58.8131, -23.8582) and (58.718, -23.8552) .. (58.6104, -23.8265)}
	\def\EA{(57.1661, -22.9845).. controls (56.916, -23.1313) and (56.7925, -23.3023) .. (56.7997, -23.6985).. controls (56.8069, -24.0948) and (56.8138, -24.3683) .. (57.1653, -24.555)}
	\def\EB{(57.3365, -22.9733).. controls (57.3178, -23.0359) and (57.2497, -23.047) .. (57.1984, -23.0678).. controls (57.1479, -23.0883) and (57.0544, -23.1197) .. (57.0719, -23.1714).. controls (57.0877, -23.218) and (57.1703, -23.1794) .. (57.2194, -23.1763).. controls (57.3016, -23.1712) and (57.3946, -23.1911) .. (57.3984, -23.3062).. controls (57.4014, -23.3973) and (57.3366, -23.4694) .. (57.2586, -23.5413).. controls (57.1451, -23.6458) and (57.0621, -23.7734) .. (57.0542, -23.9565).. controls (57.0475, -24.1101) and (57.1091, -24.3236) .. (57.2142, -24.389).. controls (57.2879, -24.435) and (57.3752, -24.4923) .. (57.4061, -24.5653).. controls (57.426, -24.6121) and (57.3606, -24.6331) .. (57.3049, -24.6322)}
	\def\EC{(57.2643, -23.7688).. controls (57.2643, -24.2351) and (57.1944, -24.6132) .. (57.1082, -24.6132).. controls (57.0219, -24.6132) and (56.952, -24.2351) .. (56.952, -23.7688).. controls (56.952, -23.3024) and (57.0219, -22.9244) .. (57.1082, -22.9244).. controls (57.1944, -22.9244) and (57.2643, -23.3024) .. (57.2643, -23.7688) -- cycle}
	\def\ED{(57.1082, -22.9243).. controls (57.2638, -22.9582) and (57.4357, -22.9979) .. (57.5604, -23.0089).. controls (57.6852, -23.0199) and (57.8333, -22.9978) .. (57.9624, -22.9258).. controls (58.0904, -22.8544) and (58.1982, -22.7621) .. (58.3076, -22.6662).. controls (58.432, -22.5571) and (58.5505, -22.4319) .. (58.6739, -22.3203).. controls (58.9729, -22.0497) and (59.346, -21.811) .. (59.7586, -21.9389).. controls (60.0127, -22.0176) and (60.2149, -22.2653) .. (60.3019, -22.4599).. controls (60.4059, -22.6925) and (60.478, -22.9605) .. (60.484, -23.2066).. controls (60.5089, -23.5315) and (60.4711, -23.8572) .. (60.3621, -24.1645).. controls (60.2385, -24.4989) and (59.9875, -24.8304) .. (59.6248, -24.9189).. controls (59.3383, -24.9888) and (59.0678, -24.867) .. (58.7936, -24.8001).. controls (58.5818, -24.7485) and (58.3689, -24.6876) .. (58.1526, -24.6576).. controls (57.9682, -24.6231) and (57.7826, -24.6279) .. (57.5962, -24.6294).. controls (57.4334, -24.6297) and (57.2986, -24.6443) .. (57.1082, -24.6132)}
	\def\F{(59.2186, -24.7852).. controls (59.338, -24.7063) and (59.367, -24.5811) .. (59.245, -24.3368).. controls (59.1229, -24.0925) and (59.0833, -24.0847) .. (59.0778, -23.8401).. controls (59.0723, -23.5955) and (59.2905, -23.3307) .. (59.3622, -23.1457).. controls (59.434, -22.9607) and (59.4843, -22.7562) .. (59.4585, -22.6219).. controls (59.4328, -22.4875) and (59.3365, -22.3752) .. (59.2107, -22.3185).. controls (59.085, -22.2618) and (59.0187, -22.3532) .. (58.9415, -22.5117).. controls (58.8643, -22.6701) and (58.7548, -22.9039) .. (58.554, -22.9205).. controls (58.3531, -22.9371) and (58.0965, -22.6967) .. (57.8933, -22.7329).. controls (57.6902, -22.7691) and (57.5561, -22.9153) .. (57.4642, -23.0767).. controls (57.3723, -23.2381) and (57.3342, -23.5842) .. (57.3655, -23.7649).. controls (57.3969, -23.9456) and (57.4162, -24.0491) .. (57.5376, -24.1128).. controls (57.6589, -24.1766) and (57.8368, -24.0473) .. (57.9676, -24.0503).. controls (58.0984, -24.0533) and (58.1957, -24.0552) .. (58.3116, -24.1168).. controls (58.4275, -24.1783) and (58.5343, -24.3318) .. (58.6204, -24.4256).. controls (58.7066, -24.5194) and (58.719, -24.5624) .. (58.8315, -24.6758).. controls (58.9441, -24.7892) and (59.0992, -24.8642) .. (59.2186, -24.7852) -- cycle}

	\useasboundingbox (-7,-2.5) rectangle (7,3.5);
	\node (pic_left) at (-3, 0) {};
    \begin{scope}[shift=(pic_left), scale=1.5]	
		\begin{scope}[shift={(-58.6043, 23.569)}]
			\begin{scope}
			\clip[]
				\A;
			\path[fill=color2, even odd rule, pattern=north west lines, even odd rule, pattern color = color2]
				\A
				\F;
			\end{scope}
			\begin{scope}
			\path[fill=color2, even odd rule, pattern=north east lines, even odd rule, pattern color = color2]
				\A
				\F;
			\end{scope}
			\path[draw, very thick]
				\A;
			\path[draw=color1, very thick]
				\F;
		\end{scope}
        \node[circle, fill, inner sep=1pt, outer sep=0.0cm]
            (zero_left) at (0, 0) {};
		\node[text=color1, anchor=west]
			(img_left) at (0.8, 1.2) {$\phi(S^1)$};
        \node[draw, fill=white, rounded corners=2pt, inner sep=3pt, anchor=west]
			(U_left) at (0.9, -0.75) {$\DefCnbhd{\phi}$};
        \node[circle, fill, inner sep=0.8pt, outer sep=0.0cm]
            (U_pt_left) at (0.5, -1.05) {};
        \node[circle, fill, inner sep=0.8pt, outer sep=0.0cm]
            (U_pt_left_in) at (0.7, -0.5) {};
    \end{scope}

    \node (pic_right) at (4, 0.5) {};
    \begin{scope}[shift=(pic_right), scale=1.5]
		\begin{scope}[shift={(-58.6043, 23.569)}]
			\path[fill=color2, even odd rule, pattern=north west lines, even odd rule, pattern color = color2]
				\B;
			\path[fill=color2, even odd rule, pattern=north east lines, even odd rule, pattern color = color2]
				\B;
			\path[fill=color2, even odd rule, pattern=north west lines, even odd rule, pattern color = color2]
				\C;
			\path[draw=color1, very thick]
				\EA;
			\path[draw=color1, very thick]
				\EB;
			\path[draw, very thick]
				\EC;
			\path[draw, very thick]
				\ED;
			\path[draw, very thick]
				\D;
		\end{scope}
		\node[text=color1, anchor=west]
			(img_right) at (-1.2, 0.3) {$\zeta_1\big(\phi(S^1)\big)$};
    \end{scope}

    \path[->] (-0.5, 0) edge[above] node {$\zeta_1$} (0.5, 0);
    \path[draw] (U_pt_left) -- (U_left);
    \path[draw] (U_pt_left_in) -- (U_left);
    
\end{tikzpicture}
\caption{
	In this illustration, a complex deformation $\phi \in \DefC$ acts on a genus one surface $\param{\Sigma, \zeta_1}$ by deformation of the single boundary component parametrized by $\zeta_1$.
	This grows the surface along parts of the boundary, and shrinks it elsewhere.
}
\label{fig:boundary}
\end{figure}

Complex deformations of the unit circle, as defined in Section~\ref{section:complex_deformations}, act on the moduli spaces $\moduligb$ in $\boundaries$ ways by deformation of one boundary component, or internally by deforming along a simple real-analytic loop.
Just like composition and inversion in $\DefC$, these actions are only partially defined.
Let us first define them in a pointwise manner.

\begin{definition}
\label{def:action}
Fix $1 \leq j \leq \boundaries$. 
A complex deformation $\phi \in \DefC$ acts on a surface $\param{\Sigma, \underline{\zeta}}$
by deformation of the $j$th boundary component if $\zeta_j \colon S^1 \to \Sigma \sew{j}{1} \disk$, as a function into the capped surface, extends univalently to an open neighborhood of the set $\DefCnbhd{\phi}$ in $\C \setminus \{0\}$ (see Figure~\ref{fig:boundary} and Equation~\eqref{eq:def_defcnbhd}).
Then, denoting the connected component of $(\Sigma \sew{j}{1} \disk) \setminus \zeta_1(\phi(S^1))$ containing $0 \in \disk$ by $V$, the deformed surface is defined as
\begin{equation}
	\Sigma \diffActing{j} \phi = \param{
	(\Sigma \sew{j}{1} \disk) \setminus V, \;
	\zeta_1, \ldots, \zeta_{j - 1}, \zeta_j \circ \phi, \zeta_{j + 1}, \ldots, \zeta_\boundaries}.
\label{eq:deformed_surface}
\end{equation}
\end{definition}

Similar to the composition of complex deformations, the requirement of univalence on $\DefCnbhd{\phi}$ guarantees that the action is uniquely defined.
For example, a complex deformation may fail to act on a boundary component either because $\zeta_j$ has a singularity between $\phi(S^1)$ and $S^1$, or because the deformation would cause boundary curves to intersect.
Note that a diffeomorphism $\phi \in \Diffpan$ always acts on a boundary component, since it just reparametrizes it.
Besides changing the boundary parametrization, complex deformations can add or subtract parts of the surface, as shown in Figure~\ref{fig:boundary}.

We will next check that the $\DefC$-action in Definition~\ref{def:action} is a well-defined operation on the Segal moduli spaces $\moduligb$
(Lemma~\ref{lem:action_assoc}), and that it is local and respects the composition operation of complex deformations (Lemma~\ref{lem:action_compat}).

\begin{lemma}
Given an isomorphism $F \colon \Sigma_1 \to \Sigma_2$ of Riemann surfaces with enumerated and analytically parametrized boundary components such that $\phi \in \DefC$ acts on their $j$th boundary components, the surfaces $\Sigma_1 \diffActing{j} \phi$ and $\Sigma_2 \diffActing{j} \phi$ are isomorphic.
\label{lem:action_assoc}
\end{lemma}
\begin{proof}
Since the sewing operation is well-defined on moduli spaces, the capped surfaces $\Sigma_1 \sew{j}{1} \disk$ and $\Sigma_2 \sew{j}{1} \disk$ are isomorphic as well.
The isomorphism is given by a biholomorphic continuation of $F$, which identifies the sets $V_1$ and $V_2$ in Definition~\ref{def:action} associated to $\Sigma_1$ and $\Sigma_2$.
Since $F \circ \zeta_j = \xi_j$, and by uniqueness of the analytic continuation of $\zeta_j$, we have $F \circ \zeta_j \circ \phi = \xi_j \circ \phi$ on $S^1$.
Thus, $\Sigma_1 \diffActing{j} \phi$ and $\Sigma_2 \diffActing{j} \phi$ are isomorphic by Definition~\ref{def:surface_iso}.
\end{proof}

\begin{lemma}
Given a Riemann surface $\Sigma$ with enumerated and analytically parametrized boundary components, complex deformations $\phi_1, \phi_2 \in \DefC$, and $j \neq k$, we have
\begin{equation}
	(\Sigma \diffActing{j} \phi_1) \diffActing{k} \phi_2
	= (\Sigma \diffActing{k} \phi_2) \diffActing{j} \phi_1 
	\qquad \textnormal{and} \qquad
	(\Sigma \diffActing{j} \phi_1) \diffActing{j} \phi_2
	= \Sigma \diffActing{j} (\phi_1 \circ \phi_2),
\label{eq:action_compat}
\end{equation}
if all respective compositions and actions exist.
\label{lem:action_compat}
\end{lemma}
\begin{proof}
If $\phi_1$ and $\phi_2$ act at distinct boundary components and may act in both orders, the action changes $\Sigma$ in disjoint regions, proving the first equality in~\eqref{eq:action_compat}.
For the second equality in~\eqref{eq:action_compat}, it is assumed that $\phi_1$ and $\phi_2$ are composable, that is, $\phi_1$ extends univalently without zeros to $\DefCnbhd{\phi_2}$ (cf.~Definition~\ref{def:action}).
Hence, the proof of Item~\ref{item:defc_inv_comp3} of Proposition~\ref{prop:defc_inv_comp} applies also to $\zeta_j$, $\phi_1$, and $\phi_2$, 
showing that $(\zeta_j \circ \phi_1) \circ \phi_2 = \zeta_j \circ (\phi_1 \circ \phi_2)$, where the compositions are defined using the respective analytic continuations.
\end{proof}

Thanks to Lemmas~\ref{lem:action_assoc}~\&~\ref{lem:action_compat}, we can endow the moduli spaces with $\DefC$-actions
\begin{equation}
\begin{aligned}
	\compdef{\genus}{\boundaries}{j} &\longrightarrow \moduligb \\
	([\Sigma], \phi) &\longmapsto [\Sigma] \diffActing{j} \phi = [\Sigma \diffActing{j} \phi]
\end{aligned}
\label{eq:action}
\end{equation}
where $\compdef{\genus}{\boundaries}{j}$ denotes the set of pairs $([\Sigma], \phi) \in \moduligb \times \DefC$ such that $\phi$ acts on the $j$th boundary component of any representative $\Sigma$ of $[\Sigma]$.
Similar to Proposition~\ref{prop:defc_local}, for each smooth curve $\gamma \in \curves{\DefC}$ rooted at $\gamma(0) = \id$ and for fixed $[\Sigma] \in \moduligb$, there exists $\varepsilon > 0$ such that $([\Sigma], \gamma(t)) \in \compdef{\genus}{\boundaries}{j}$ for all $t \in (-\varepsilon, \varepsilon)$.

We also define actions of complex deformations along simple real-analytic loops in the interior $\ring{\Sigma}$ of the surface by combining the unraveling, boundary deformation, and sewing:

\begin{definition}
	A complex deformation $\phi \in \DefC$ acts on a surface $\param{\Sigma, \underline{\zeta}}$ by deformation along a simple real-analytic loop $\vartheta \colon S^1 \to \ring{\Sigma}$ by
\begin{equation}
	\Sigma \diffActing{\vartheta} \phi =
	\begin{cases}
		(\unravelplus{\vartheta} \Sigma) \sew{j}{k} \big((\unravelminus{\vartheta} \Sigma) \diffActing{k} \phi\big) , & \textnormal{$\vartheta$ separating,} \\
		\sewself{\boundaries + 1}{\boundaries + 2} \big((\unravel{\vartheta} \Sigma) \diffActing{\boundaries + 2} \phi\big) , & \textnormal{$\vartheta$ non-separating.}
	\end{cases}
	\label{eq:action_interior}
\end{equation}
This operation is defined when $([\unravelminus{\vartheta} \Sigma], \phi) \in \compdef{\genus}{k}{k}$ or $([\unravel{\vartheta} \Sigma], \phi) \in \compdef{\genus}{\boundaries + 2}{\boundaries + 2}$, respectively.
\end{definition}

Observe that the unraveling operation~(\ref{eq:unravel_separating},~\ref{eq:unravel_nonseparating}) is not a well-defined operation on the equivalence class $[\Sigma] \in \moduligb$.
However, if $F \colon \Sigma_1 \to \Sigma_2$ is an isomorphism, then we have
\begin{equation}
	[\unravelplus{\vartheta} \Sigma_1] = [\unravelplus{(F \circ \vartheta)} \Sigma_2], \quad
	[\unravelminus{\vartheta} \Sigma_1] = [\unravelminus{(F \circ \vartheta)} \Sigma_2], \quad \textnormal{or} \quad
	[\unravel{\vartheta} \Sigma_1] = [\unravel{(F \circ \vartheta)} \Sigma_2] .
\end{equation}
For the action of a complex deformation $\phi \in \DefC$, this implies that
\begin{equation}
	[\Sigma_1 \diffActing{\vartheta} \phi] = [\Sigma_2 \diffActing{F \circ \vartheta} \phi] \in \moduligb.
\end{equation}

Lastly, we remark that the actions of $\DefC$ have various interactions with the sewing operation (provided that all surfaces exist).
Write $\invinv(\phi) = \inversion \circ \phi^{-1} \circ \inversion$.
For $j \neq k$, we have
\begin{align}
	(\Sigma_1 \diffActing{j} \phi) \sew{k}{l} \Sigma_2
	&= (\Sigma_1 \sew{k}{l} \Sigma_2) \diffActing{j} \phi,
	\\
	(\Sigma_1 \diffActing{j} \phi) \sew{j}{l} \Sigma_2
	&= \Sigma_1 \sew{j}{l} (\Sigma_2 \diffActing{l} \invinv(\phi)),
	\label{eq:sew_deform_move} \\
	\sewself{k}{l} (\Sigma \diffActing{j} \phi)
	&= (\sewself{k}{l} \Sigma) \diffActing{j} \phi,
	\\
	\sewself{j}{l} (\Sigma \diffActing{j} \phi)
	&= \sewself{j}{l} (\Sigma \diffActing{l} \invinv(\phi)).
	\label{eq:sew_deform_move_self}
\end{align}
From Equation~\eqref{eq:sew_deform_move}, it also follows that in the separating case, we have
\begin{equation}
	\Sigma \diffActing{\vartheta} \phi
	= \Big((\unravelplus{\vartheta} \Sigma) \diffActing{j} \invinv(\phi) \Big) \sew{j}{k} (\unravelminus{\vartheta} \Sigma)
	= \Sigma \diffActing{\vartheta \circ \inversion} \invinv(\phi),
	\label{eq:action_interior_plus}
\end{equation}
and in the non-separating case, we have 
\begin{equation}
	\Sigma \diffActing{\vartheta} \phi
	= \sewself{\boundaries + 1}{\boundaries + 2} \Big((\unravel{\vartheta} \Sigma) \diffActing{\boundaries + 1} \invinv(\phi) \Big)
	= \Sigma \diffActing{\vartheta \circ \inversion} \invinv(\phi).
	\label{eq:action_interior_inversion}
\end{equation}
Moreover, if all surfaces exist, the interior action commutes with the boundary action
\begin{equation}
	(\Sigma \diffActing{\vartheta} \phi_1) \diffActing{j} \phi_2 = 
	(\Sigma \diffActing{j} \phi_2) \diffActing{\vartheta} \phi_1.
\end{equation}
For the interior deformation, we also find a statement similar to Proposition~\ref{prop:defc_local}, asserting that for fixed $\Sigma$ and $\vartheta$, and for each smooth curve $\gamma \in \curves{\DefC}$ rooted at $\gamma(0) = \id$, there exists $\varepsilon > 0$ such that $\Sigma \diffActing{\vartheta} \gamma(t)$ exists and Equation~\eqref{eq:action_interior_plus} holds for all $t \in (-\varepsilon, \varepsilon)$.

\subsection{Fr\"olicher structure and Virasoro uniformization}
\label{section:moduligb_fr}

The actions of $\DefC$ on $\moduligb$ by deformations of boundary components and along interior curves induce a Fr{\"o}licher structure on the moduli spaces $\moduligb$ in the  following manner, in the sense of Appendix~\ref{section:frolicher}.
We first define the generating sets $\curvesgen{\moduligb}$ containing the respective curves
rooted at $[\Sigma] \in \moduligb$ with choice of representative $\Sigma$, 
\begin{equation}
	\R \longrightarrow \moduligb, \qquad t \longmapsto [\Sigma \diffActing{\vartheta_1} \gamma_1(t) \cdots \diffActing{\vartheta_n} \gamma_n(t)] \diffActing{1} \eta_1(t) \cdots \diffActing{\boundaries} \eta_{\boundaries}(t),
\label{eq:moduligb_curvesgen}
\end{equation}
on which we act along a finite number of disjoint simple real-analytic curves $\vartheta_1, \ldots, \vartheta_n$ in the interior $\ring{\Sigma}$ by $\gamma_1, \ldots, \gamma_n \in \curves{\DefC}$, and moreover, we deform the boundary components by $\eta_1, \ldots, \eta_\boundaries \in \curves{\DefC}$, assuming that the deformation exists for all $t \in \R$.
The Fr{\"o}licher spaces $\big(\moduligb, \curves{\moduligb}, \functions{\moduligb}\big)$ are then defined by Equations~(\ref{eq:condition_functions},~\ref{eq:F_generated}).
In Section~\ref{section:hyperbolic}, we will show that the Fr{\"o}licher structure is non-trivial (Corollary~\ref{corollary:fr_nontrivial}).

By construction, the actions of $\DefC$ by boundary deformations are Fr-smooth, and by Lemma~\ref{lem:action_compat}, the actions are compatible with the composition operation of $\DefC$.
Moreover, the sewing operations~(\ref{eq:sewing_moduligb},~\ref{eq:self_sewing_moduligb}) are smooth, because a deformation of a boundary component creates a deformation in the interior, also smooth by definition.

The following may be compared to the action of a Lie group on a smooth manifold, where we obtain a homomorphism from the Lie algebra of that group to the Lie algebra of vector fields on the manifold.
Since the Lie algebra of $\DefC$ is identified as $\VectC$ in Proposition~\ref{prop:defc_algebra}, and since the flow $\Phi_v(t, \blank) \in \DefC$ for $v \in \VectC$ acts on any boundary component of any $[\Sigma] \in \moduligb$ if $t$ is small enough (depending on $[\Sigma]$), the Lie algebra homomorphisms induced by deformations of boundary components are
\begin{equation}
\begin{aligned}
	\rho_j \colon \VectC &\longrightarrow \vect(\moduligb) , \qquad 1 \leq j \leq \boundaries , \\
	v &\longmapsto \Big([\Sigma] \mapsto \big[t \mapsto [\Sigma] \diffActing{j} \Phi^v(t, \blank)\big]_\sim\Big),
\end{aligned}
\end{equation}
where $\vect(\moduligb)$ is the space of tangent vectors on $\moduligb$ with Lie bracket as defined in Appendix~\ref{section:frolicher}.
(One can prove this by replicating the proof of~\cite[Theorem~3.12]{Laubinger:Lie_algebra_for_Frolicher_groups}.)

The following result is a variant of what is often called ``Virasoro uniformization'' 
in the literature~\cite{Kontsevich:Virasoro_and_Teichmuller_spaces, 
ACKP:Moduli_spaces_of_curves_and_representation_theory, Beilinson-Schechtman:Determinant_bundles_and_Virasoro_algebra, 
Frenkel-Ben-Zvi:Vertex_Algebras_and_Algebraic_Curves, Dubedat:SLE_and_Virasoro_representations_localization, Gui-Zhang:Analytic_conformal_blocks_of_C2_cofinite_vertex_operator_algebras2}.
Since we have not characterized the Fr{\"o}licher structures induced by the whole $\curvesgen{\moduligb}$, we restrict to tangent vectors which may be represented by initial curves defined by Equation~\eqref{eq:moduligb_curvesgen}.
We denote those tangent spaces by $T^{\mathcal{C}_0} \moduligb$ (see Equation~\eqref{eq:curvesgen_tangent}).
We show that any initial tangent vector, also including interior deformations, may be represented by deformations of boundary components only.

\begin{theorem}
For $\boundaries \geq 1$ and $[\Sigma] \in \moduligb$, the derivative of the simultaneous action of $\boundaries$~copies of $\DefC$ on the $\boundaries$ boundary components of $[\Sigma]$, given by
\begin{equation}
\begin{aligned}
	\rho \colon \VectC^{\times \boundaries} &\longrightarrow T_{[\Sigma]} \moduligb , \\
	(v_1, \ldots, v_\boundaries) &\longmapsto \Big[ t \mapsto [\Sigma] \diffActing{1} \Phi_{v_1}(t, \blank) \diffActing{2} \cdots \diffActing{\boundaries} \Phi_{v_\boundaries}(t, \blank) \Big]_\sim ,
\end{aligned}
\label{eq:tangent_action}
\end{equation}
is surjective onto the space of tangent vectors which may be represented by initial curves.
Moreover, the kernel of $\rho$ comprises pullbacks of holomorphic vector fields on a representative $\Sigma$ of $[\Sigma]$ via the boundary parametrizations, yielding the short exact sequence~\eqref{eq:virasoro_uniformization_ses}\textnormal{:}
\begin{equation}
	\{0\} \; \longrightarrow \; \vect(\Sigma) \; \overset{\underline{\zeta}^*}{\longhookrightarrow} \; \big(\VectC\big)^{\times \boundaries} \; \overset{\rho}{\longtwoheadrightarrow} \; T_{[\Sigma]}^{\mathcal{C}_0} \moduligb \; \longrightarrow \; \{0\} ,
\label{eq:virasoro_uniformization_ses_again}
\end{equation} 
involving the pullback by $\underline{\zeta} = (\zeta_1,\ldots,\zeta_\boundaries)$.
\label{thm:virasoro_uniformization}
\end{theorem}
\begin{proof}
Any initial tangent vector represented by an initial curve, as in~\eqref{eq:moduligb_curvesgen}, comes with a representative $\param{\Sigma, \underline{\zeta}}$ of $[\Sigma]$ relative to which the interior deformations in~\eqref{eq:moduligb_curvesgen} are defined.
Such a tangent vector is a linear combination
\begin{equation}
	\sum_{j = 1}^n \Big[t \mapsto [\Sigma \diffActing{\vartheta_j} \gamma_j(t)]\Big]_\sim
	+
	\sum_{j = 1}^{\boundaries} \Big[t \mapsto [\Sigma] \diffActing{j} \eta_j(t)\Big]_\sim
\end{equation}
of curves generated by each curve of complex deformations acting individually.
Thus, we may restrict to the case of a single interior deformation by $\gamma \in \curves{\DefC}$ at $\vartheta \subset \ring{\Sigma}$.
Moreover, by Proposition~\ref{prop:defc_algebra}, that tangent vector may be represented by the curve
\begin{equation}
	\Big[t \mapsto [\Sigma \diffActing{\vartheta} \gamma(t)]\Big]_\sim
	= \Big[t \mapsto [\Sigma \diffActing{\vartheta} \Phi_v(t, \blank)]\Big]_\sim,
\end{equation}
that is, by the flow of a time-independent vector field $v \in \DefC$ such that $[t \mapsto \gamma(t)]_\sim$ equals $[t \mapsto \Phi_v(t, \blank)]_\sim$ in $T_{\id} \DefC$.

Let $\dot{\Sigma}$ denote the capped surface $\Sigma \sewall \underline{\disk}$ in~\eqref{eq:capped} minus the zeros $0 \in \disk$ of all the caps.
Let $V_\vartheta$ be a small neighborhood of $S^1$ to which both $v$ and $\vartheta$ extend, respectively holomorphically and univalently.
Consider the open cover of $\dot{\Sigma}$ given by $U_\vartheta = \vartheta(V_\vartheta)$ and $U_0 = \dot{\Sigma} \setminus \vartheta(S^1)$.
The intersection $U_0 \cap U_\vartheta$ has two connected components, which we denote by $U_+$ and $U_-$, such that $U_+ \subset \unravelplus{\vartheta} \Sigma$ in the case that $\vartheta$ is separating, and $U_+$ is on the side of $\partial_{\boundaries + 1} \unravel{\vartheta} \Sigma$ in the case that $\vartheta$ is non-separating.
Since there are no triple intersections in the cover, the vector fields $u_+ = 0$ on $U_+$ and $u_- = \vartheta_* v$ on $U_-$ define a cocycle in the sheaf cohomology $H^1(\dot{\Sigma}, T\dot{\Sigma})$ of the capped and punctured surface $\dot{\Sigma}$ with values in the tangent sheaf $T\dot{\Sigma}$.
Since $\boundaries \geq 1$, the Riemann surface $\dot{\Sigma}$ is non-compact, and hence, $H^1(\dot{\Sigma}, T\dot{\Sigma})$ trivial~\cite[Theorems~26.1 and~30.3]{Forster:Lectures_on_Riemann_surfaces}.
Therefore, there exist holomorphic vector fields $u_0$ on $U_0$ and $u_\vartheta$ on $U_\vartheta$ such that\footnote{
	Since $u_0 = u_\vartheta$ on $U_+$, we observe that if $\vartheta$ is a non-separating curve, then $u_0$ and $u_\vartheta$ restricted to $U_-$ are two different analytic continuations of the same vector field on $U_+$.
}
\begin{equation}
\begin{cases}
	u_+ = u_0 - u_\vartheta & \textnormal{on } U_+ , \\
	u_- = u_0 - u_\vartheta & \textnormal{on } U_- .
\end{cases}
\end{equation}
Since both $u_- + u_\vartheta$ and $u_\vartheta$ extend across $\vartheta(S^1)$, also $u_0$ extends across $\vartheta(S^1)$, taking different values depending on the side of $\vartheta$ where it is extended from.
Pulling back these vector fields along $\vartheta$ to $S^1$, we find the decomposition
\begin{equation}
	\Big[t \mapsto [\Sigma \diffActing{\vartheta} \Phi_v(t, \blank)]\Big]_\sim
	= \Big[t \mapsto [\Sigma \diffActing{\vartheta} \Phi_{\vartheta^* u_0}(t, \blank)]\Big]_\sim
	- \Big[t \mapsto [\Sigma \diffActing{\vartheta} \Phi_{\vartheta^* u_\vartheta}(t, \blank)]\Big]_\sim.
\label{eq:decompose_v_u}
\end{equation}

Let us first consider a non-separating $\vartheta$.
Combining~(\ref{eq:unravel_nonseparating},~\ref{eq:deformed_surface},~\ref{eq:action_interior}), we find that
\begin{equation}
\Sigma \diffActing{\vartheta} \Phi_{\vartheta^* u_0}(t, \blank)
= \sewself{\boundaries+1}{\boundaries+2} \param{\blank, \zeta_1, \ldots, \zeta_{\boundaries}, \vartheta \circ \inversion,  \Phi_{u_0}(t, \blank) \circ \vartheta},
\end{equation}
where ``$\blank$'' stands for the subset of the unraveled and capped surface $(\unravel{\vartheta} \Sigma) \sew{\boundaries +1}{1} \disk \sew{\boundaries + 2}{1} \disk$ bounded by the parametrizations.
For small $t$, the inverse flow $\Phi^{-1}_{u_0}(t, \blank)$ is defined on such a subset, mapping it into the unraveled and capped surface.
Hence, we can apply it as an isomorphism in the sense of Definition~\ref{def:surface_iso}, yielding
\begin{equation}
	[\Sigma \diffActing{\vartheta} \Phi_{\vartheta^* u_0}(t, \blank)]
	= \sewself{\boundaries+1}{\boundaries+2} \parameq{\blank, \ldots, \Phi^{-1}_{u_0}(t, \blank) \circ \zeta_{\boundaries}, \Phi^{-1}_{u_0}(t, \blank) \circ \vartheta \circ \inversion, \vartheta}.
\end{equation}
For the tangent vector, we thus have
\begin{equation}
\begin{aligned}
	\Big[t \mapsto [\Sigma \diffActing{\vartheta} \Phi_{\vartheta^* u_0}(t, \blank)]\Big]_\sim
	= - \sum_{j = 1}^\boundaries 
	\; & \Big[t \mapsto [\Sigma \diffActing{j} \Phi_{\zeta_j^* u_0}(t, \blank)]\Big]_\sim
	\\
	-
	\; & \Big[t \mapsto [\Sigma \diffActing{\vartheta \circ \inversion} \Phi_{\inversion^* \vartheta^* u_0}(t, \blank)]\Big]_\sim .
\end{aligned}
\label{eq:compute_inverse_flow_u0}
\end{equation}
Note that the vector field $\inversion^* \vartheta^* u_0$ is pulled back from the side of the curve where $u_0 = u_{\vartheta}$.
Thus, by applying Equation~\eqref{eq:action_interior_inversion}, we see that the last term in~\eqref{eq:compute_inverse_flow_u0} agrees with the negative of the last term in~\eqref{eq:decompose_v_u}.
This implies that
\begin{equation}
	\Big[t \mapsto [\Sigma \diffActing{\vartheta} \Phi_v(t, \blank)]\Big]_\sim
	= - \sum_{j = 1}^\boundaries 
	\Big[t \mapsto [\Sigma \diffActing{j} \Phi_{\zeta_j^* u_0}(t, \blank)]\Big]_\sim,
\label{eq:interior_to_boundary}
\end{equation}
which is indeed a linear combination of deformations of boundary components only.

If $\vartheta$ is separating such that $\zeta_{a_1}, \ldots, \zeta_{a_{j - 1}}$, 
and $\vartheta \circ \inversion$ are the boundary parametrizations of $\unravelplus{\vartheta} \Sigma$, and $\zeta_{b_1}, \ldots, \zeta_{b_{k - 1}}$, and $\vartheta$ are the boundary parametrizations of $\unravelminus{\vartheta} \Sigma$ , we find that
\begin{equation}
	\Sigma \diffActing{\vartheta} \Phi_{\vartheta^* u_0}(t, \blank)
	= \param{\blank, \zeta_1, \ldots, \zeta_{a_{j - 1}}, \vartheta \circ \inversion} \sew{j}{k} \param{\blank, \zeta_1, \ldots, \zeta_{b_{k - 1}}, \Phi_{u_0}(t, \zeta_1, \blank) \circ \vartheta} .
\end{equation}
By applying the inverse flow $\Phi^{-1}_{u_0}(t, \blank)$ as an isomorphism to both sides yields
\begin{equation}
\begin{aligned}
	[\Sigma \diffActing{\vartheta} \Phi_{\vartheta^* u_0}(t, \blank)]
	=  \;& \parameq{\blank, \Phi^{-1}_{u_0}(t, \blank) \circ \zeta_{a_1}, \ldots, \Phi^{-1}_{u_0}(t, \blank) \circ \zeta_{a_{k - 1}}, \Phi^{-1}_{u_0}(t, \blank) \circ \vartheta \circ \inversion}
	\\ \;&\sew{j}{k} \parameq{\blank, \Phi^{-1}_{u_0}(t, \blank) \circ \zeta_{b_1}, \ldots, \Phi^{-1}_{u_0}(t, \blank) \circ \zeta_{b_{j - 1}}, \vartheta} .
\end{aligned}
\end{equation}
The corresponding tangent vector is also given by Equation~\eqref{eq:compute_inverse_flow_u0}, and Equation~\eqref{eq:interior_to_boundary} also holds in the separating case.
This concludes the proof of the surjectivity of $\rho$.

Next, consider the flow $\Phi_v(t, \blank)$ of a vector field $v \in \vect(\Sigma)$, as a conformal map into the capped surface.
It generates an isomorphism of $\Sigma$ and the surfaces given by
\begin{equation}
	\param{\Phi_v(t, \Sigma), \Phi_v(t, \blank) \circ \zeta_1, \ldots, \Phi_v(t, \blank) \circ \zeta_{\boundaries}}
	= \Sigma \diffActing{1} \Phi_{\zeta_1^* v}(t, \blank) \diffActing{2} \cdots \diffActing{\boundaries} \Phi_{\zeta_{\boundaries}^* v}(t, \blank),
\end{equation}
which indeed yields a constant curve in $\moduligb$.
This verifies the injectivity of $\underline{\zeta}^*$.

Lastly, to finish the proof of exactness of the sequence~\eqref{eq:virasoro_uniformization_ses_again}, note that if a curve $t \mapsto \Sigma_t = \Sigma \diffActing{1} \Phi_{v_1}(t, \blank) \diffActing{2} \cdots \diffActing{\boundaries} \Phi_{v_\boundaries}(t, \blank)$ with $v_1, \ldots, v_\boundaries \in \VectC$ represents the trivial tangent vector, 
then there exists a time-dependent isomorphism $F_t \colon \Sigma \to \Sigma_t$ in the sense of~\ref{def:surface_iso}.
In particular, we have $F_t \circ \zeta_j = \zeta_j \circ \Phi_{v_j}(t, \blank)$ by compatibility with the boundary parametrizations, which implies that the flow of $(\zeta_j)_*v_j$ is globally defined on $\Sigma$ and given by $F_t$.
Hence, the vector fields $v_1, \ldots, v_\boundaries$ are pullbacks of the restriction of $v(t, z) = \dot{F}_t(t, F_t^{-1}(z)) \in \vect(\Sigma)$.
This shows that the image $\image\underline{\zeta}^*$ agrees with $\kernel \rho$.
\end{proof}

\subsection{Hyperbolic surfaces with marked boundary points}
\label{section:hyperbolic}

In this section, by proving that the Fenchel--Nielsen coordinates define smooth functions on $\moduligb$, we show that the Fr\"olicher structure on $\moduligb$ is non-trivial (see Theorem~\ref{thm:fn_smooth}).

Recall that if $\genus$ and $\boundaries$ are such that the Euler characteristic $\chi(\Sigma) = 2 - 2 \genus - \boundaries$ is negative, there is a unique conformal metric $\gcc(\Sigma)$ of constant Gaussian curvature $-1$ in the conformal class of the Riemann surface $\Sigma$ such that the boundary components are geodesics (termed ``uniform metric of type I'' in~\cite{OPS:Extremals_of_determinants_of_Laplacians})\footnote{
	Note that $\gcc(\Sigma)$ is not the \emph{complete} hyperbolic metric unless $\boundaries=0$, that is, $\Sigma$ is a closed surface.
}. 
This enables one to consider bordered hyperbolic surfaces:

\begin{definition}
A Riemann surface $\param{\Sigma, \underline{\zeta}}$ with enumerated and analytically parametrized boundary components is \emph{hyperbolic} if $\chi(\Sigma) < 0$ and each boundary parametrization has constant speed $|\partial_\theta \zeta_j(e^{\ii \theta})|_\gcc$ in the hyperbolic metric $\gcc(\Sigma)$.
\end{definition}

If $F \colon \param{\Sigma_1, \underline{\zeta}} \to \param{\Sigma_2, \underline{\xi}}$ is an isomorphism such that $\Sigma_2$ is hyperbolic with metric $\gcc(\Sigma_2)$, then the pullback metric is $F^* \gcc(\Sigma_2) = \gcc(\Sigma_1)$, because both Gaussian curvature and boundary curvature are covariant.
The speed of $\zeta_j$ equals the speed of $F \circ \zeta_j = \xi_j$, which is constant.
Therefore, we see that $\Sigma_1$ is hyperbolic if and only if $\Sigma_2$ is hyperbolic.
In that case, we say that $[\Sigma_1] = [\Sigma_2] \in \moduligb$ is \emph{hyperbolic}.
We denote the subspace of hyperbolic surfaces by $\hmoduligb \subset \moduligb$.
We define a projection by replacing the boundary parametrizations with the hyperbolic ones, 
\begin{align}
\begin{aligned}
	\projhgb \colon \moduligb &\longrightarrow \hmoduligb , \\
	\parameq{\Sigma, \zeta_1, \ldots, \zeta_\boundaries}
	&\longmapsto
	\parameq{\Sigma, \check{\zeta}_1, \ldots, \check{\zeta}_\boundaries},
\end{aligned}
\label{eq:hmoduligb_projection}
\end{align}
where a representative $\param{\Sigma, \zeta_1, \ldots, \zeta_\boundaries}$ 
of $[\Sigma]$ is sent to the same Riemann surface with new boundary parametrizations $\check{\zeta}_j \colon S^1 \to \partial_j \Sigma$ being the unique negatively oriented constant speed parametrizations with respect to the hyperbolic metric $\gcc(\Sigma)$ such that ${\check{\zeta_j}(1) = \zeta_j(1)}$.
This map only depends on the conformal structure of $\Sigma$ and the marked points $\zeta_j(1)$, both preserved by isomorphisms.
Hence, the projection $\projhgb$ in~\eqref{eq:hmoduligb_projection} is well-defined, and we see that an element of $\smash{\hmoduligb}$ is determined by the hyperbolic metric together with one marked point on each boundary component.
Below, we relate $\smash{\hmoduligb}$ to a moduli space $\smash{\fmoduligbzi}$ of such hyperbolic surfaces with marked boundary points.
For $\boundaries \geq 1$, this moduli space is a finite-dimensional smooth manifold with coordinate charts given by the Fenchel--Nielsen coordinates. 

\begin{lemma}
Let $\Sigma$ be a smooth, connected, compact, and oriented surface with smooth Riemannian metric $\gcc(\Sigma)$ of constant Gaussian curvature $-1$ such that the $\boundaries \geq 1$ boundary components are geodesics.
If an orientation-preserving isometry $F \colon \Sigma \to \Sigma$ has a fixed point on the boundary $\partial \Sigma$, then $F = \id_\Sigma$.
\label{lem:fixed_pt_id}
\end{lemma}
\begin{proof}
Let $x \in \partial \Sigma$ be a fixed point of $F$, and let $\partial_x \Sigma$ denote the boundary component containing $x$.
Note that $F$ fixes $\partial_x \Sigma$ setwise.
Let $\zeta \colon S^1 \to \partial_x \Sigma$ be the unique negatively oriented unit-speed parametrization rooted at $\zeta(1) = x$.
Then, since $F^*\gcc(\Sigma) = \gcc(\Sigma)$, the composition $F \circ \zeta$ is a unit-speed parametrization of $\partial_x \Sigma$, too. As $F(\zeta(1)) = F(x) = x$, it is rooted at $x$ as well.
Thus, we have $F \circ \zeta = \zeta$, and $F$ restricts to the identity on~$\partial_x \Sigma$.
The claim now follows from the identity theorem on the Riemann surface $(\Sigma, \gcc(\Sigma))$.
\end{proof}

For concreteness, let us outline the classical construction of the moduli space with boundary marked points~\cite{Ivashkovich-Shevchishin:Holomorphic_structure_on_the_space_of_Riemann_surfaces_with_marked_boundary, Liu:Moduli_of_J-holomorphic_curves_with_Lagrangian_boundary_conditions_and_open_Gromov-Witten_invariants_for_an_S1-equivariant_pair} via Teichm{\"u}ller theory~\cite{Farb-Margalit:Primer_on_mapping_class_groups}.
Let $\param{S, \underline{a}}$ denote a smooth, compact, connected, and oriented surface~$S$ of genus $\genus \geq 0$ with $\boundaries \geq 1$ boundary components, and $\underline{a} = (a_1, \ldots, a_\boundaries)$ one marked point in each boundary component.
A~marked hyperbolic surface modeled on $\param{S, \underline{a}}$ is a tuple $(\Sigma, g, \underline{x}, f)$, 
where $f \colon S \to \Sigma$ is a diffeomorphism from $S$ to a smooth surface $\Sigma$ with hyperbolic (constant Gaussian curvature $-1$) metric $g$ such that the boundary components are geodesics, and $f(a_j) = x_j$ for $1 \leq j \leq \boundaries$ where $\underline{x} = (x_1, \ldots, x_\boundaries)$ consists of one point on each boundary component.
The \emph{Teichm{\"u}ller space} is the set of equivalence classes
\begin{equation}
	\teichm(S, \underline{a}) =
\Big\{\textnormal{marked hyperbolic surfaces modeled on $(S, \underline{a})$}\Big\}_{/ \; \sim},
\end{equation}
where $(\Sigma_1, g_1, \underline{x}, f_1) \sim (\Sigma_2, g_2, \underline{y}, f_2)$ if there exists an isometry $F \colon \Sigma_1 \to \Sigma_2$ such that $F(x_j) = y_j$ for all $j$, and $f_2^{-1} \circ F \circ f_1$ is isotopic to the identity on $S$ relative to the marked points~$\underline{a}$.
Let $\Diff_+(S, \underline{a})$ denote the group of orientation-preserving diffeomorphisms of $S$ fixing each marked point, and $\Diff_0(S, \underline{a})$ the connected component of the identity in $\Diff_+(S, \underline{a})$.
The \emph{mapping class group} is the quotient
\begin{equation}
	\mcg(S, \underline{a}) = \Diff_+(S, \underline{a}) / \Diff_0(S, \underline{a}),
\end{equation}
and it acts\footnote{Sometimes this is called the ``pure'' mapping class group, since it fixes the marked points pointwise.} on the Teichm{\"u}ller space $\teichm(S, \underline{a})$ by
\begin{equation}
\begin{aligned}
		\teichm(S, \underline{a}) \times \mcg(S, \underline{a}) \; & \longrightarrow \teichm(S, \underline{a}), \\ 
		(\parameq{\Sigma, g, \underline{x}, f}, [h]) \; & \longmapsto \parameq{\Sigma, g, \underline{x}, f \circ h}.
\end{aligned}
\label{eq:mcg_action}
\end{equation}
The moduli space\footnote{
	Our notation suggests the more general case $\fmoduligbnm$ of surfaces with $n$ interior marked points, and $m_j$ marked points on the $j$th boundary component, for $\underline{m} = (m_1,\ldots,m_\boundaries)$.
} $\fmoduligbzi$ is the quotient of $\teichm(S, \underline{a})$ by this action of $\mcg(S, \underline{a})$.
The action~\eqref{eq:mcg_action} is smooth and properly discontinuous with respect to the usual smooth structure on the Teichm{\"u}ller space~\cite{Farb-Margalit:Primer_on_mapping_class_groups}.
Generally, the smooth structure on Teichm{\"u}ller space induces a smooth orbifold structure on the moduli space.
However, due to Lemma~\ref{lem:fixed_pt_id} in our setup, the mapping class group acts freely, making the moduli space a smooth manifold
\begin{equation}
	\fmoduligbzi = \teichm(S, \underline{a}) \,/\, \mcg(S, \underline{a}) .
\end{equation}

\begin{figure}
\centering
\begin{tikzpicture}
	\def\A{
		(57.2678, -23.4789).. controls (57.6883, -23.6376) and (58.1632, -23.6526) .. (58.5909, -23.5122)
		(57.695, -23.5882).. controls (57.81, -23.5201) and (57.9092, -23.4896) .. (58.0274, -23.4915).. controls (58.1457, -23.4934) and (58.2567, -23.5028) .. (58.3585, -23.5724)
		(61.0516, -23.1637).. controls (60.765, -23.1057) and (60.2214, -23.1622) .. (60.212, -23.3896).. controls (60.2026, -23.617) and (60.7827, -23.6889) .. (61.0518, -23.6238)
		(55.9761, -23.4364) ellipse (0.1562cm and 0.8444cm)
		(61.0518, -24.4682) ellipse (0.1562cm and 0.8444cm)
		(61.0516, -22.3193) ellipse (0.1562cm and 0.8444cm)
		(61.0518, -25.3126).. controls (60.8177, -25.3029) and (60.6923, -25.048) .. (60.4319, -24.9624).. controls (60.156, -24.8716) and (59.8548, -25.0659) .. (59.5648, -25.108).. controls (59.0539, -25.1823) and (58.5375, -25.3767) .. (58.0274, -25.2967).. controls (57.7996, -25.261) and (57.5963, -25.1265) .. (57.3994, -25.0063).. controls (57.1692, -24.8659) and (57.0017, -24.6313) .. (56.7616, -24.5086).. controls (56.5189, -24.3845) and (56.3816, -24.2307) .. (55.9761, -24.2808)
		(55.9761, -22.5919).. controls (56.3418, -22.6486) and (56.8977, -22.3015) .. (57.1614, -22.1142).. controls (57.4251, -21.9269) and (57.6675, -21.7849) .. (58.1038, -21.8045).. controls (58.5402, -21.8242) and (59.377, -22.1827) .. (59.6252, -22.1977).. controls (59.8733, -22.2127) and (60.1181, -22.1536) .. (60.3526, -21.919).. controls (60.5871, -21.6843) and (60.664, -21.4719) .. (61.0516, -21.4749)
	}
	\def\B{
		(59.2077, -22.0904).. controls (59.2787, -22.1127) and (59.2582, -22.3213) .. (59.1095, -22.4334).. controls (58.9607, -22.5455) and (58.5288, -22.4726) .. (58.3485, -22.4331).. controls (58.0832, -22.3629) and (57.822, -22.2667) .. (57.5453, -22.2526).. controls (57.4045, -22.2368) and (57.2499, -22.2575) .. (57.1449, -22.3615).. controls (56.9584, -22.529) and (56.8274, -22.7503) .. (56.6277, -22.9048).. controls (56.512, -23.0105) and (56.3735, -23.0726) .. (56.1161, -23.0621)
		(56.1156, -23.8389).. controls (56.1911, -23.8046) and (56.2766, -23.7518) .. (56.359, -23.7537).. controls (56.49, -23.7566) and (56.5884, -23.8514) .. (56.6877, -23.933).. controls (56.9363, -24.1372) and (57.1679, -24.4372) .. (57.4425, -24.5954).. controls (57.6567, -24.7188) and (57.8695, -24.796) .. (58.112, -24.7688).. controls (58.3685, -24.74) and (58.9098, -24.6203) .. (59.0989, -24.7372).. controls (59.288, -24.8542) and (59.2869, -25.1379) .. (59.2077, -25.1729)
		(60.9259, -23.9682).. controls (60.7313, -23.9393) and (60.0321, -23.7651) .. (60.0153, -23.3556).. controls (59.9984, -22.946) and (60.5783, -22.6242) .. (60.9049, -22.6086)
		(60.8962, -24.3927).. controls (60.3585, -24.3154) and (59.8782, -24.0442) .. (59.3842, -23.8358).. controls (59.0475, -23.7072) and (58.7087, -23.8861) .. (58.3929, -23.9903).. controls (58.0468, -24.0791) and (57.7513, -24.0197) .. (57.411, -23.9324).. controls (57.2154, -23.8907) and (56.976, -23.7634) .. (56.8815, -23.5694).. controls (56.7953, -23.3924) and (56.9025, -23.1458) .. (57.079, -23.0441).. controls (57.4081, -22.8556) and (57.8072, -22.9096) .. (58.1705, -22.9676).. controls (58.416, -23.0068) and (58.7946, -22.9812) .. (59.0735, -22.8427).. controls (59.3934, -22.6839) and (59.608, -22.7851) .. (60.0491, -22.6354).. controls (60.3249, -22.5419) and (60.5993, -22.3289) .. (60.8954, -22.3193)
	}
	\def\BA{
		(59.2077, -25.1729)arc(90.0001:150.0:0.2844 and -1.5412)arc(150.0:210.0:0.2844 and -1.5412)arc(210.0:269.9999:0.2844 and -1.5412)
	}
	\def\C{
		(58.0274, -23.4915)arc(89.9998:150.0:0.1562 and -0.8444)arc(150.0:210.0:0.1562 and -0.8444)arc(210.0:270.0002:0.1562 and -0.8444)
		(59.7576, -25.0663)arc(90.0006:150.0:0.1088 and -1.4341)arc(150.0:210.0:0.1088 and -1.4341)arc(210.0:269.9994:0.1088 and -1.4341)
		(58.0274, -25.2967)arc(89.9998:150.0:0.1562 and -0.8444)arc(150.0:210.0:0.1562 and -0.8444)arc(210.0:270.0002:0.1562 and -0.8444)
	}
	\def\CA{
		(58.0274, -21.8027)arc(270.0001:360.0:0.1562 and -0.8444)arc(0.0:89.9999:0.1562 and -0.8444)
		(58.0274, -23.6079)arc(270.0001:360.0:0.1562 and -0.8444)arc(0.0:89.9999:0.1562 and -0.844)
		(59.7576, -22.1981)arc(-90.0004:0.0:0.1088 and -1.4341)arc(0.0:90.0004:0.1088 and -1.4341)
	}

    \node (pic_left) at (0, 0) {};
    \begin{scope}[shift=(pic_left), scale=1.5]	
		\begin{scope}[shift={(-58.6043, 23.569)}]
			\path[draw=color2, very thick, densely dashed]
				\C;
			\path[draw=color2, very thick]
				\CA;
			\path[draw=color1, very thick]
				\B;
			\path[draw=color1, very thick, densely dashed]
				\BA;
			\path[draw, very thick]
				\A;
			\node[circle, fill, inner sep=1pt, outer sep=0.0cm, label={[label distance=0.3cm]left:{$x_1$}}]
				(x1) at (56.1161, -23.0621) {};
			\node[circle, fill, inner sep=1pt, outer sep=0.0cm, label={[label distance=0.3cm]right:{$x_3$}}]
				(x2) at (60.9259, -23.9682) {};
			\node[circle, fill, inner sep=1pt, outer sep=0.0cm, label={[label distance=0.4cm]right:{$x_2$}}]
				(x3) at (60.8954, -22.3193) {};
		\end{scope}

    \end{scope}
    
\end{tikzpicture}
\caption{
	Example of a choice of pants decomposition (the loops $f(q_1), f(q_2), f(q_3)$ are depicted in orange) and seams (the arcs $f(p_1), f(p_2), f(p_3)$ are depicted in green) of a hyperbolic surface with one marked point on each boundary component.
}
\label{fig:marking}
\end{figure}
To characterize the smooth structure of $\teichm(S, \underline{a})$, and thus of $\fmoduligbzi$, we define the Fenchel--Nielsen coordinates in our setup.
Let $\underline{q} = (q_1, \ldots, q_{3 \genus - 3 + \boundaries})$ be disjoint smooth oriented loops decomposing $S$ into $2 \genus - 2 + \boundaries$ pairs of pants, 
and let $\underline{p} = (p_1, \ldots, p_n)$ for certain $n \geq 1$ be seams for the pants decomposition, that is, disjoint simple loops, or arcs between boundary components, such that for any pair of pants $P$ in the decomposition of $S$ by $\underline{q}$, the intersection $P \cap \underline{p}$ is a union of three disjoint arcs connecting the boundary components of $P$ pairwise.
We assume that each marked point has one arc in $\underline{p}$ ending there.
(See also Figure~\ref{fig:marking}.)
The \emph{Fenchel--Nielsen coordinates} are defined as follows:
\begin{equation}
\begin{aligned}
	\mathsf{FN}_{\underline{q}, \underline{p}} \colon \teichm(S, \underline{a}) &\longrightarrow
	\Rp^{3 \genus - 3 + \boundaries} 
	\times \R^{3 \genus - 3 + \boundaries} 
	\times \Rp^{\boundaries} 
	\times \R^{\boundaries}, 
	\\
	\parameq{\Sigma, g, \underline{x}, f} &\longmapsto
	(\underline{l}, \underline{\theta}, \underline{L}, \underline{\Theta}),
\end{aligned}
\label{eq:fenchel_nielsen}
\end{equation}
$\underline{l} = (l_1, \ldots, l_{3 \genus - 3 + \boundaries})$, 
$\underline{\theta} = (\theta_1, \ldots, \theta_{3 \genus - 3 + \boundaries})$, 
$\underline{L} = (L_1, \ldots, L_{\boundaries})$, 
$\underline{\Theta} = (\Theta_1, \ldots, \Theta_{\boundaries})$: 
\begin{itemize}
	\item
	$l_j$ is the \emph{length} of the unique geodesic in the homotopy class of~$f(q_j) \subset \Sigma$.
	\item
	$\theta_j$ is the \emph{twist coordinate} corresponding to $l_j$. 
	Namely, suppose that $p_k$ is one of the seams crossing $q_j$, and consider the (one or two) pairs of pants with geodesic boundary in $\Sigma$ containing $f(q_j)$.
	On these pants, the appropriate part of the seam $f(p_k)$ is an arc between boundary components $\beta_1$ and $\beta_2$ (that may agree). 
	This arc is isotopic to a composition of paths comprising the unique geodesic between $\beta_1$ and the geodesic homotopic to $f(q_j)$, then a segment along that geodesic possibly winding around many times with total signed length $\delta \in \R$, and finally, the geodesic to $\beta_2$.
	Then, $\theta_j = \frac{2\pi}{l_j} \delta$.
	It is independent of the choice of $k$~\cite[Section~10.6]{Farb-Margalit:Primer_on_mapping_class_groups}.
	\item
	$L_j$ is the \emph{length} of the boundary component containing the marked point $x_j$.
	\item
	$\Theta_j$ is the \emph{twist coordinate} corresponding to $L_j$.
	Namely, let $p_k$ be the seam at the boundary component containing $x_j$, but not ending at $x_j$, and consider the part of $f(p_k)$ that stays within one pair of pants with geodesic boundary.
	This arc is isotopic to the composition of a geodesic between the boundary components of the pants, and a segment along the boundary up to $x_j$ with signed length $\delta \in \R$.
	Then, $\Theta_j = \frac{2\pi}{L_j} \delta$.
\end{itemize}
These lead to an overall dimension of $6 \genus - 6 + 4 \boundaries$ for the smooth manifold $\teichm(S, \underline{a})$.
Moreover, the transition maps between coordinates with different choices of pants decompositions and seams are real-analytic, see~\cite[Appendix]{Wolpert:On_the_Weil-Petersson_geometry_of_the_moduli_space_of_curves}.

Given reference surfaces $\param{S_1, \underline{a}}$ and $\param{S_2, \underline{b}}$, an orientation-preserving diffeomorphism $h \colon S_1 \to S_2$, such that $h(a_j) = b_j$ for all $j$, induces a diffeomorphism of Teichm{\"u}ller spaces
\begin{equation}
	\teichm(S_1, \underline{a}) \longrightarrow \teichm(S_2, \underline{b}), \qquad
	\parameq{\Sigma, g, \underline{x}, f} \longmapsto \parameq{\Sigma, g, \underline{x}, f \circ h},
\label{eq:change_of_base_point}
\end{equation}
which is known as a change of base point.
In particular, the action of the mapping class group~\eqref{eq:mcg_action} changes the base point relative to the same surface.
Thus, we identify the moduli spaces $\fmoduligbzi$ coming from quotients of Teichm{\"u}ller spaces with various base points through~\eqref{eq:change_of_base_point}.
Under these identifications, the moduli spaces of hyperbolic Riemann surfaces with enumerated and analytically parametrized boundary components are related to those with marked points on the boundary by
\begin{equation}
	\hgbtofgbzi \colon \hmoduligb \longrightarrow \fmoduligbzi, \qquad
	\parameq{\Sigma, \underline{\zeta}} \longmapsto
	\parameq{\Sigma, \gcc(\Sigma), \underline{\zeta}(1), \id_\Sigma},
\label{eq:hmoduligb_fmoduligbzi}
\end{equation}
where $\id_\Sigma$ is the identity map on $\Sigma$, thus representing the right-hand side as the base point in $\teichm(\Sigma, \underline{\zeta}(1))$.
Note that isomorphic representatives on the left-hand side are equivalent on the right-hand side through the change of base point provided by the isomorphism.
The map is bijective: a well-defined inverse is obtained by finding for any representative $(\Sigma, g, \underline{x}, f)$ the negatively oriented parametrizations $\zeta_j$ rooted at the marked points $x_j$  which have constant speed with respect to $g$. 
This does not depend on the marking~$f$.

Given a hyperbolic surface $\param{\Sigma, \underline{\zeta}}$ representing $[\Sigma] \in \hmoduligb$, and a marking $(\underline{q}, \underline{p})$ on the basepoint of $\teichm(\Sigma, \underline{\zeta}(1))$, the Fenchel--Nielsen coordinates are defined in a neighborhood $W$ of $\hgbtofgbzi(\Sigma)$ in $\smash{\fmoduligbzi}$.
Hence, the composition $\mathsf{FN}_{\underline{q}, \underline{p}} \circ \hgbtofgbzi$ is well-defined on the subset $\smash{\hgbtofgbzi^{-1}(W)}$ of $\smash{\hmoduligb}$.
Moreover, by precomposing the projection $\projhgb$ from~\eqref{eq:hmoduligb_projection}, this enables the study of Fenchel--Nielsen coordinates of initial curves $\curvesgen{\hmoduligb}$ as defined in Equation~\eqref{eq:moduligb_curvesgen}.

\begin{theorem}
Consider the deformation of a Riemann surface $\param{\Sigma, \underline{\zeta}}$ with enumerated and analytically parametrized boundary components as in Equation~\eqref{eq:moduligb_curvesgen}, that is,
\begin{equation}
\Sigma_t = \Sigma \diffActing{\vartheta_1} \gamma_1(t) \cdots \diffActing{\vartheta_n} \gamma_n(t) \diffActing{1} \eta_1(t) \cdots \diffActing{\boundaries} \eta_{\boundaries}(t).
\end{equation}
There exists $\varepsilon > 0$ such that the following composition is smooth:
\label{thm:fn_smooth}
\begin{equation}
	t \longmapsto \big(\mathsf{FN}_{\underline{q}, \underline{p}} \circ \hgbtofgbzi \circ \projhgb\big)([\Sigma_t]), \qquad t \in (-\varepsilon, \varepsilon).
\label{eq:composition_FN}
\end{equation}
\end{theorem}

The restriction to $t \in (-\varepsilon, \varepsilon)$ is only because the Fenchel--Nielsen coordinates are not globally defined.
However, we will be able to find $\varepsilon > 0$ such that the twists stay small.

\begin{proof}
There exist disjoint open neighborhoods $V_{\vartheta_k}$ of the curves $\vartheta_k(S^1)$, for $k = 1, \ldots, n$, and $V_j$ of the boundary components $\partial_j \Sigma$, for $j = 1, \ldots, \boundaries$, in the capped surface $\Sigma \sewall \underline{\disk}$, and open neighborhoods $U_{\vartheta_k}$ and $U_j$ of $S^1$ in $\C \setminus \{0\}$ such that we have the biholomophisms
\begin{equation}
\zeta_j \colon U_j \to V_j
\qquad \textnormal{and} \qquad 
\vartheta_k \colon U_{\vartheta_k} \to V_{\vartheta_k}.
\end{equation}
Shrinking the neighborhoods $U_{\vartheta_k}$, $V_{\vartheta_k}$, $U_j$, and $V_j$ further, there exists $\varepsilon > 0$ such that for $t \in (-\varepsilon, \varepsilon)$, all of the complex deformations $\gamma_k$ and $\eta_j$ extend univalently and without zeros to their respective neighborhoods $U_{\vartheta_k}$ and $U_j$, and such that they map $S^1$ into that same neighborhood.
(This is analogous to the proof of Proposition~\ref{prop:defc_local}.)

Consider the open cover of $\Sigma \sewall \underline{\disk}$ given by $\setsuchthat{V_{\vartheta_k}, V_j}{1 \leq k \leq n, \; 1 \leq j \leq \boundaries}$ together with a (possibly disconnected) open subset $V_0 \subset \Sigma \sewall \underline{\disk}$ intersecting each of $V_{\vartheta_k}$ and $V_j$ in a disjoint union of two annuli, such that for all $t \in (-\varepsilon, \varepsilon)$, all of the curves $\vartheta_k(\gamma_k(t, S^1))$ and $\zeta_j(\eta_j(t, S^1))$ are disjoint from~$V_0$.
Denote the two connected components of $V_j \cap V_0$ by $V_{j}^{+}$ and $V_{j}^{-}$, such that $V_{j}^{+}$ lies in the cap.
Likewise, denote the two connected components of $V_{\vartheta_k} \cap V_0$ by $V_{\vartheta_k}^{+}$ and $V_{\vartheta_k}^{-}$, such that $V_{\vartheta_k}^{+}$ is on the side of $\unravelplus{\vartheta} \Sigma$. 
The complex manifold $\Sigma_t \sewall \underline{\disk}$ may also be defined by $(\bigsqcup_{k} V_{\vartheta_k}) \sqcup (\bigsqcup_{j} V_j) \sqcup V_0$ under identification via the $t$-dependent conformal maps mapping from the respective intersections to $V_0$, where defined,
\begin{equation}
	h_{\vartheta_k}^{-}(t) = \vartheta_k \circ \gamma_k(t) \circ \vartheta_k^{-1}, \qquad
	h_{j}^{-}(t) = \zeta_j \circ \eta_j(t) \circ \zeta_j^{-1}, \qquad
	h_{\vartheta_k}^{+}(t) = \id .
\label{eq:deformation_transition_maps}
\end{equation}
By unraveling the disks bounded by $\zeta_j \circ \eta_j(t)$, for $j = 1, \ldots, \boundaries$, this defines the complex manifold structure of $\Sigma_t$ as a deformation of that on $\Sigma$ in the sense of~\cite{Kodaira:Complex_manifolds_and_deformation_of_complex_structures}.

Denote by $\gcc = \gcc(\Sigma)$ the hyperbolic metric with geodesic boundary on $\Sigma$.
Then, the restrictions $\restrict{\gcc}{V_0}$, $\restrict{\gcc}{V_{\vartheta_k}}$, and $\restrict{\gcc}{V_j}$ are conformal also on $\Sigma_t$.
However, this does not yield a globally defined conformal metric on~$\Sigma_t$, since the identifications~\eqref{eq:deformation_transition_maps} 
result in conformal factors $|h_{j}^{-}(t)|^2$, $|h_{\vartheta_k}^{-}(t)|^2$, and $|h_{\vartheta_k}^{+}(t)|^2$.
A global conformal metric is obtained by interpolating the conformal factors using smooth bump functions on the intersections, which are independent of $t$.
Denote this metric by $g_t$.
Note that all the restrictions $\restrict{\gcc}{V_0}$, $\restrict{\gcc}{V_{\vartheta_k}}$, and $\restrict{\gcc}{V_j}$ are smooth in $t$.
Observe that the metrics $g_t$ are generally not uniform, i.e., $g_t \neq \gcc(\Sigma_t)$.

By~\cite[Theorem~2.3]{Kodaira:Complex_manifolds_and_deformation_of_complex_structures}, we may define diffeomorphisms $f_t \colon \Sigma \to \Sigma_t$ by interpolating the identity map on the intersections of the open cover on $\Sigma$ so that the restriction of $f_t^{-1}$ to each of $V_{j}^{\pm}, V_{\vartheta_k}^{\pm}$ is smooth in $t$.  
Since $f_t$ equal the identity map away from the intersections of the cover, they identify the marked points, $x_j(t) = f_t(\zeta_j(1)) = \zeta_j(1) \in V_j$, so they define markings on $\Sigma_t$.
We have thus lifted $\Sigma_t$ to a curve in the Teichm{\"u}ller space,
\begin{equation}
	t \longmapsto \parameq{\Sigma_t, \gcc(\Sigma_t), \underline{x}(t), f_t} \in \teichm(\Sigma, \underline{\zeta}(1)).
\end{equation}
It is not yet clear whether $\gcc(\Sigma_t)$ depends smoothly on $t$, since the uniformization problem depends globally on $\Sigma_t$.
The pullback $f_t^* g_t$, however, depends smoothly on $t$ as a metric on $\Sigma$.
Finding the uniform metrics on $\Sigma$ within the conformal classes of $f_t^* g_t$ yields
$f^* \gcc(\Sigma_t)$, which now depends smoothly on $t$ since the uniformization problem is solved on the fixed manifold $\Sigma$. 
Consider a homotopy class of loops in $\Sigma$, and denote by $l(t)$ the length of the unique hyperbolic geodesic in it with respect to $f^* \gcc(\Sigma_t)$.
Solving the geodesic equation yields a smoothly $t$-dependent geodesic, so the length $l(t)$ of the geodesic also depends smoothly on $t$.
We conclude that the geodesic length functions in the Fenchel--Nielsen coordinates depend smoothly on $t$.
The same follows for the twist coordinates, as intersection points of geodesics and their distances depend smoothly on the uniform metric. 
This shows that the composition~\eqref{eq:composition_FN} is smooth, as claimed.
\end{proof}

\begin{example}
Let $\vartheta \colon S^1 \to \ring{\Sigma}$ be the unit-speed parametrization of a closed hyperbolic geodesic in $\ring{\Sigma}$, 
and let $j$ be the index corresponding to the length and twist at $\vartheta$ in a marking involving $\vartheta$ as one of the loops $q_j$.
The deformation at $\vartheta$ by rotations $\rotation{\alpha}$ changes the twist linearly and keeps all other Fenchel--Nielsen coordinates constant: 
\begin{equation}
	\qquad \theta_j(\hgbtofgbzi(\Sigma \diffActing{\vartheta} \rotation{\alpha})) = \theta_j(\hgbtofgbzi(\Sigma)) + \alpha,
	\qquad \alpha \in \R .
\label{eq:act_rotate}
\end{equation}
Scaling transformations $\scaling{\tau}$ define a curve $\hgbtofgbzi(\Sigma \diffActing{\vartheta} \scaling{\tau})$ for $\tau \in (-\varepsilon, \infty)$ for some $\varepsilon > 0$ such that $\Sigma$ still contains the annulus of modulus $\varepsilon$ at $\vartheta$.
This deformation inserts a flat cylinder at $\vartheta$, and thus changes the hyperbolic lengths not just of the hyperbolic geodesics homotopic to $\vartheta$.
The deformations by rotations and scalings may respectively be compared to \emph{earthquakes} and the \emph{grafting} operation; see~\cite{McMullen:Complex_earthquakes_and_Teichmuller_theory} and references therein.
\label{example:fn_deformations}
\end{example}

\begin{example}
We may also realize \emph{Schiffer variations} by curves in $\curvesgen{\moduligb}$. Namely, let $\vartheta \colon \cdisk \to \ring{\Sigma}$ be univalent with real-analytic boundary regularity, and consider the simple analytically parametrized loop obtained from the boundary values of $\vartheta$.
Then for $\phi_t(z) = z + \frac{t}{z}$, the curve $\Sigma \diffActing{\vartheta} \phi_t$ realizes the Schiffer variation
--- see~\cite[Section~4.3]{Nag:Complex_analytic_theory_of_Teichmuller_spaces}, and~\cite{Wolpert:Schiffer_variations_and_Abelian_differentials}, where Schiffer variation is related to Kodaira--Spencer deformation theory.
\label{example:Schiffer_variations}
\end{example}

Examples~\ref{example:fn_deformations}~\&~\ref{example:Schiffer_variations} show that any smooth curve in $\fmoduligbzi$ may be realized by complex deformations acting on $\moduligb$ in various ways.
Conversely, Theorem~\ref{thm:fn_smooth} shows that any Fr-smooth curve in $\moduligb$ smoothly changes the Fenchel--Nielsen coordinates.
By Virasoro uniformization (Theorem~\ref{thm:virasoro_uniformization}), the deformation of boundary components by complex deformations provides a new set of coordinates on $\fmoduligbzi$.
Hence, we find that the bijection $\hgbtofgbzi$ is actually a diffeomorphism:

\begin{corollary}
	The bijection $\hgbtofgbzi \colon \hmoduligb \to \fmoduligbzi$ defined in Equation~\eqref{eq:hmoduligb_fmoduligbzi} is an isomorphism of Fr{\"o}licher spaces, which in this case is a diffeomorphism of $(6 \genus - 6 + 4 \boundaries)$-dimensional smooth manifolds.
\end{corollary}

\begin{proof}
Any smooth curve on the right-hand side of~\eqref{eq:fenchel_nielsen} may be realized by simultaneous Schiffer variations and rotations of the unit speed boundary parametrizations.
Thus, the inverse $\hgbtofgbzi^{-1}$ maps smooth curves to initial Fr-smooth curves, and hence it is Fr-smooth by Proposition~\ref{prop:on_generators_curves} in Appendix~\ref{section:frolicher}.
The smoothness of $\hgbtofgbzi$ follows from Theorem~\ref{thm:fn_smooth}.
\end{proof}

Note that Theorem~\ref{thm:fn_smooth} is stronger than this, since it concerns initial curves on $\moduligb$.
Thus, we arrive at our main conclusion from this section:

\begin{corollary}
The Fr{\"o}licher structure generated by $\curvesgen{\moduligb}$ defined in Equation~\eqref{eq:moduligb_curvesgen} is non-trivial,
since Fenchel--Nielsen coordinates define Fr-smooth functions on $\moduligb$.
\label{corollary:fr_nontrivial}
\end{corollary}
\begin{proof}
A small caveat is that the twist coordinates are not globally single-valued, but we may compose them by smooth compactly supported bump functions so that they become single-valued.
Thus, Corollary~\ref{corollary:fr_nontrivial} can be deduced from Theorem~\ref{thm:fn_smooth}.
\end{proof}

\appendix
\section{Fr\"olicher spaces}
\label{section:frolicher}

In this appendix, we summarize aspects of the theory of Fr{\"o}licher spaces that are sufficient to obtain the analytical results in this work.
We refer interested readers to~\cite[Chapter~23]{Kriegl-Michor:Convenient_setting_of_global_analysis} and~\cite{Stacey:Comparative_smootheology} for a more detailed introduction to Fr{\"o}licher spaces and their relation to other geometric structures, focusing only on the essentials in the present article.

Fr{\"o}licher spaces are generalizations of manifolds where, instead of charts, the structure consists of sets of functions from $\R$ into a set $X$, and functions from $X$ into $\R$ satisfying a completeness relation based on $C^\infty(\R, \R)$, the usual set of smooth functions.

\begin{definition}
A \emph{Fr{\"o}licher space} is a set $X$ with a \emph{Fr{\"o}licher structure} $(X, \curves{X}, \functions{X})$ consisting of $X$ together with a set $\curves{X}$ of \emph{curves} and a set $\functions{X}$ of \emph{functions} such that
\label{def:fr}
\begin{align}
	\curves{X} = & \; \setsuchthat{\gamma \colon \R \to X}{f \circ \gamma \in C^\infty(\R, \R)\: \textnormal{ for all } f \in \functions{X}},
	\label{eq:condition_curves}
	\\
	\functions{X} = & \; \setsuchthat{f \colon X \to \R}{f \circ \gamma \in C^\infty(\R, \R)\: \textnormal{ for all } \gamma \in \curves{X}}.
\label{eq:condition_functions}
\end{align}
A map $\varphi \colon X \to Y$ is \emph{Fr{\"o}licher smooth} (Fr-smooth) with respect to Fr{\"o}licher structures $(X, \curves{X}, \functions{X})$ and $(Y, \curves{Y}, \functions{Y})$ if one of the following equivalent conditions holds:
\begin{align}
	\gamma \in \curves{X}  & \quad \implies \quad \varphi\circ \gamma \in \curves{Y}, 
	\label{eq:smoothness_curves} \\
	f \in \functions{Y}  & \quad \implies \quad f \circ \varphi \in \functions{X},
	\label{eq:smoothness_functions} \\
	\begin{cases}
	\gamma \in \curves{X} \\
	f \in \functions{Y} 
	\end{cases}
	& \quad \implies \quad f \circ \varphi \circ \gamma \in \C^\infty(\R, \R).
\end{align}
\end{definition}

A finite-dimensional smooth manifold structure on a set $X$ may be reconstructed from the induced Fr{\"o}licher structure $(X, C^\infty(\R, X), C^\infty(X, \R))$ using Boman's theorem~\cite[Theorem~3.4]{Kriegl-Michor:Convenient_setting_of_global_analysis}.
However, Fr{\"o}licher structures are also suitable for spaces of inhomogeneous or infinite dimension.
They form a categorical framework in which many other types of geometric structures can be compared.
For example, finite-dimensional smooth manifolds, or Fr\'echet manifolds, form full subcategories of Fr{\"o}licher spaces~\cite[Theorem~3.2]{Frolicher:Smooth_structures}.

Generally speaking, differential geometric notions that only involve differentiation have generalizations defined entirely in terms of the Fr{\"o}licher structure.
Below, we define tangent spaces, Lie algebras, and differential forms in this way.
Integration, however, typically needs some form of coordinate charts.
Thus, it is helpful if a given Fr{\"o}licher space $(X, \curves{X}, \functions{X})$ also has a manifold structure, 
such that the smooth curves and functions agree respectively with $\curves{X}$ and $\functions{X}$.
For instance, the complex deformations $X=\DefC$ defined in Section~\ref{section:complex_deformations} 
form a manifold modelled on a convenient vector space in the sense of~\cite{Kriegl-Michor:Convenient_setting_of_global_analysis} --- 
but in the present work we are mainly interested in certain subsets on which we only have the induced Fr{\"o}licher structure. 
We use, however, the surrounding manifold structure to apply the Poincar\'e lemma~\cite[Lemma~33.20]{Kriegl-Michor:Convenient_setting_of_global_analysis} in Section~\ref{section:cohomology}.

An elementary way to define a Fr{\"o}licher structure on a set $X$ is 
to begin with an initial set $\curvesgen{X}$ of curves, which in principle can be any set of maps from $\R$ to $X$.
Then, the Fr{\"o}licher structure \emph{generated by} $\curvesgen{X}$ is given by the functions
\begin{equation}
\functions{X} = \setsuchthat{f \colon X \to \R}{f \circ \gamma \in C^\infty(\R, \R) \: \textnormal{ for all } \gamma \in \curvesgen{X}},
\label{eq:F_generated}
\end{equation}
and the curves $\curves{X}$ are defined by Equation~\eqref{eq:condition_curves}.
Note that if one chooses too many initial curves, $\functions{X}$ may be empty, and consequently, $\curves{X}$ consists of all maps $\R \to X$.
Thus, it is important to check that the set of functions is non-empty 
(see Corollary~\ref{corollary:fr_nontrivial} for an example).
It is useful to know that in order to show that a function is Fr-smooth, it is sufficient to check condition~\eqref{eq:smoothness_curves} on a generating set of curves.

\begin{proposition}
For a map $Q \colon X \to Y$, assume that
\begin{equation}
	\gamma \in \curvesgen{X} \quad \implies \quad Q \circ \gamma \in \curves{Y}.
\end{equation}
Then, $Q$ is Fr-smooth.
\label{prop:on_generators_curves}
\end{proposition}

\begin{proof}
Since $f \circ Q \circ \gamma \in C^\infty(\R, \R)$ for all $f \in \functions{Y}$ and $\gamma \in \curvesgen{X}$, by~\eqref{eq:F_generated} it follows that $f \circ Q \in \functions{X}$. 
Thus, the map $Q$ is Fr{\"o}licher smooth by~\eqref{eq:smoothness_functions} (as per Definition~\ref{def:fr}).
\end{proof}

There are two concepts of tangent space on Fr{\"o}licher spaces that do not necessarily agree.
On the one hand, there is the functional definition via \emph{derivations} on $\functions{X}$.
On the other hand (and this is the more relevant notion in the present work), 
there are the \emph{curvaceous} tangent spaces defined by
\begin{equation}
	T_x X = \setsuchthat{\gamma \in \curves{X}}{\gamma(0) = x}/\sim, \qquad x \in X,
\label{eq:curvaceous_tangent}
\end{equation}
where
\begin{equation}
	\gamma_1 \sim \gamma_2 \quad \iff \quad (f \circ \gamma_1)'(0) = (f \circ \gamma_2)'(0) \quad \textnormal{ for all } f \in \functions{X}.
\end{equation}
The full curvaceous tangent bundle is the disjoint union $\smash{TX = \underset{{x \in X}}{\bigsqcup} T_x X}$.
For a Fr-smooth function $Q \colon X \to Y$, 
the derivative is defined by
\begin{equation}
\begin{aligned}
	\Dder Q \colon TX &\longrightarrow TY, \\
	[\gamma]_\sim &\longmapsto [Q \circ \gamma]_\sim.
\end{aligned}
\label{eq:Dder_def}
\end{equation}
The tangent bundle then comes with a Fr{\"o}licher structure generated by the functions $\Dder f \colon TX \to T\R \cong \R^2$ for $f \in \functions{X}$.
Note that this notion of tangent space does not always yield a vector space.
A curvaceous tangent vector $v = [\gamma]_\sim \in T_x X$ with $\gamma(0)=x$ acts on a function $f \in \functions{X}$ as a derivation via $v f = (f \circ \gamma)'(0)$.
This induces a Lie bracket on $\vect(X)$, the Fr{\"o}licher space of Fr-smooth sections of the tangent bundle $TX \to X$, given by the usual Lie bracket of vector fields in terms of derivations,
\begin{equation}
	[v, w] = vw - wv , \qquad v, w \in \vect(X) .
\end{equation}
The vector field $[v,w]$ might take values in the functional tangent space defined via derivations.
If we do not have a grasp on the full tangent space, like in Theorem~\ref{thm:virasoro_uniformization}, 
we may also consider the subset of tangent vectors represented by initial curves, which we denote 
\begin{equation}
	T_x^{\mathcal{C}_0} X = \setsuchthat{\gamma \in \curvesgen{X}}{\gamma(0) = x} / \sim, \qquad x \in X.
\label{eq:curvesgen_tangent}
\end{equation}
Similar to the curvaceous tangent space~\eqref{eq:curvaceous_tangent}, there is also a cotangent space
\begin{equation}
	T^x X = \functions{X} / \sim_x, \qquad x \in X,
\end{equation}
where the equivalence relation $\sim_x$ is defined by
\begin{equation}
	f \sim_x g \quad \iff \quad
	(f \circ \gamma)'(0) = (g \circ \gamma)'(0) \: \textnormal{ for all } \textnormal{$\gamma \in \curves{X}$ such that $\gamma(0) = x$}.
\end{equation}
There is a canonical pairing of tangent and cotangent spaces
\begin{equation}
	\evaluation([f]_{\sim_x}, [\gamma]_\sim) = (f \circ \gamma)'(0) \in \R, \qquad [\gamma]_\sim \in T_x X,
\label{eq:eval_dual}
\end{equation}
and we let it generate the Fr{\"o}licher structure on the cotangent bundle $\smash{T^\vee X = \underset{{x \in X}}{\bigsqcup} T^x X}$ by requiring the functions 
\begin{equation}
	\evaluation(\blank , v) \colon \functions{TX} \to \functions{X} , \qquad v \in \vect(X) , 
\end{equation}
to be Fr-smooth.
The definition of cotangent space leads to $n$-forms $\omega \in \Omega^n(X)$ as sections of $\bigwedge^n T^\vee X$, where $\bigwedge^0 T^\vee X = \functions{X}$.
We define the exterior derivative $\DdR \omega$ of $\omega \in \Omega^n(X)$ evaluated on vector fields $v_1, \ldots, v_{n+1} \in \vect(X)$ by generalizing the invariant formula~\cite[Proposition~14.32]{Lee:Introduction_to_smooth_manifolds},
\begin{equation}
\begin{aligned}
	(\DdR \omega)(v_1, \ldots, v_{n+1})
	= \; & \sum_{j = 1}^{n+1} (-1)^{j - 1} v_j \, \omega(v_1, \ldots, \hat{v}_{j}, \ldots, v_{n+1}) \\
	\; &+ \sum_{j = 1}^{n+1} \sum_{k = j + 1}^{n + 1} (-1)^{j + k}
	\omega([v_j, v_k], v_1, \ldots, \hat{v}_{j}, \ldots, \hat{v}_k, \ldots, v_{n+1}),
\end{aligned}
\label{eq:invariant_formula}
\end{equation}
where the hat ``$\hat{\blank}$'' stands for the absence of the symbol.

\bibliographystyle{annotate}
\newcommand{\etalchar}[1]{$^{#1}$}

\end{document}